\documentclass[conference]{IEEEtran}

\pdfoutput=1 %uncomment to ensure pdflatex processing (mandatatory e.g. to submit to arXiv)

\usepackage{tikz}
\usepackage{amsmath}
\usepackage{amssymb}
\usepackage{amsthm}
\usepackage{dsfont}
\usepackage{thm-restate}
\usetikzlibrary{arrows,calc,fit,shapes,automata,backgrounds,positioning}
\usepackage{algpseudocode}
\usepackage{algorithm}
\usepackage{pgfplots}
\usepackage{cite}
\usepackage[hidelinks]{hyperref}

\pgfplotsset{width=5cm,compat=1.13}  % instead of 1.19
\usepgfplotslibrary{fillbetween}

\usepackage{xcolor}
\usepackage{stmaryrd}
\usepackage{tikz}
\usetikzlibrary{automata, positioning, arrows, petri, backgrounds}

\newcommand{\N}{\mathbb{N}}

\newcommand{\Z}{\mathbb{Z}}

\newcommand{\Q}{\mathbb{Q}}

\definecolor{niceredbright}{HTML}{bd0310}
\definecolor{nicebluebright}{HTML}{197b9b}
\definecolor{nicered}{HTML}{7f0a13}
\definecolor{niceblue}{HTML}{104354}
\definecolor{nicegreen}{HTML}{217516}
\definecolor{nicepurple}{HTML}{884bab}
\definecolor{nicebg}{HTML}{f6f0e4}
\definecolor{niceredlight}{HTML}{c9888d}
\definecolor{nicebluelight}{HTML}{78a4b8}
\definecolor{nicegreenlight}{HTML}{76de68}
\definecolor{nicepurplelight}{HTML}{bc87db}

\newtheorem{theorem}{Theorem}[section]
\newtheorem{lemma}[theorem]{Lemma}
\newtheorem{definition}[theorem]{Definition}
\newtheorem{proposition}[theorem]{Proposition}
\newtheorem{corollary}[theorem]{Corollary}

\theoremstyle{definition}
\newtheorem{remark}[theorem]{Remark}
\newtheorem{example}[theorem]{Example}

\newcommand{\reducible}{\trianglelefteq\hspace{-1.8mm}\trianglelefteq}
\newcommand{\vect}[1]{\mathbf{#1}}

\newcommand{\vectSet}[1]{\mathbf{#1}}
\newcommand{\vectset}[1]{\vectSet{#1}}

\newcommand{\HybridizationRelation}[0]{\trianglelefteq}
\newcommand{\HybridizationRelationDirected}[0]{\HybridizationRelation_{\text{dir}}}
\newcommand{\RelationClass}[0]{\mathcal{C}}

\newcommand{\UpwardClosureL}[1]{#1 \uparrow_{\vectSet{L}}}

\DeclareMathOperator{\VAS}{\mathcal{V}}
\DeclareMathOperator{\FO}{FO}

\DeclareMathOperator{\Pumps}{Pumps}

\DeclareMathOperator{\dir}{dir}
\DeclareMathOperator{\dirOfRun}{ends}

\DeclareMathOperator{\interior}{int}

\DeclareMathOperator{\target}{tgt}
\DeclareMathOperator{\tgt}{\target}
\DeclareMathOperator{\source}{src}
\DeclareMathOperator{\src}{\source}
\DeclareMathOperator{\effect}{\Delta}
\DeclareMathOperator{\Effect}{\effect}
\DeclareMathOperator{\CharSys}{Char}
\DeclareMathOperator{\HomCharSys}{HomChar}
\DeclareMathOperator{\LocalCharSys}{LocChar}

\DeclareMathOperator{\Parikh}{pk}

\DeclareMathOperator{\Rel}{Rel}

\DeclareMathOperator{\PX}{\vectSet{P}_{\vectSet{X}}}

\newcommand{\ApproximationAlgorithm}[0]{\mathcal{A}_{\RelationClass}}
\newcommand{\qin}[0]{q_{\text{in}}}
\newcommand{\qini}[0]{q_{\text{in,i}}}
\newcommand{\qiniplusone}[0]{q_{\text{in,i+1}}}
\newcommand{\qinj}[0]{q^j_{\text{in}}}
\newcommand{\qfin}[0]{q_{\text{fin}}}
\newcommand{\qfini}[0]{q_{\text{fin,i}}}
\newcommand{\qfinj}[0]{q^j_{\text{fin}}}
\newcommand{\piin}[0]{\pi_{\text{in}}}
\newcommand{\piout}[0]{\pi_{\text{out}}}

\newcommand{\Iin}[0]{I_{\text{in}}}
\newcommand{\Iout}[0]{I_{\text{out}}}
\newcommand{\InfinityNorm}[1]{||#1||_{\infty}}

\DeclareMathOperator{\rank}{rank}

\DeclareMathOperator{\size}{size}
\DeclareMathOperator{\sol}{sol}

\newcommand{\ConsideredModel}[0]{VASSnz}

\newcommand{\COne}[0]{A1}
\newcommand{\CTwo}[0]{A2}
\newcommand{\WeakPOne}[0]{C1}
\newcommand{\WeakPTwo}[0]{C3}
\newcommand{\WeakPThree}[0]{C2}
\newcommand{\WeakPAll}[0]{\WeakPOne-\WeakPTwo}
\newcommand{\POne}[0]{R1}
\newcommand{\PTwo}[0]{R2}

\newcommand{\para}[1]{\vspace{1mm}\textbf{#1}}

\DeclareMathOperator{\Concretizable}{Concrete}

\begin{document}

\title{Reachability and Related Problems in Vector Addition Systems with Nested Zero Tests}

\IEEEoverridecommandlockouts

\author{\IEEEauthorblockN{Roland Guttenberg}
%\thanks{$^a$?}
\IEEEauthorblockA{%Department of Computer Science \\
Technical University of Munich,
Germany % \\ guttenbe@in.tum.de
}
\and
\IEEEauthorblockN{Wojciech Czerwi\'{n}ski$^a$}
\thanks{$^a$~Partially supported by the ERC grant INFSYS, agreement no. 950398.}
\IEEEauthorblockA{%Department of Computer Science \\
University of Warsaw, Poland % \\ wczerwin@mimuw.edu.pl
}
\and
\IEEEauthorblockN{S{\l}awomir Lasota$^b$}
\thanks{$^b$~Partially supported by NCN grant 2021/41/B/ST6/00535 and by  ERC Starting grant INFSYS, agreement no.~950398.}
\IEEEauthorblockA{%Department of Computer Science \\
University of Warsaw, Poland %\\ sl@mimuw.edu.pl
}
}

\maketitle

\begin{abstract}
Vector addition systems with states (VASS), also known as Petri nets, are a popular model of concurrent systems. Many problems from many areas reduce to the reachability problem for VASS, which consists of deciding whether a target configuration of a VASS is reachable from a given initial configuration. In this paper, we obtain an Ackermannian (primitive-recursive in fixed dimension) upper bound for the reachability problem in VASS with nested zero tests. Furthermore, we provide a uniform approach which also allows to decide most related problems, for example semilinearity and separability, in the same complexity. For some of these problems like semilinearity the complexity was unknown even for plain VASS.
\end{abstract}

\begin{IEEEkeywords}
Vector Addition Systems, Extended Vector Addition Systems, Reachability, Separability, Semilinearity
\end{IEEEkeywords}

\section{Introduction}\label{SectionIntroduction}

% !TEX root = Main.tex

Vector addition systems (VAS), also known as Petri nets, are a popular model of concurrent systems. VAS have a very rich theory and have been intensely studied. In particular, the \emph{reachability problem} for VAS, which consists of deciding whether a target configuration of a VAS is reachable from a given initial configuration, has been studied for over 50 years. It was proved decidable in the early 1980s \cite{SacerdoteT77, Mayr81,Kosaraju82, Lambert92}, but its complexity (Ackermann-complete) could only be determined recently \cite{LerouxS19, CzerwinskiLLLM19, CzerwinskiO21, Leroux21}. 

In \cite{Leroux09} and \cite{Leroux13}, Leroux proved two fundamental results about the reachability sets of VAS. 
In \cite{Leroux09}, he showed that every configuration outside the reachability set $\vect{R}$ of a VAS is separated from $\vect{R}$ by a semilinear inductive invariant (for basic facts on semilinear sets see e.g. \cite{Haase18}). This immediately led to a very simple algorithm for the reachability problem consisting of two semi-algorithms, one enumerating all possible paths to certify reachability, and one enumerating all semilinear sets and checking if they are separating inductive invariants.
In \cite{Leroux13}, he showed that semilinear relations contained in the reachability relation are flatable, which immediately led to an algorithm for checking whether a semilinear relation is included in or equal to the reachability relation. Both of these results were obtained using abstract \emph{geometric properties} of reachability relations. This strand of research was later continued in \cite{GuttenbergRE23}, where some additional geometric axioms were used to reprove that the semilinearity problem, i.e. deciding whether the reachability relation of a given VAS is semilinear, is decidable. The semilinearity problem was first shown to be decidable in \cite{Hauschildt90}.

One major branch of ongoing research in the theory of VAS studies whether results like the above extend to more general systems \cite{Reinhardt08, Bonnet11, Bonnet12, RosaVelardoF11, AtigG11, LerouxPS14, LerouxST15, HofmanLLLST16, LazicS16, FinkelLS18, LerouxS20, BlondinL23}. In particular, in a famous but very technical paper, Reinhardt proved that the reachability problem is decidable for VASS with nested zero tests (VASSnz) \cite{Reinhardt08}, in which counters can be tested for zero in a restricted manner: There is an order on the counters such that whenever counter \(i\) is tested for \(0\), also all counters \(j \leq i\) are tested for \(0\). Later \cite{AtigG11} proved that reachability in VASS controlled by finite-index grammars is decidable by reducing to VASSnz. In \cite{Bonnet11, Bonnet12}, Bonnet presented a more accessible proof of VASSnz reachability by extending the result of \cite{Leroux09}, separability by inductive semilinear sets. Recently, in \cite{Guttenberg24}, also the result of \cite{Leroux13} was extended to VASSnz.

%Another major tool which is utilized in the theory of VASS are \emph{well-quasi-orders} (wqo) \cite{Kruskal72, FinkelG09, FigueiraFSS11, Bonnet12}. A partial order \(\leq\) on a set \(\vectSet{X}\) is a wqo if \(\leq\) is well-founded and every subset \(\vectSet{U} \subseteq \vectSet{X}\) has finitely many minimal elements. In addition to the standard example of \(\N^d\) with the usual component-wise order being a wqo, for most known VASS extensions, all the way to Pushdown VASS and recently amalgamation systems \cite{Jancar90, LerouxPSS19, AnandSSZ24}, even the set of runs admits such an ordering. Whenever a set \(\vectSet{X}\), in this case the set of runs, is well-quasi-ordered, then its subsets have an \emph{ideal decomposition}, which can be utilized to decide the reachability problem \cite{LerouxS15}.

\para{Our contribution}. The following questions remain for \ConsideredModel: The complexity of the reachability and related problems as well as the decidability status of the semilinearity problem, which we both resolve in this paper. We establish the complexity of reachability, semilinearity and related problems as being Ackermann-complete. The complexity of semilinearity was unknown even for plain VAS. %To the best of our knowledge we are hence the first paper to provide a complexity upper bound on the reachability problem in an extension of VASS, since neither Reinhardt \cite{Reinhardt08}, nor Atig \cite{AtigG11}, Bonnet \cite{Bonnet12} or Guttenberg \cite{Guttenberg24} provide any bounds.%We are therefore the first paper giving an upper bound on reachability in an extension of VASS, since neither Reinhart

To this end, we proceed as follows: In step 1/Lemma \ref{LemmaConvertVASSnzToEVASS} we convert the input VASSnz into what we call monotone \(\RelationClass\)-extended VASS. This is an extension of VASS allowing more general counter operations than only addition. As the framework is quite general, we expect that there might be other applications for \(\RelationClass\)-extended VASS. 

In step 2/Theorem \ref{TheoremIdealDecompositionEVASS} we then provide an algorithm which overapproximates the reachability relation of a given input \(\RelationClass\)-extended VASS, assuming its operations can be overapproximated accordingly. In particular the overapproximation allows to decide reachability in \(\RelationClass\)-extended VASS. Our algorithm to perform the overapproximation is an adaptation of the classic algorithm,
called KLM decomposition, solving reachability in VASS.

Finally we solve also problems related to reachability by observing that our overapproximation can be viewed as an improvement upon a framework of \cite{GuttenbergRE23}: For any class \(\RelationClass\) of systems which can be overapproximated, the reachability, semilinearity, separability, etc. problems are decidable in Ackermann time. We consider finalizing this framework started in \cite{GuttenbergRE23} a main contribution of this paper, since it unifies the ideas behind the algorithms solving the reachability, semilinearity, etc. problems in an elegant fashion.

%Other noteworthy contributions include an interesting subclass of semilinear sets we call \emph{directed hybridlinear} sets, as well as progress towards a characterization of which transitions one may allow in a VASS s.t. reachability remains decidable, as Theorem \ref{TheoremIdealDecompositionEVASS} can also be viewed from this direction.

Above we only stated Ackermann-completeness, but in fact we provide more fine-grained complexity. The usual parameters of \ConsideredModel \ are the dimension \(d\) and the number of priorities \(k\), i.e. the number of different \(i\) such that the VASSnz uses a zero test on all \(j \leq i\). Often \(k\) is just \(1\) or \(2\), because only one or two types of zero tests are used. Our algorithm overapproximating the reachability relation has a time bound of \(\mathfrak{F}_{2kd+2k+2d+4}\) in the fast-growing function hierarchy. The complexity is the result of a sequence of \(2k+2\) Turing reductions, each of which will add \(d+1\) to the subscript of the fast-growing complexity class, ending at \(\mathfrak{F}_{2}\) for dealing with semilinear sets. These reductions form a chain as follows, where \(Reach_k\) stands for reachability in VASSnz with \(k\) priorities, and \(Cover_k\) for coverability, that is, deciding whether it is possible to reach a configuration at least as large as the given target from the given initial configuration: 
\[Reach_k \xrightarrow{+d+1} Cover_k \xrightarrow{+d+1} Reach_{k-1} \xrightarrow{+d+1} \dots\]

\textbf{Outline of the paper.} Section \ref{SectionSimplePreliminaries} provides a few preliminaries. Section \ref{SectionVAS} introduces VASS and \ConsideredModel. Section \ref{AlgorithmToolbox} introduces the most important tools we utilize in our algorithm, and in particular the definition of the overapproximation we compute. Section \ref{SectionMainAlgorithm} introduces monotone \(\RelationClass\)-extended VASS and proves that VASSnz can be converted into extended VASS. Section \ref{SectionProofTheoremEVASS} provides an algorithm overapproximating the reachability relation of \(\RelationClass\)-extended VASS. Finally, Section \ref{SectionHybridization} introduces the geometric properties and shows how to solve most relevant problems of \ConsideredModel \ using only the axioms, before we conclude in Section \ref{SectionConclusion}.

%Some proofs are only contained in the full version \cite{FullVersion}.

\section{Preliminaries} \label{SectionSimplePreliminaries}

% !TEX root = Main.tex

We let $\N, \mathbb{Z}, \mathbb{Q}, \mathbb{Q}_{\geq 0}$ denote the sets of natural numbers containing \(0\), the integers, and the (non-negative) rational numbers respectively. 
Let $\mathbb{N}_{\geq n} = \{ m\in\mathbb{N} \mid m \geq n\}$.
We use uppercase letters for sets/relations and boldface for vectors and sets/relations of vectors. 
Given a vector \(\vect{x} \in \Q^n\), we use an array like notation \(\vect{x}[i]\) to refer to the \(i\)-th coordinate. 

Given sets \(\vectSet{X},\vectSet{Y} \subseteq \mathbb{Q}^n, Z \subseteq \mathbb{Q}\), we write \(\vectSet{X}+\vectSet{Y}:=\{\vect{x}+\vect{y} \mid \vect{x} \in \vectSet{X}, \vect{y} \in \vectSet{Y}\}\) for the Minkowski sum and \(Z \cdot \vectSet{X}:=\{\lambda \cdot \vect{x} \mid \lambda \in Z, \vect{x} \in \vect{X}\}\). By identifying elements \(\vect{x}\in \mathbb{Q}^n\) with \(\{\vect{x}\}\), we define \(\vect{x}+\vectSet{X}:=\{\vect{x}\}+\vectSet{X}\), and similarly \(\lambda \cdot \vectSet{X}:=\{\lambda\} \cdot \vectSet{X}\). %for \(\lambda \in \mathbb{Q}\).%We denote by \(\vect{X}^C\) the complement of \(\vect{X}\). 

%We will sometimes perform addition between a vector \(\vect{m} \in \N^d\) and a pair \(\vect{c}=(q, \vect{x}) \in Q \times \N^d\), where \(Q\) is a finite set. We define addition via \(\vect{c}+\vect{m}:=(q, \vect{x}+\vect{m})\).

Given \(\vect{b}=(\vect{b}_s, \vect{b}_t) \in \N^n \times \N^n\) we let \(\Effect(\vect{b}):=\vect{b}_t-\vect{b}_s \in \Z^n\) be the effect of \(\vect{b}\) and extend it to \(\vectSet{F}\subseteq \N^n \times \N^n \) via \(\Effect(\vectSet{F}):=\{\Effect(\vect{b}) \mid \vect{b} \in \vectSet{F}\}\). 

A relation \(\vectSet{R} \subseteq \N^{n_1} \times \N^{n_2}\) is \emph{monotone} if \(n_1=n_2\) and for all \((\vect{x}, \vect{y}) \in \vectSet{R}\) and \(\vect{m} \in \N^n\) we have \((\vect{x}+\vect{m}, \vect{y}+\vect{m}) \in \vectSet{R}\).

%We write \(\DiagD:=\{(\vect{x}, \vect{x}) \mid \vect{x} \in \N^d\}\) for the diagonal, i.e. the minimal monotone relation containing \((\vect{0}, \vect{0})\). We write \(\Diag\) if the dimension is clear from the context. Adding \(\Diag\) to any relation \(\vectSet{R}\) (in the sense of Minkowski sum as above) produces the minimal monotone relation containing \(\vectSet{R}\).

Let \(\mathbb{S} \in \{\Q, \Z, \Q_{\geq 0}, \N, \N_{\geq 1}\}\) and let \(\vectSet{F} \subseteq \Z^n\). The set of \(\mathbb{S}\)-linear combinations of \(\vectSet{F}\) is 
\[\mathbb{S}(\vectSet{F}):=\{ \sum_{i=1}^m \lambda_i \vect{f}_i \mid m \in \N, \vect{f}_i \in \vectSet{F}, \lambda_i \in \mathbb{S}\}.\]
By convention, for \(\vectSet{F}=\emptyset\) we have \(\sum_{x \in \emptyset} x:=\vect{0} \in \mathbb{S}(\vectSet{F})\). 
A set \(\vectSet{X}\) is \(\mathbb{S}\)-(finitely) generated  if there exists a (finite) set \(\vectSet{F}\) with \(\vectSet{X}=\mathbb{S}(\vectSet{F})\). The properties \(\mathbb{S}\)-generated and \(\mathbb{S}\)-finitely generated are respectively abbreviated \(\mathbb{S}\)-g. and \(\mathbb{S}\)-f.g.. 

%The following is well-known:

%\begin{lemma}
%Let \(\vectSet{X} \subseteq \Q^d\) and \(\mathbb{S} \in \{\Q, \Z, \Q_{\geq 0}, \N\}\). Then 
%
%\(\vectSet{X}\) is \(\mathbb{S}\)-generated \(\iff\) \(\vectSet{X}=\mathbb{S}(\vectSet{X})\) \(\iff\) \(\vect{0} \in \vectSet{X}\) and \(\vectSet{X}\) has the following closure properties depending on \(\mathbb{S}\):
%\begin{enumerate}
%\item Case \(\mathbb{S}=\mathbb{N}\): Closure under addition, i.e.\ \(\vectSet{X}+\vectSet{X} \subseteq \vectSet{X}\),
%\item Case \(\mathbb{S}=\mathbb{Z}\): Closure under addition and \(-\vectSet{X} \subseteq \vectSet{X}\),
%\item Case \(\mathbb{S}=\mathbb{Q}_{\geq 0}\): Closure under addition and \(\Q_{\geq 0} \cdot \vectSet{X} \subseteq \vectSet{X}\),
%\item Case \(\mathbb{S}=\Q\): Closure under addition and \(\Q \cdot \vectSet{X} \subseteq \vectSet{X}\).
%\end{enumerate}
%\end{lemma} 

%For \(\mathbb{S} \in \{\mathbb{Z}, \mathbb{Q}\}\) every \(\mathbb{S}\)-g. set is \(\mathbb{S}\)-f.g., while for \(\mathbb{S} \in \{\N, \Q_{\geq 0}\}\) this is not the case. For more details see Section \ref{SectionNewLinearSets}. 

There are different names for \(\mathbb{S}\)-g. sets depending on \(\mathbb{S}\), for example \(\Q\)-g. sets are usually called vector spaces, \(\Q_{\geq 0}\)-g. ones are cones, etc. but we will avoid this in favor of the general terminology of being \(\mathbb{S}\)-generated.

A set \(\vectSet{L} \subseteq \N^n\) is \emph{hybridlinear} if \(\vectSet{L}=\vectSet{B}+\N(\vectSet{F})\) for some finite sets \(\vectSet{B}, \vectSet{F} \subseteq \N^n\). If \(\vectSet{B}=\{\vect{b}\}\) can be chosen as a singleton, then \(\vectSet{L}\) is \emph{linear}. A set \(\vectSet{S}\) is \emph{semilinear} if it is a finite union of (hybrid)linear sets. The semilinear sets are equivalently definable via formulas \(\varphi \in \FO(\mathbb{N}, +)\), called Presburger Arithmetic.

The \emph{dimension} of a \(\Q\)-generated set defined as its minimal number of generators is a well-known concept. It can be extended to arbitrary subsets of \(\mathbb{Q}^n\) as follows.

\begin{definition}{\cite{Leroux11}}
Let \(\vect{X} \subseteq \mathbb{Q}^n\). The \emph{dimension} of \(\vect{X}\), denoted \(\dim(\vect{X})\), is the smallest natural number \(k\) such that there exist finitely many \(\Q\)-g. sets \(\vect{V}_i \subseteq \mathbb{Q}^n\) with \(\dim(\vect{V}_i)\leq k\) and \(\vect{b}_i \in \mathbb{Q}^n\) such that \(\vect{X} \subseteq \bigcup_{i=1}^r \vect{b}_i + \vect{V}_i\). 
%[\(\dim(\emptyset):=-\infty\)]
\end{definition}

Let \(\vectSet{X}_1, \vectSet{X}_2 \subseteq \N^n\) be sets. Then \(\vectSet{X}_1\) and \(\vectSet{X}_2\) \emph{have a non-degenerate intersection} if \(\dim(\vectSet{X}_1 \cap \vectSet{X}_2)=\dim(\vectSet{X}_1)=\dim(\vectSet{X}_2)\).
In the sequel many of our lemmas only work under the assumption that two sets have a non-degenerate intersection, i.e.\ this is a very central notion.

\begin{example}
%Let us provide some examples of (non-)degenerate intersections. 
Intersecting \(\N^2 \cap \N=\N\) is degenerate because we did not intersect same dimension objects. Also 
the intersection
\(\{(x,y) \mid y \geq x\} \cap \{(x,y) \mid y \leq x\}=\{(x,y) \mid y=x\}\) is degenerate because it is lower dimensional. 
On the other hand \(\{(x,y) \mid y \geq x\} \cap \{(x,y) \mid y \leq 2x\}\) is a typical non-degenerate intersection. 
\end{example}

Because ``\(\vectSet{L}\) and \(\vectSet{L}'\) have a non-degenerate intersection'' is rather long, we will often just write ``\(\vectSet{L} \cap \vectSet{L}'\) is non-degenerate''.

All the definitions above defined for sets apply to relations \(\vectSet{R} \subseteq \N^{n_1} \times \N^{n_2}\) by viewing them as sets \(\vectSet{R} \subseteq \N^{n_1+n_2}\).

\para{Fast-Growing Function Hierarchy}: We let \(F_1: \N \to \N, m \mapsto 2m\), and define \(F_d: \N \to \N, m \mapsto F_{d-1}^{(m)}(1)\), where \(F_{d-1}^{(m)}\) is \(m\)-fold application of the function \(F_{d-1}\). For example \(F_2(m)=2^m\), \(F_3=\text{Tower}\) and so on. The functions \(F_d\) are called the fast-growing functions. (One possible) Ackermann-function \(F_{\omega}\) is obtained from these via diagonalization: \(F_{\omega}: \N \to \N, m \mapsto F_m(m)\). The \(d\)-th level of Grzegorczyk's hierarchy \cite{Schmitz16} is the set of functions \(\mathfrak{F}_d:=\{F_{d} \circ r_1 \circ \dots \circ r_k \mid r_1, \dots, r_k \in \mathfrak{F}_{d-1}\}\). I.e. we close \(F_d\) under applying reductions of the lower level \(\mathfrak{F}_{d-1}\). For example \(\mathfrak{F}_3\) is the set of functions obtained by inputting an elementary function into the tower function. We consider \(\mathfrak{F}_d\) as a complexity class via identifying \(\mathfrak{F}_d\) with \(\bigcup_{f \in \mathfrak{F}_d} \mathbf{DTIME}(f)\). Observe that the choice of deterministic time is irrelevant: Starting at \(\mathfrak{F}_3\) these classes are closed under exponential time reductions, and hence time and space complexity classes coincide.

One main way to prove that a function falls into some level of this hierarchy is a theorem of \cite{FigueiraFSS11}. It considers sequences over \(\N^n\). 
%Given \(a \in \N\), write \(\vect{a}\) for the constant vector \(\vect{a}=(a,\dots, a) \in \N^n\). 
Given a sequence \((\vect{x}_0, \vect{x}_1, \dots, \vect{x}_\ell)\in\N^{\ell+1}\), \(m \in \N\) and \(f: \N \to \N\) we call the sequence \((f,m)\)-\emph{controlled} if \(||\vect{x_i}||_{\infty} \leq f^i(m)\) for all \(i\), i.e. the sequence starts below \(m\), and in every step entries grow at most by an application of the function \(f\). The sequence \emph{contains an increasing pair} if \(\vect{x}_i \leq \vect{x}_j\) for some \(i<j\).

\begin{proposition} \label{PropositionFastGrowingComplexity}
\cite[Prop. 5.2]{FigueiraFSS11} Let \(k, \gamma \geq 1\) and  \(f\in\mathfrak{F}_{\gamma}\) be a monotone function with \(f(x) \geq x\) for all \(x\). Then the function mapping \(m \in \N\) to the length of the longest \((f,m)\)-controlled sequence in \(\N^k\) without an increasing pair is in \(\mathfrak{F}_{\gamma+k-1}\).
\end{proposition}

For example if a rank in \(\N^k\) decreases lexicographically in an algorithm, then the sequence of ranks has no increasing pair. We therefore obtain a complexity bound for the algorithm.

\section{Extended Vector Addition Systems} \label{SectionVAS}

% !TEX root = Main.tex

Let \(\RelationClass\) be a class of relations on \(\N^n\). A \(\RelationClass\)-extended VASS (\(\RelationClass\)-eVASS) with \(n\) counters is a finite directed multigraph \((Q,E)\) which is labelled as follows:

\begin{enumerate}
\item Every strongly-connected component (SCC) \(S \subseteq Q\) has a subset \(I(S)\) of active counters.
\item Every edge \(e\) inside an SCC is labelled with \(\vectSet{R}(e) \in \RelationClass\).
\item Every edge \(e\) leaving an SCC is labelled with one of the following three: a relation \(\vectSet{R}(e) \in \RelationClass\), a subset \(I_+(e) \subseteq \{1,\dots, n\}\) of counters to add, or a subset \(I_-(e) \subseteq \{1,\dots, n\}\) of counters to delete.
\end{enumerate}

The set of configurations is \(\bigcup_{SCC\ Q'} Q' \times \N^{I(Q')}\), and an edge \(e\) has semantics \(\to_e\) defined by \((q,\vect{x}) \to_e (p, \vect{y})\) if \(e=(q,p)\) and,
depending on the label of $e$, 
either \((\vect{x}, \vect{y}) \in \vectSet{R}(e)\), or \(\vect{y}=\vect{x}_{I_+(e)}\)  is equal to \(\vect{x}\) with coordinates from \(I_+(e)\) added with value \(0\), or, symmetrically, \(\vect{x}=\vect{y}_{I_-(e)}\).
%in case 3 we require \(\vect{x}[i]=0\) for all \(i \in I_-(e)\) and then \(\vect{y}\) is equal to \(\vect{x}\) with coordinates in \(I_-(e)\) deleted. 
In particular, a counter deletion performs a zero test.

%The set \(\Omega\) of all runs is wqo with the amalgamation property as follows, called the \emph{Jancar ordering}: We view \(\Omega\) as \(\Omega \subseteq Conf \times (\bigcup_{e \in E} \{e\} \times \Omega(e))^{\ast} \times Conf\)%, where a run \(\rho=(\vect{x}_0, \dots, \vect{x}_r)\) induces \((\vect{x}_0, ((\vect{x}_0, \rho(e_0), \vect{x}_1), \dots, (\vect{x}_{r-1}, \rho(e_{r-1}), \vect{x}_r)),\vect{x}_r) \in \Omega\). The set \(\N^{I_{in}} \times (\bigcup_{e \in E} \{e\} \times \Omega(e))^{\ast} \times \N^{I_{fin}}\) and hence \(\Omega\) is well-quasi-ordered by Lemma \ref{DicksonsLemma} and \ref{HigmanLemma}.
%, where \(Conf\) is the set of configurations. I.e.\ a run \(\rho\) is identified with the triple \((\source(\rho), steps(\rho), \target(\rho)) \in Conf \times (\bigcup_{e \in E} \{e\} \times \Omega(e))^{\ast} \times Conf\). Then \(\Omega\) is well-quasi-ordered by Lemma \ref{DicksonsLemma} and \ref{HigmanLemma}.

Intuitively, \(\RelationClass\)-eVASS model finite automata operating on counters with values in \(\N\). An operation consists of a state change of the automaton and updating the counters according to the relation \(\vectSet{R}(e)\) written on the edge. Sometimes it is convenient to change the dimension of the system when leaving an SCC, hence the automaton is allowed to add/delete counters when leaving an SCC.

There are multiple classes of systems which fall into this definition. For example consider the class \(Add\) of relations of the form \(\to_{\vect{a}}\) for \(\vect{a} \in \Z^n\) defined as \(\vect{x} \to_{\vect{a}} \vect{y} \iff \vect{y}=\vect{x}+\vect{a}\). Then \(Add\)-eVASS without counter additions/deletions form the class of \emph{vector addition system with states} (VASS).

Furthermore, counter machines are \(Semil\)-eVASS, where \(Semil\) is the class of semilinear relations. To see this, observe that zero tests \(ZT(I,n):=\{(\vect{x}, \vect{y}) \in \N^n \times \N^n \mid \vect{x}=\vect{y} \text{ and }\vect{x}[i]=0\ \forall\  i \in I\}\) are a special case of linear relations.

Another subclass we will consider are VASSnz. Let \(NZT\) be the class of \emph{nested zero tests}, defined as relations of the form \(NZT(j,n):=ZT(\{1,\dots, j\},n)\) for some \(j \leq n \in \N\). Intuitively, the vector \(\vect{x}\) stays the same and the first \(j\) coordinates are ``tested for \(0\)''. In particular if counter \(j\) is tested also all lower index counters are tested. 

The class \((Add \cup NZT)\)-eVASS is called \ConsideredModel.

Observe that class \(\RelationClass\) is allowed to contain non-deterministic relations \(\vectSet{R} \in \RelationClass\), in this way \(\RelationClass\)-eVASS can naturally model both determinism and non-determinism.

We continue with a few more semantic definitions. Let \(\qin\in Q\) and \(\qfin\in Q\) be states called the initial and final state, respectively. A run \(\rho\) from \(\qin\) to \(\qfin\) is a sequence \((p_0(\vect{x}_0), \dots, p_r(\vect{x}_r))\) of configurations s.t. \(p_i(\vect{x}_i) \to_{e_i} p_{i+1}(\vect{x}_{i+1})\) for some edges \(e_i \in E\) and \(p_0=\qin\) and \(q_r=\qfin\). The source of \(\rho\) is \(\vect{x}_0\), the target is \(\vect{x}_r\), and the source/target pair is \(\dirOfRun(\rho):=(\vect{x}_0, \vect{x}_r)\). We write \(\Omega_{\qin, \qfin}\) for the set of all runs from \(\qin\) to \(\qfin\).  The reachability relation is \(\Rel(\VAS, \qin, \qfin):=\dirOfRun(\Omega_{\qin, \qfin}) \subseteq \N^n \times \N^n\).

The \emph{dimension} of a \(\RelationClass\)-eVASS \(\VAS\), intuitively its ``complexity'', is defined as \(\dim(\VAS):=\max_{q \in Q} \dim(\Q(\Delta(\Rel(\VAS, q, q))))\), i.e.\ for every state \(q \in Q\), we consider the \(\Q\)-g. set generated by all effects of runs from \(q\) to \(q\), take its dimension (number of generators), and finally the maximum over all \(q \in Q\). Clearly \(\dim(\VAS) \leq n\) if \(\VAS\) has \(n\) counters, but the dimension may be considerably smaller. E.g. an important subclass of VASS are \emph{flat} VASS, which are defined by allowing only one cycle on every state. Then in particular the \(\Q\)-g. set of cycle effects has dimension \(1\), i.e.\ all flat VASS have dimension \(1\).

\section{Algorithm Toolbox} \label{AlgorithmToolbox}

% !TEX root = Main.tex

In this section we introduce important lemmas, definitions and tricks we will use throughout. This is intended as a quick lookup location containing most relevant lemmas/definitions.

We start with one of the most important  definitions, to be
the basis of our reachability algorithm.
By $\size(\vectSet X)$ we mean the size of any reasonable representation of $\vectSet X$.
%For instance, the representation of a \(\RelationClass\)-eVASS section 
%$\vectSet X= \pi(\vectSet R \cap \vectSet L)$
%includes recursively the representations of all relations $\vectSet R(e)$ appearing
%in the underlying \(\RelationClass\)-eVASS.

\begin{definition} \label{DefinitionGoodOverapproximation}
Let \(\vectSet{X} \subseteq \vectSet{L} \subseteq \N^{d}\) be sets, \(\vectSet{L}=\vectSet{B}+\N(\vectSet{F})\) hybridlinear. Then \(\vectSet{L}\) is an \emph{asymptotic overapproximation} of \(\vectSet{X}\) (written \(\vectSet{X} \HybridizationRelation \vectSet{L}\)) if for every \(\vect{x} \in \vectSet{L}\) and every \(\vect{w} \in \N_{\geq 1}(\vectSet{F})\), there exists \(N \in \N\) such that \(\vect{x}+\N_{\geq N}\vect{w} \subseteq \vectSet{X}\). 

If \(N \leq g(\size(\vectSet{X})+\size(\vect{x})+\size(\vect{w}))\) for some function \(g \colon \N \to \N\), then \(\HybridizationRelation\) is \(g\)-bounded, written \(\HybridizationRelation_{g}\).

We use \(\N_{\geq 1}(\vectSet{F})\) instead of \(\N(\vectSet{F})\), which is crucial.

This definition applies to relations \(\vectSet{X}, \vectSet{L} \subseteq \N^{n_1} \times \N^{n_2}\) by viewing them as sets \(\vectSet{X}, \vectSet{L} \subseteq \N^{n_1+n_2}\).
\end{definition}

I.e.~\(\vectSet{L}\) is an asymptotic overapproximation  of a relation \(\vectSet{X}\) if starting at any point \(\vect{x} \in \vectSet{L}\) and ``walking in an interior direction'' \(\vect{w}\) of \(\vectSet{L}\), eventually all the visited points are in \(\vectSet{X}\). 
%Since this is a very central notion, let us carefully introduce an example and some properties.

\begin{example}
Clearly, every hybridlinear set \(\vectSet{L}\) is its own asymptotic overapproximation. A more interesting example is on the right of Figure \ref{FigureIntuitionSemilinearityAlgorithm}: While \(\vectSet{X}\) is non-semilinear, as long as \(\vect{w}\) is not ``vertical'', i.e.\ \(\N_{\geq 1}\) is indeed crucial, one can find an \(N\) as in Definition \ref{DefinitionGoodOverapproximation}. Observe also that there is no uniform \(N \in \N\) which works for every \(\vect{x}, \vect{w}\): Different \(\vect{x} \in \N^2\) can have a different ``distance'' from \(\vectSet{X}\).
\end{example}

\begin{figure}[h!]
\begin{minipage}{4.5cm}
\begin{tikzpicture}
\begin{axis}[
    axis lines = left,
    xlabel = { },
    ylabel = { },
    xmin=0, xmax=8,
    ymin=0, ymax=8,
    xtick={0,2,4,6,8},
    ytick={0,2,4,6,8},
    ymajorgrids=true,
    xmajorgrids=true,
    thick,
    smooth,
    no markers,
]

\addplot[
    fill=blue,
    fill opacity=0.5,
    only marks,
    ]
    coordinates {
    (0,0)(0,1)(0,2)(0,3)(0,4)(0,5)(0,6)(0,7)(0,8)(1,0)(2,0)(3,0)(4,0)(5,0)(6,0)(7,0)(8,0)(1,1)(1,2)(1,3)(1,4)(1,5)(1,6)(1,7)(1,8)(2,1)(2,2)(2,3)(2,4)(2,5)(2,6)(2,7)(2,8)(3,1)(3,2)(3,3)(4,1)(4,2)(4,3)(5,1)(5,2)(5,3)(6,1)(6,2)(6,3)(7,1)(7,2)(7,3)(8,1)(8,2)(8,3)
    };
    
%\addplot[
%    fill=green,
%    fill opacity=0.7,
%    only marks,
%    ]
%    coordinates {
%    (3,4)(3,5)(3,6)(3,7)(3,8)(4,4)(4,5)(4,6)(4,7)(4,8)(5,4)(5,5)(5,6)(5,7)(5,8)(6,4)(6,5)(6,6)(6,7)(6,8)(7,4)(7,5)(7,6)(7,7)(7,8)(8,4)(8,5)(8,6)(8,7)(8,8)
%    };
    
\addplot[
   draw=red,
   no marks,
   very thick,
   ]
   coordinates {
   (0,0)(0,8)
   };
   
\addplot[
   draw=red,
   no marks,
   very thick,
   ]
   coordinates {
   (1,0)(1,8)
   };
   
\addplot[
   draw=red,
   no marks,
   very thick,
   ]
   coordinates {
   (2,0)(2,8)
   };
   
\addplot[
   draw=red,
   no marks,
   very thick,
   ]
   coordinates {
   (0,0)(8,0)
   };
   
\addplot[
   draw=red,
   no marks,
   very thick,
   ]
   coordinates {
   (0,1)(8,1)
   };
   
\addplot[
   draw=red,
   no marks,
   very thick,
   ]
   coordinates {
   (0,2)(8,2)
   };
   
\addplot[
   draw=red,
   no marks,
   very thick,
   ]
   coordinates {
   (0,3)(8,3)
   };

\end{axis}
\end{tikzpicture}
\end{minipage}%
\begin{minipage}{4.5cm}
\begin{tikzpicture}
\begin{axis}[
    axis lines = left,
    xlabel = { },
    ylabel = { },
    xmin=0, xmax=4,
    ymin=0, ymax=16,
    xtick={0,1,2,3,4},
    ytick={0,4,9,16},
    ymajorgrids=true,
    xmajorgrids=true,
    thick,
    smooth,
    no markers,
]
    
\addplot[
    fill=blue,
    fill opacity=0.5,
    only marks,
    ]
    coordinates {
    (0,0)(1,0)(1,1)(2,0)(2,1)(2,2)(2,3)(2,4)(3,0)(3,1)(3,2)(3,3)(3,4)(3,5)(3,6)(3,7)(3,8)(3,9)(4,0)(4,1)(4,2)(4,3)(4,4)(4,5)(4,6)(4,7)(4,8)(4,9)(4,10)(4,11)(4,12)(4,13)(4,14)(4,15)(4,16)
    };
    
\addplot[
    name path=A,
    domain=0:4,
    color=red,
]
{x^2};

\end{axis}
\end{tikzpicture}
\end{minipage}%

\caption{\textit{Left}: Illustration of \(\dim(\vectSet{L} \setminus (\vect{p}+\vectSet{L}))<\dim(\vectSet{L})\) with \(\vectSet{L}=\N^2\) and \(\vect{x}=(3,4)\). The dimension drops from 2 to 1, since the set is now coverable by finitely many lines.  \newline
\textit{Right}: \(\vectSet{X}=\{(x,y)\in \N^2 \mid y \leq x^2\}\) has the asymptotic overapproximation \(\N^2\). It is \(g\)-bounded for a quadratic function \(g\).}\label{FigureIntuitionSemilinearityAlgorithm}
\end{figure}

Some more in-depth understanding of asymptotic overapproximations can be gained by considering closure properties:

\begin{restatable}{lemma}{LemmaShiftGoodOverapproximation} \label{LemmaShiftGoodOverapproximation}
Let \(\vectSet{X} \HybridizationRelation_g \vectSet{L}\) and let \(\vectSet{L}'\) be a hybridlinear set s.t. \(\vectSet{L} \cap \vectSet{L}'\) is non-degenerate. Then \(\vectSet{X} \cap \vectSet{L}' \HybridizationRelation_g \vectSet{L} \cap \vectSet{L}'\).
%Let \(\vectSet{X}_{12} \HybridizationRelation \vectSet{L}_{12} \subseteq \N^{d_1} \times \N^{d_2}\) and \(\vectSet{X}_{23} \HybridizationRelation \vectSet{L}_{23}\subseteq \N^{d_2} \times \N^{d_3}\). Assume that \(\pi_{in}(\vectSet{L}_{23}) \cap \pi_{out}(\vectSet{L}_{12})\) is non-degenerate, where \(\pi_{in} \colon \N^{d_2} \times \N^{d_3} \to \N^{d_2}\) and \(\pi_{out} \colon \N^{d_1} \times \N^{d_2} \to \N^{d_2}\) are the projections to in- and output. Then \(\vectSet{X}_{12} \circ \vectSet{X}_{23} \HybridizationRelation \vectSet{L}_{12} \circ \vectSet{L}_{23}\).
\end{restatable}

\begin{proof}[Proof idea]
Observe that the definition of asymptotic overapproximation makes a claim about every point \(\vect{x}\in \vectSet{L}\) and every \(\vect{w} \in \N_{\geq 1}(\vectSet{F})\), hence decreasing \(\vectSet{L}\) and decreasing the set \(\N_{\geq 1}(\vectSet{F})\) of possible \(\vect{w}\)'s preserves the property. Hence the main part of the proof is to check that \(\N_{\geq 1}(\vectSet{F})\) decreases through the intersection. Observe that if we for example intersect \(\N^2 \cap (\{0\} \times \N)\), then \(\N_{\geq 1}(\{(1,0),(0,1)\}) \not \supseteq \N_{\geq 1}(\{(0,1)\})\): We would  get the vertical direction \(\vect{w}\) as before. Therefore the non-degenerate intersection is crucial.
\end{proof}

In order for asymptotic overapproximations to be useful, we need an algorithm which given \(\vectSet{X}\in \RelationClass\), splits it into finitely many parts with an asymptotic overapproximation each.

\begin{definition} \label{DefinitionApproximable}
Let \(\RelationClass\) be a class of relations. Then \(\RelationClass\) is \emph{approximable} in \(\mathfrak{F}_{\alpha}\) if 
there is a function
\(g \in \mathfrak{F}_{\alpha}\) and 
an algorithm in \(\mathfrak{F}_{\alpha}\) which given \(\vectSet{X} \in \RelationClass\), outputs finitely many \(\vectSet{X}_1, \dots, \vectSet{X}_k \in \RelationClass\) and hybridlinear relations \(\vectSet{L}_1, \dots, \vectSet{L}_k\) s.t. 
\begin{enumerate}
\item \(\vectSet{X}=\bigcup_{j=1}^k \vectSet{X}_j\),
\item If \(\vectSet{X}\) is monotone, then also all the \(\vectSet{X}_j\) are so.
\item For all \(1 \leq j \leq k\) we have \(\vectSet{X}_j \HybridizationRelation_g \vectSet{L}_j\).
\end{enumerate}
\end{definition}
%
%\(g \in \mathfrak{F}_{\alpha}\) and 
%
We will later give such an algorithm for reachability relations of a large class of eVASS, in particular VASS, \ConsideredModel, etc.

On the left of Figure \ref{FigureIntuitionSemilinearityAlgorithm}, we illustrate an important lemma:

\begin{lemma}[Cor.~D.2 in \cite{Leroux13}] \label{LemmaDimensionDecreaseShiftedL}
Let \(\vectSet{L}=\vect{B}+\N(\vectSet{F})\) be hybridlinear, \(\vect{p} \in \N(\vectSet{F})\). Then \(\dim(\vectSet{L} \setminus (\vect{p}+\vectSet{L}))<\dim(\vectSet{L})\).
\end{lemma}

Both Lemma \ref{LemmaShiftGoodOverapproximation} and Lemma \ref{LemmaDimensionDecreaseShiftedL} can be used for recursion in algorithms: By Lemma \ref{LemmaShiftGoodOverapproximation}, problems only occur if the intersection has lower dimension, and by Lemma \ref{LemmaDimensionDecreaseShiftedL}, it is sufficient to understand a set \emph{asymptotically}, as we can deal with the rest via an appropriate recursion.

Finally, since we will often have to deal with projections, which can be slightly annoying, we use the following:

\begin{lemma} \label{LemmaNiceOverapproximationProjection}
If \(\vectSet{X} \HybridizationRelation_g \vectSet{L}\) for some function \(g\) and \(\pi\) is any projection, then \(\pi(\vectSet{X}) \HybridizationRelation_g \pi(\vectSet{L})\).
\end{lemma}

\begin{proof}
Immediate from the definition.
\end{proof}

\section{VASSnz and monotone eVASS} \label{SectionMainAlgorithm}

% !TEX root = Main.tex

In this section we introduce a restriction of \(\RelationClass\)-eVASS, called \emph{monotone} \(\RelationClass\)-eVASS (\(\RelationClass\)-m-eVASS), and prove that \ConsideredModel\ can be converted into \(\RelationClass\)-m-eVASS. This will be the first step of our reachability algorithm.

Similar to how \(Semil\)-eVASS can simulate counter machines, if transitions inside an SCC are allowed to be non-monotone, then \(\RelationClass\)-eVASS will likely be undecidable. Hence we define:

\begin{definition} \label{DefinitionMonotoneEVASS}
Let \(\RelationClass\) be a class of relations. A \emph{monotone} \(\RelationClass\)-eVASS is a \(\RelationClass\)-eVASS where transition labels \(\vectSet{R}(e)\) for edges \(e\) \emph{inside an SCC} are \emph{monotone} relations \(\vectSet{R}(e) \in \RelationClass\).
\end{definition}

For example for \(\RelationClass=Add\) all \(\RelationClass\)-eVASS are monotone, because this class \(\RelationClass\) contains only monotone relations. On the other hand, \emph{monotone} \(Semil\)-eVASS can no longer simulate counter machines, in fact they have the same sections as VASS: Using the 
\emph{controlling counter technique} \cite{CzerwinskiO21}, VASS can in fact simulate zero tests on the exits of SCCs. 

We will later require that \(\RelationClass\) is closed under intersection with semilinear sets, which is not true for the class of reachability relations of \(\RelationClass\)-m-eVASS. Hence we use the class of sections, as considered for VASS in \cite{ClementeCLP17}. 

\begin{definition}
A relation \(\vectSet{X} \subseteq \N^{n_1} \times \N^{n_2}\) is a \(\RelationClass\)-m-eVASS section if \(\vectSet{X}=\pi(\vectSet{R} \cap \vectSet{S})\), where \(\vectSet{S} \subseteq \N^n \times \N^n\) for some \(n \geq n_1, n_2\) is semilinear, \(\vectSet{R}\) is the reachability relation of a \(\RelationClass\)-m-eVASS with \(n\) counters, and \(\pi \colon \N^n \times \N^n \to \N^{n_1} \times \N^{n_2}\) is a \emph{projection}, i.e.\ a function deleting some coordinates.
\end{definition}

Sections are often defined by only allowing zero tests as \(\vectSet{S}\) \cite{ClementeCLP17}. 
These definitions are in fact equivalent:

\begin{restatable}{lemma}{LemmaRestatableEquivalentlyZeroTest} \label{LemmaLabelEquivalentlyZeroTest}
Assume that \(\RelationClass\) contains \(Add\). 

Let \(\vectSet{X}=\pi(\vectSet{R} \cap \vectSet{S}) \subseteq \N^{n_1} \times \N^{n_2}\) be a section of a
\(\RelationClass\)-m-eVASS $\VAS$. Then \(\vectSet{X}=\pi'(\vectSet{R}' \cap WZT(I,I',n))\) for some linear set \(WZT(I,I',n)=\{(\vect{x}, \vect{y})\mid \vect{x}[i]=0\ \forall i \in I, \vect{y}[i]=0\ \forall i \in I'\}\) with \(n \geq n_1, n_2\), \(I,I' \subseteq \{1,\dots, n\}\), and the reachability relation $\vectSet R'$ of a \(\RelationClass\)-m-eVASS
of the same dimension as $\VAS$,
and the new representation is computable in \(\mathbf{PTIME}\).
\end{restatable}

Given a VASSnz \(\VAS\), the number of priorities \(k\) is the number of different \(j\) s.t. some transition has the label \(NZT(j,d)\), i.e.\ the number of different types of zero tests. We prove that VASSnz with \(k\) priorities can be converted into monotone eVASS, whose labels are VASSnz sections with \(k-1\) priorities.

\begin{lemma} \label{LemmaConvertVASSnzToEVASS}
Let \(\RelationClass\):=\((k{-}1)\)-VASSnzSec be the class of sections of VASSnz with \(k{-}1\) priorities. There is a polytime algorithm converting a VASSnz  with $k$ priorities into a monotone \(\RelationClass\)-eVASS  with the same reachability relation and dimension.
\end{lemma}

\begin{proof}
First let us explain the crucial observation this construction is based on: What is a zerotest? It is a linear relation, which is monotone on some counters and \emph{leaves the others} \(0\). If the fixed counters are considered as \emph{non-existent} in the current SCC \(S\) (eVASS have this capability), then a zero test is monotone and may be used as transition label.

Let \(\VAS_k=(Q,E)\) be a VASSnz with $k$ priorities, and
\(m\) be the index of its highest counter which is zero tested.
Let initial/final states be \(\qin, \qfin\). Let \(\VAS_{k-1}\) be the VASSnz obtained from \(\VAS_k\) by deleting all edges zero testing counter \(m\).

The target \(\RelationClass'\)-eVASS \(\VAS'=(Q',E')\) has \(5\) types of states \(Q':=\{src, del, main, add, tgt\} \times Q\) and \(7\) types of edges \(E_1, \dots, E_7\) with \(E':=\bigcup_{i=1}^7 E_i\). The new initial/final states are \(\qin':=(\source, \qin)\) and \(\qfin':=(\target, \qfin)\). Let us start with the simplest type of edge: \(E_1=\{e_1\}\), where \(e_1: (src, \qin) \to (tgt,\qfin)\) is labelled with \(\Rel(\VAS_{k-1}, \qin, \qfin)\). This allows \(\VAS'\) to simulate runs of \(\VAS_k\) not using zero tests on \(m\).

Next let us explain the states \(\{main\} \times Q\), where the interesting computation happens: The set of active counters for these states is 
%\(I(\{main\} \times Q)=
\(\{m+1,\dots, d\}\), i.e.\ in these states we assume that the first \(m\) counters have been deleted. We have two types of edges on \(main\) called \(E_2\) and \(E_3\), where \(E_2\) simulates zero tests on \(m\). Formally, we define \(E_2:=\{((main,q), (main,p)) \mid (q,p) \in E, \text{ and has the label }\vectSet{R}(q,p)=NZT(m,d)\}\). We give every \(e_2 \in E_2\) the label \(\vectSet{R}(e_2):=\{(\vect{x}, \vect{x}) \mid \vect{x} \in \N^{d-m}\}\). Observe that this label faithfully applies the zero test: The first \(m\) counters do not exist in this SCC, and on the rest of the counters the ``zero test'' is the identity.

Next we consider \(E_3\). The idea is that in state \(main\) we want to still be able to apply transitions of \(\VAS_{k-1}\). Set \(E_3:=(\{main\} \times Q) \times (\{main\} \times Q)\), i.e.\ for all \((p,q) \in Q^2\) we have an edge.  For every \(e_3=(main, p) \to (main, q) \in E_3\) we set the label \(\vectSet{R}(e_3)\) of \(e_3\) to
\[\pi_m(\Rel(\VAS_{k-1}, p, q) \cap \{(\vect{x}, \vect{y}) \mid \vect{x}[i]=\vect{y}[i]=0\ \forall 1 \leq i \leq m\}),\] where \(\pi_m\) projects the first \(m\) coordinates away. Hence edges \(e_3\) apply arbitrary runs of \(\VAS_{k-1}\) between the corresponding states which \emph{preserve value} \(0\) on the first \(m\) counters.

Next we explain the states \(del\) and \(add\): They ensure that when entering and leaving the \(main\) component we have only the counters \(m+1, \dots, d\). 
We set \(E_4:=\{((del, p), (main, p)) \mid p \in Q\}\), and add the label \(I_-(e_4):=\{1,\dots, m\}\), i.e.\ we delete the first \(m\) many counters and do not change the state of \(\VAS_k\). Similarly, we set \(E_5:=\{((main, p),(add, p)) \mid p \in Q\}\) with label \(I_+(e_5):=\{1,\dots, m\}\), which readds the first \(m\) counters. 

We set \(E_6=\{((\source, \qin),(del, p)) \mid p \in Q\}\), and for every \(e_6 \in E_6\) the label \(\vectSet{R}(e_6):=\Rel(\VAS_{k-1}, q_{in}, p) \cap Z$, where $Z=\{(\vect{x}, \vect{y}) \mid \vect{y}[1]= \ldots = \vect{y}[m] = 0\}\), i.e.\ we perform any run of \(\VAS_{k-1}\) which starts in \(q_{in}\) and sets the first \(m\) counters to \(0\). 

Symmetrically, we set \(E_7:=\{((add, p), (\target, \qfin)) \mid p \in Q\}\) and  \(\vectSet{R}(e_7):=\Rel(\VAS_{k-1}, p, \qfin)\cap Z^{-1}\), i.e.\ we end with any run of \(\VAS_{k-1}\) that starts with the first $m$ counters equal 0 and ends in $\qfin$. This completes the construction of \(\VAS'\).

To see that \(\VAS'\) is equivalent to \(\VAS_k\), let \(\rho\) be any run of \(\VAS_k\). If \(\rho\) does not use a zero-test on \(m\) then we simulate \(\rho\) using \(E_1\). Otherwise \(\rho\) eventually visits a configuration \(\vect{x}\) with \(\vect{x}[i]=0\) for all \(1 \leq i \leq m\) to enable the zero test. 

Write \(\rho=\rho_{\text{pre}} \rho_{\text{mid}} \rho_{\text{suf}}\) where \(\rho_{\text{pre}}\) is the prefix until the first time \(\vect{x}[i]=0\ \forall\ 1 \leq i \leq m\) and \(\rho_{\text{suf}}\) is the suffix starting from the last time we have \(\vect{x}[i]=0\ \forall\ 1 \leq i \leq m\). Then \(\rho_{\text{pre}}\) and \(\rho_{\text{suf}}\) do not use zero tests on \(m\), hence we can simulate them using edges of type \(E_6\) and \(E_7\) respectively. To simulate \(\rho_{mid}\), we use edge types \(E_2\) and \(E_3\) repeatedly.
\end{proof}

Lemma \ref{LemmaConvertVASSnzToEVASS} shows that even if \(\RelationClass\) contains non-monotone relations, one can recover some monotonicity at the cost of complicated edge labels. This might seem like a difficulty, but they will turn out to be surprisingly easy to deal with.

Recall Definition \ref{DefinitionApproximable} of approximability.  We will now overapproximate monotone \(\RelationClass\)-eVASS and therefore \ConsideredModel:

\begin{theorem} \label{TheoremIdealDecompositionEVASS}
Let \(\alpha \geq 2\) and let \(\RelationClass\) be a class of relations approximable in \(\mathfrak{F}_{\alpha}\), containing \(Add\) and effectively closed under intersection with semilinear relations in \(\mathfrak{F}_{\alpha}\).
Then sections of monotone \(\RelationClass\)-eVASS are approximable in \(\mathfrak{F}_{\alpha+2d+2}\).
\end{theorem}

We move the proof of Theorem \ref{TheoremIdealDecompositionEVASS} to Section \ref{SectionProofTheoremEVASS}, and first use it to show Theorem \ref{TheoremVASSnzIdealDecomposition}, our first main result:

\begin{theorem} \label{TheoremVASSnzIdealDecomposition}
The class of sections of VASSnz of dimension \(d\) and with \(k\) priorities is approximable in \(\mathfrak{F}_{2kd+2k+2d+4}\). 

The class of all VASSnz sections is approximable in \(\mathfrak{F}_{\omega}\) time.
\end{theorem}

\begin{proof}
By induction on \(k\).

\para{\(k=0\):} VASS are a special case of monotone Semil-eVASS, and since the class of semilinear sets is approximable in \(\mathfrak{F}_2\) (namely we have \(\vectSet{L} \HybridizationRelation \vectSet{L}\ \forall\ \vectSet{L}\)), we hence obtain \(\mathfrak{F}_{2+2d+2}=\mathfrak{F}_{2d+4}\) time for VASS by Theorem \ref{TheoremIdealDecompositionEVASS}.

\para{\(k{-}1 \to k\):} By Lemma \ref{LemmaConvertVASSnzToEVASS} we can convert \ConsideredModel\ with \(k\) priorities into \(\RelationClass\)-eVASS for \(\RelationClass:=\) \((k{-}1)\)-VASSnzSec. By induction together with Theorem \ref{TheoremIdealDecompositionEVASS}, \(\RelationClass\)-eVASS are  approximable in 
\(\mathfrak{F}_{2kd+2k+2d+4}\).

\smallskip

For the class of all VASSnz we immediately obtain \(\mathfrak{F}_{\omega}\).
\end{proof}

%Theorem \ref{TheoremVASSnzIdealDecomposition} trivially implies the following:

\begin{corollary} \label{CorollaryVASSnzIdealDecomposition}
Reachability for \ConsideredModel \ is in \(\mathfrak{F}_{\omega}\).
\end{corollary}

\begin{proof}
If \(\vectSet{X} \HybridizationRelation \vectSet{L}\), then \(\vectSet{X} \neq \emptyset\) by simply choosing any \(\vect{x} \in \vectSet{L}\) and \(\vect{w} \in \N_{\geq 1}(\vectSet{F})\) and obtaining an  \(\vect{x}+N\vect{w} \in \vectSet{L}\). Therefore a given \(\vectSet{X}\) is non-empty if and only if the algorithm
of Theorem \ref{TheoremVASSnzIdealDecomposition} applied to \(\vectSet{X}\) outputs some 
\(\vectSet{X_j} \HybridizationRelation \vectSet{L}_j\).
\end{proof}

Hence we can now focus solely on proving Theorem \ref{TheoremIdealDecompositionEVASS}.

\section{Approximating monotone eVASS} \label{SectionProofTheoremEVASS}

% !TEX root = Main.tex

In this section we prove Theorem \ref{TheoremIdealDecompositionEVASS}, i.e.\ we prove that the class of sections of monotone \(\RelationClass\)-eVASS is approximable.

In order to overapproximate a section \(\vectSet{X}=\pi(\vectSet{R} \cap \vectSet{L})\), we first perform some preprocessing: By Lemma \ref{LemmaLabelEquivalentlyZeroTest} we assume that \(\vectSet{L}=WZT(I, I',n)\) is a zero test, and afterwards by Lemma \ref{LemmaNiceOverapproximationProjection} we may overapproximate \(\vectSet{R} \cap WZT(I, I',n)\) instead. Since removing a counter and subsequently readding it simulates a zero test, we can even assume that $\vectSet X = \vectSet R$.
% we only have to approximate \(\vectSet{R}\) itself.
%
Hence the main data structure of our algorithm are simply \(\RelationClass\)-m-eVASS%
\footnote{
Notably, this is simpler than standard KLM sequences.
},
abbreviated below as \emph{m-eVASS}. 

\begin{algorithm}[h!]
\caption{Approximation(\(\vectSet{X}=\Rel(\VAS, \qin, \qfin)\))}\label{AlgorithmMainStructure}
\begin{algorithmic}
\State \textbf{Output}: Set of Pairs \((\vectSet{X}_j, \vectSet{L}_j)\) of reachability relations \(\vectSet{X}_j=\Rel(\VAS_j, \qinj, \qfinj)\) of m-eVASS \(\VAS_j\) and hybridlinear \(\vectSet{L}_j\) such that \(\vectSet{X}= \bigcup_{j=1}^k \vectSet{X}_j\) and \(\vectSet{X}_j \HybridizationRelation \vectSet{L}_j\) for all \(j\).
\State \textbf{Initialize}: Workset \(\gets\) \(\{\vectSet{X}\}\).
\While{exists \(\vectSet{X} \in\) Workset: \(\vectSet{X}\) is not perfect}
    \State Workset \(\gets\) (Workset \(\setminus \{\vectSet{X}\})\ \cup\) Decompose(\(\vectSet{X}\)).
\EndWhile
\State \textbf{Return}\!\! \(\{(\vectSet{X}, \pi_{\source_\vectSet X, \target_\vectSet X}(\sol(\CharSys_\vectSet{X})) \mid \vectSet{X} \in \text{Workset}\}\).
\end{algorithmic}
\end{algorithm}

%This is however not actually an extension of the class of sections: Using the projection in the definition of sections, every composition of sections can be written as a single section.

The main idea of the algorithm is the following (see Algorithm \ref{AlgorithmMainStructure}): Maintain a workset of current m-eVASS. While one of them does not fulfill a condition we call ``perfectness'', decompose such an m-eVASS into finitely many m-eVASS which are ``smaller'' w.r.t. some well-founded ordering. If every m-eVASS is perfect, output each of them ($\vectSet X$) together with the projection of solutions of the characteristic integer linear program ($\CharSys_\vectSet X$)  
to specific variables \(\source_\vectSet X\) and \(\target_\vectSet X\).

The rest of this section is structured as follows: First we set up some simplifying assumptions. Then we define
 the characteristic system, and then using the system we define
 perfectness. Fourthly, we define the rank of an m-eVASS. Then we describe how to decide perfectness and decompose. We end with a correctness proof and complexity analysis.

\para{Setup}: For the rest of this section fix an m-eVASS \((\VAS, \qin, \qfin)\).
As the class $\RelationClass$ is approximable in $ \mathfrak{F}_{\alpha}$,
there is a function $g\in \mathfrak{F}_{\alpha}$ and an algorithm \(\ApproximationAlgorithm\) that 
given \(\vectSet{R} \in \RelationClass\), in time \(\mathfrak{F}_{\alpha}\) outputs finitely many \(\vectSet{R}_j, \vectSet{L}_j\) s.t. \(\vectSet{R}=\bigcup_{j=1}^k \vectSet{R}_j\) with \(\vectSet{R}_j \HybridizationRelation_g \vectSet{L}_j\).
We start with two basic assumptions: 

\begin{enumerate}
\item[\COne] W.l.o.g. the SCC graph of m-eVASS is a line.
\item[\CTwo] W.l.o.g. for every edge \(e\) labelled with $\vectSet{R}(e)$, we have a linear asymptotic overapproximation \(\vectSet{R}(e) \HybridizationRelation_g \vectSet{L}(e)=\vect{b}(e)+\N(\vectSet{F}(e))\) which is \emph{basic}, i.e., \(\vect{b}(e) \in \vectSet{R}(e)\). 
%We call such an overapproximation \emph{basic}.
\end{enumerate}

To have \COne, in exponential time we can split any graph into exponentially many line graphs, preserving the set of paths. 

Concerning \CTwo, 
%
%let \(\ApproximationAlgorithm\) be an algorithm which given \(\vectSet{R} \in \RelationClass\), in time \(\mathfrak{F}_{\alpha}\) outputs finitely many \(\vectSet{R}_j, \vectSet{L}_j\) s.t. \(\vectSet{R}=\bigcup_{j=1}^k \vectSet{R}_j\) with \(\vectSet{R}_j \HybridizationRelation_g \vectSet{L}_j\), which exists since \(\RelationClass\) is approximable.
%
suppose that some edge \(e\) labelled with $\vectSet{R}(e)$ does not fulfill \CTwo. 
%Observe first that every edge between SCCs which adds counters \(I_+(e)\) or respectively removes counters \(I_-(e)\) has linear semantics, hence \CTwo \ trivially holds for such edges. The only counter examples can be edges \(e\) labelled with \(\vectSet{R}(e) \in \RelationClass\). 
%
We apply \(\ApproximationAlgorithm\) 
on \(\vectSet{R}(e)\), replacing \(e\) by edges \(e_1, \dots, e_k\) with the same source/target and \(\vectSet{R}(e_j):=\vectSet{R}_j\) and \(\vectSet{L}(e_j):=\vectSet{L}_j=\vect{b}_j+\N(\vectSet{F}_j)\). Here we already assumed that \(\vectSet{B}_j=\{\vect{b}_j\}\) is a singleton, since by Lemma \ref{LemmaShiftGoodOverapproximation}, whenever \(\vectSet{X} \HybridizationRelation_g \vectSet{B}_j+\N(\vectSet{F}_j)\), then for all \(\vect{b}_j \in \vectSet{B}_j\) we have also \(\vectSet{X} \cap (\vect{b}_j+\N(\vectSet{F}_j)) \HybridizationRelation_g \vect{b}_j+\N(\vectSet{F}_j)\). Namely, the intersection \((\vect{b}_j+\N(\vectSet{F}_j)) \cap (\vectSet{B}_j+\N(\vectSet{F}_j))\) is non-degenerate. 
Due to monotonicity-preservation in Definition \ref{DefinitionApproximable},
As \(\vectSet{R}(e)\) was monotone, each of \(\vectSet{R}_j\) is also monotone.
It remains to ensure that the overapproximations
\(\vectSet{R}(e_j) \HybridizationRelation_g \vectSet{L}(e_j)\) are basic.

To obtain this, choose \(\vect{w}_j:=\sum_{\vect{f} \in \vectSet{F}_j} \vect{f}\) and \(\vect{p}_j:=g(\size(\vectSet{X}_j)+\size(\vect{b}_j)+\size(\vect{w}_j)) \cdot \vect{w}_j\) and set \(\vect{b}_j':=\vect{b}_j+\vect{p}_j\). By Lemma \ref{LemmaShiftGoodOverapproximation}, \(\vectSet{L}_j':=\vect{b}_j'+\N(\vectSet{F}_j)\) is an asymptotic overapproximation of \(\vectSet{R}_j \cap \vectSet{L}_j'\), and by definition of the function \(g\) we have \(\vect{b}_j' \in \vectSet{R}_j=\vectSet{R}(e_j)\). It remains to overapproximate also \(\vectSet{R}_j \cap (\vectSet{L}_j \setminus (\vect{p}_j+\vectSet{L}_j))\): By Lemma \ref{LemmaDimensionDecreaseShiftedL} \(\dim(\vectSet{L}_j \setminus (\vect{p}_j+\vectSet{L}_j))< \dim(\vectSet{L}_j)\). Hence we finish by recursively applying  \(\ApproximationAlgorithm\).

\CTwo \ is running in \(\mathfrak{F}_{\alpha+1}\): We have \(g \in \mathfrak{F}_{\alpha}\), and recursion depth is at most the number of counters \(n\), i.e.\ linear.

After achieving \CTwo \ we may need to recompute \COne: Since edges may vanish during \CTwo,
which happens when \(\ApproximationAlgorithm\) yields the empty list of overapproximations,
SCCs may split further.

\begin{example}
Suppose \(\vectSet{R}(e)=\{(x,y) \in \N \times \N \mid x \geq 1, y \leq x^2\}\),
similarly to Figure \ref{FigureIntuitionSemilinearityAlgorithm}. 
\(\ApproximationAlgorithm\) might at first output \(\vectSet{L}_1(e)=\N^2\). However, \((0,0) \not \in \vectSet{R}(e)\). Instead the algorithm computes \(g(\dots)\) (say \(1\)), adds \(\vectSet{L}_1'(e):=(1,1)+\N^2\) as an overapproximation and recursively overapproximates \(\vectSet{R}(e) \cap (\{0\} \times \N\)) and \(\vectSet{R}(e) \cap (\N \times \{0\})\) to further add \(\vectSet{L}_2'(e)=(\N+1) \times \{0\}\).
\end{example}

\vspace{-1mm}
\para{Notation}: We write \((\VAS_i=(Q_i, E_i), \qini, \qfini)\) for the \(i\)-th SCC of \((\VAS, \qin, \qfin)\) \emph{including trivial SCCs} (those with no edges), for $i = 1, \ldots, r$. 
We let \(e_i\) be the unique edge from the \(i\)-th to the \((i+1)\)-st SCC, and \(n_i\) to be the number of counters of \(\VAS_i\).
The edges $e_1, \ldots, e_{r-1}$ we call \emph{bridges}.

\para{Characteristic System}: We define an integer linear program (ILP) whose goal is to overapproximate reachability in an m-eVASS. It requires access to the overapproximations \(\vectSet{L}(e)\) as in \CTwo. It has three types of nonnegative integer variables: 
\begin{itemize}
\item
Vectors of variables \(\vect{x}_1, \dots, \vect{x}_r\) and \(\vect{y}_1, \dots, \vect{y}_r\) of appropriate length,
namely  \(\vect{x}_i, \) and \(\vect{y}_i\) have length  \(n_i\). 
These stand for the source/target configurations of \(\VAS_i\). 
\item
\emph{Edge variables} \(\#(e)\) for every edge \(e\) \emph{inside an SCC} that
count how often the edge $e$ is used.
\item
\emph{Period variables} \(\#(\vect{p})\) for every period $\vect p$  of every linear set \(\vectSet{L}(e)\) 
that count how often $\vect p$ is used.
Depending on $e$, period variables split into bridge and SCC ones.
\end{itemize}
The \(\pi_{\source, \target}\) in Algorithm \ref{AlgorithmMainStructure} denotes the projection to \(\vect{x}_1, \vect{y}_r\).

\smallskip

For every \(i\) we will define an ILP \(\LocalCharSys_i\) which overapproximates reachability in \(\VAS_i\). 
The full ILP is then defined in terms of these local ILP as follows:
\begin{align}\label{eq:CharSys}
\CharSys:=\bigwedge_{i=1}^r \LocalCharSys_i \wedge \bigwedge_{i=1}^{r-1} (\vect{y}_i, \vect{x}_{i+1}) \in \vectSet{L}(e_i).
\end{align}
\vspace{-3mm}
\begin{remark}
Above, we write ``\(\vect{x} \in \vectSet{L}\)'', where \(\vect{x}\) 
is a vector of variables and \(\vectSet{L}=\vect{b}+\N(\{\vect{p}_1, \dots, \vect{p}_\ell\})\) is a linear set/relation. 
The expanded technical meaning is that we add the set of equations
\(\vect{x}-\sum_{j=1}^\ell \#(\vect{p}_j) \cdot \vect{p}_j=\vect{b}\).
We will frequently use this notation later on, implicitly assuming additional auxiliary period variables when necessary. 
\end{remark}

The first step is of definition of \(\LocalCharSys_i\) is
to overapproximate the set of paths in the control graph via the so-called \emph{Euler-Kirchhoff-equations}.
%, defined as follows:
%
%\begin{definition}
%Let \(\VAS=(Q,E)\) be a connected graph with initial/final states \(\qin, \qfin\). 
Let $\text{EK}_i$ %(\(\VAS\)) is 
denote the following ILP with edge variables \(\#(e)\) for every \(e \in E_i\):

\[\bigwedge_{q \in Q_i}\ \sum_{e \in \text{in}(q)} \#(e)-\sum_{e \in \text{out}(q)} \#(e)=\mathds{1}_{q=\qfini}(q)-\mathds{1}_{q=\qini}(q),\] where in(\(q\)) and out(\(q\)) are the in- and outgoing neighbourhoods of state \(q\), respectively, and \(\mathds{1}_S \colon Q_i \to \{0,1\}\) for a condition \(S\) is the indicator function, i.e.\ \(\mathds{1}_S(q)=1\) if \(q\) fulfills \(S\), and \(\mathds{1}_S(q)=0\) otherwise.

We write the corresponding homogeneous ILP, where we replace the RHS
\(\mathds{1}_{q=\qfini}(q)-\mathds{1}_{q=\qini}(q)\) by \(0\), as \(\text{HEK}_i\). %(\(\VAS\)).
%\end{definition}

Intuitively, \(\#(e)\) is the number of times an edge \(e\) is used, and for every state which is not initial or final, we require it to be entered as often as left, the final state has to be entered one more time than left etc. This leads to Euler's lemma:

\begin{lemma} \label{LemmaBasicEulerKirchhoff}
%Let \(\VAS=(Q,E)\) be a connected graph with initial/final states \(\qin, \qfin\). 
Let \(\pi\) be a path from \(\qini\) to \(\qfini\). Then \(\Parikh(\pi)\), Parikh image of \(\pi\), fulfills 
$\text{EK}_i$. Conversely, if \(\vect{w} \in \N_{\geq 1}^E\) fulfills $\text{EK}_i$ then there is a path \(\pi\) with \(\Parikh(\pi)=\vect{w}\). Observe that the converse requires \(\vect{w}[e]\geq 1\) for every \(e\).
Analogous connection holds between cycles in $\VAS_i$ and solutions of $\text{HEK}_i$.
\end{lemma}

Now we can continue defining \(\LocalCharSys\) for an SCC \(\VAS_i\). Remember that for \(\vect{b}=(\vect{b}_s, \vect{b}_t) \in \N^{n_i} \times \N^{n_i}\) we write \(\Effect(\vect{b}):=\vect{b}_t-\vect{b}_s\) for the effect, and that \(\vectSet{L}(e)=\vect{b}(e)+\N(\vectSet{F}(e))\) is an asymptotic overapproximation of \(\vectSet{R}(e)\). 

Let \(\vectSet{F}(\VAS_i):=\bigcup_{e \in E_i} \vectSet{F}(e) \subseteq \N^{n_i} \times \N^{n_i}\). For every \(\vect{w}_i \in \N^{E_i}\) define \(\Effect(\vect{w}_i):=\sum_{e \in E_i}\vect{w}_i(e) \cdot \Effect(\vect{b}(e))\), i.e.\ the effect of an edge sequence only taking the bases \(\vect{b}(e)\) into account. \[\LocalCharSys_i:=\text{EK}_i(\vect{w}_i) \wedge \vect{y}_i-\vect{x}_i \in \Effect(\vect{w}_i) + \Effect(\N(\vectSet{F}(\VAS_i))).\]

Let us give intuition on this overapproximation: Assume first that the overapproximations \(\vectSet{L}(e)\) of the edges have no periods, i.e.\ that \(\vectSet{F}(\VAS_i)=\emptyset\). Then the overapproximation says the following: The effect \(\vect{y}_i-\vect{x}_i\) of the pair \((\vect{x}_i, \vect{y}_i)\) has to be the effect of a Parikh vector \(\vect{w}_i\) fulfilling $\text{EK}_i$. Usually this is called \(\Z\)-VASS-reachability. With periods, we require that \(\vect{y}_i-\vect{x}_i-\Effect(\vect{w}_i)\) is the effect of some periods. I.e.\ the \emph{effect} which we did not manage to produce using normal \(\Z\)-reachability may be compensated by using periods arbitrarily often. The system does not distinguish between which linear relation the necessary periods would be coming from. 

\para{Perfectness:}
We write $\sol(\CharSys)$ for the set of solutions of $\CharSys$. A variable is called \emph{unbounded}
if there are solutions assigning arbitrary large values to this variable.

\begin{definition}[\textbf{Perfectness}] An m-eVASS \((\VAS, \qin, \qfin)\) satisfying \COne-\CTwo
\ is \emph{perfect} if:\label{DefinitionPerfectness}%

\begin{enumerate}
%\item[P1] Every \(\vectSet{R}(e) \in \RelationClass\) occurring on an edge \(e\) of some \(\VAS_i\) has a nice overapproximation \(\vectSet{L}(e)=\vect{b}(e)+\N(\vectSet{F}(e))\) \emph{fulfilling} \(\vect{b} \in \vectSet{R}(e)\).
%\item[P2] For every \(i\), \(\VAS_i\) is either \emph{strongly-connected}, or \(\VAS_i\) fulfills \(\qini \neq \qfini\) and has only one edge \(e_i=(\qini,\qfini)\).% In the second case, remember that the label of \(e_i\) may be either \(\vectSet{R}(e_i) \in \RelationClass\), or \(I_+(e_i)\) and \(I_-(e_i)\).
%\item[P3] For all \(i\) the relation \(\vectSet{L}_i\) is the projection of the solutions of \(\CharSys\) to the source/target variables for \(\VAS_i\).
\item[\WeakPOne] In \(\CharSys\), every bridge period variable is unbounded.
\item[\WeakPThree] If a source variable \(x_{i,j}\) or target variable \(y_{i,j}\) is bounded, then all solutions of \(\CharSys\) agree on the value of this variable.
\item[\WeakPTwo] For every non-trivial \(\VAS_i\) and counter \(j\) s.t.~the source variable \(x_{i,j}\) and target variable \(y_{i,j}\) are bounded, \(\VAS_i\) has a cycle with non-zero effect on \(j\).
\item[\POne] In \(\CharSys\), every SCC period variable and every edge variable is unbounded.
\item[\PTwo] \emph{Every} configuration of every \(\VAS_i\) can be forwards- and backwards-covered, i.e.\ for every \(i\) and every configuration \(p(\vect{x})\) of \(\VAS_i\), we have 
\[\Rel(\VAS_i,\qini, p) \cap (\pi_{\vect x_i}(\sol(\CharSys))) \times (\vect{x}+\N^{n_i})) \neq \emptyset,\]\[\Rel(\VAS_i, p, \qfini) \cap ((\vect{x}+\N^{n_i}) \times \pi_{\vect y_i}(\sol(\CharSys))) \neq \emptyset.\]
\end{enumerate}
\end{definition}

%\begin{lemma}
%Any full support solution of \(\LocalCharSys_i\) gives rise to a \(\Z\)-run from \(\qini\) to \(\qfini\).
%\end{lemma}
\para{Recap}: We explained the algorithm structure, defined the characteristic system and perfectness. Next we define the rank, followed by deciding perfectness and the decomposition.

\begin{definition}[\textbf{Rank}] The \(\rank(\VAS) \in \N^{d+1}\) of an m-eVASS is similar to \cite{LerouxS19} for VASS. Recall that the dimension of an m-eVASS is defined using the vector space spanned by cycle effects. For every \(1\leq i\leq r\), if \(\VAS_i\) is a non-trivial SCC, we define \(\vect{r}_i \in \N^{d+1}\) as \(\vect{r}_i[\dim(\VAS_i)]=|Q_i|\) and \(\vect{r}_i[j]=0\) otherwise. For trivial SCCs we have \(\vect{r}_i=0^{d+1}\). Finally, \(\rank(\VAS)=\sum_{i=1}^r \vect{r}_i\) is the sum of the ranks of all SCCs.

We order ranks using the \emph{lexicographic ordering} on \(\N^{d+1}\).
\end{definition}

The intuition on the rank is as follows: If we replace a ``high-dimensional'' SCC by an arbitrary number of ``lower-dimensional'' SCCs, then the rank decreases. For technical reasons, SCCs are weighted with the number of states.

\para{Deciding Perfectness and Decomposition}: We explain how to decide perfectness (Def. \ref{DefinitionPerfectness}) and decompose otherwise.

Following Leroux \cite{LerouxS19}, we split the properties of perfectness  into two categories: \emph{cleaning properties}  \COne-\CTwo, \WeakPAll, and \emph{actual properties} \POne \ and \PTwo. The difference is that the decomposition for a cleaning property \emph{does not decrease the rank}. In exchange, these properties essentially do not add new states or edges to SCCs (except replacing an edge by a number of parallel edges),
and hence by repetitively doing the corresponding decomposition steps 
we arrive at an m-eVASS fulfilling \emph{all of} \COne-\CTwo, \WeakPAll.
The termination argument does not refer thus to rank. 
%
%\emph{do not harm each other} in the following sense: In our case \COne-\CTwo, \WeakPOne-\WeakPThree \ are cleaning properties, and if we do the corresponding decomposition steps in sequence, we arrive at an m-eVASS fulfilling \emph{all of} \COne-\CTwo, \WeakPOne-\WeakPThree. In particular, when we decompose for \CTwo, we do not lose \COne, when we decompose for \WeakPOne \ we do not lose \COne, \CTwo, etc. 

This is in contrast to \POne \ and \PTwo: When a transition is bounded (may only be used finitely often) and hence \POne \ does not hold, we will remove this transition from SCCs. Removing transitions might cause certain configurations to no longer be coverable in the SCC, and property \PTwo \ is lost. Conversely if a counter is bounded in \PTwo \ and we store it in the control state, then the characteristic system now has new variables for the new edges, etc. and some of these might be bounded.

With this new knowledge, consider again Algorithm \ref{AlgorithmMainStructure}. We can now describe what Decompose(\(\vectSet{X}\)) does: It first guarantees properties \COne-\CTwo, \WeakPAll \ simultaneously since they are cleaning properties, and afterwards applies one actual decomposition for property \POne \ or \PTwo. This way, even though \COne-\CTwo, \WeakPAll \ did not decrease the rank, the whole procedure Decompose did. This leads to the following two lemmas:

\begin{lemma}[Cleaning] \label{LemmaCleaning}
Let \((\VAS, \qin, \qfin)\) be an m-eVASS. Then one can compute in \(\mathfrak{F}_{\alpha+2}\) a finite set of m-eVASS \((\VAS_1, q_{\text{in}}^{1}, q_{\text{fin}}^{1}), \dots, (\VAS_k, q_{\text{in}}^{k}, q_{\text{fin}}^{k})\) s.t. \(\rank(\VAS_j) \leq \rank(\VAS)\), \(\Rel(\VAS, \qin, \qfin)=\bigcup_{j=1}^k \Rel(\VAS_j, \qinj, \qfinj)\), and every \((\VAS_j, \qinj, \qfinj)\) fulfills properties \COne-\CTwo, \WeakPAll.
\end{lemma}

\begin{lemma}[Decomposition] \label{LemmaActualDecomposition}
Let \((\VAS, \qin, \qfin)\) be an m-eVASS fulfilling \COne-\CTwo, \WeakPAll \ which is not perfect. Then one can in time \(\mathfrak{F}_{\alpha+d+1}(\size(\VAS))\) compute a finite set of m-eVASS \((\VAS_1, q_{\text{in,1}}, q_{\text{fin,1}}), \dots, (\VAS_k, q_{\text{in,k}}, q_{\text{fin,k}})\) s.t. \(\rank(\VAS_j) < \rank(\VAS)\) and \(\Rel(\VAS, \qin, \qfin)=\bigcup_{j=1}^k \Rel(\VAS_j, \qinj, \qfinj)\).
\end{lemma}

%In the following we prove Lemma \ref{LemmaCleaning} and Lemma \ref{LemmaActualDecomposition}.

\begin{proof}[\textbf{Proof of Lemma \ref{LemmaCleaning}}] 
The proof proceeds by achieving  \COne-\CTwo \ 
(as explained above), \WeakPAll \ respectively in sequence.

\para{\WeakPOne:} In exponential time we can check boundedness of the variables in \(\CharSys\).

Let \(e_i\) be a bridge edge where \WeakPOne \ does not hold. Write \(\vectSet{L}(e_i)=\vect{b}(e_i)+\N(\vectSet{F}(e_i))\), and let \(\vectSet{F}'(e_i) \subsetneq \vectSet{F}(e_i)\) be the set of unbounded periods. Let \(\vectSet{B}_i\) be the set of possible values of the bounded periods. Essentially, we will replace \(\vectSet{L}(e_i)\) by 
\[\vectSet{L}'(e_i):=(\vect{b}(e_i)+\vectSet{B}_i \cdot (\vectSet{F}(e_i) \setminus \vectSet{F}'(e_i)))+\N(\vectSet{F}'(e_i)),\] where \(\cdot\) denotes matrix multiplication. Matrix multiplication is a succinct way to express that the (bounded) periods \(\vect{p} \in \vectSet{F}(e_i) \setminus \vectSet{F}'(e_i)\) should be used as often as prescribed by \(\vectSet{B}_i\). There is however a difficulty: If \(\vectSet{L}'(e_i) \cap \vectSet{L}(e_i)\) is degenerate, then \((\vectSet{R}(e_i)\cap \vectSet{L}'(e_i)) \HybridizationRelation \vectSet{L}'(e_i)\) might not hold. In this case we 
take $\vectSet R(e_i) := \vectSet R(e_i) \cap \vectSet L'(e_i)$ and recompute $\vectSet L(e_i)$, as in
\CTwo. In this case  \(\CharSys\) changes to accommodate the newly introduced overapproximation, and again some periods might be bounded, and we have to repeat the decomposition for \WeakPOne. Observe however the very important fact that the degenerate intersection guarantees \(\dim(\vectSet{L}'(e_i))<\dim(\vectSet{L}(e_i))\), which ensures that after at most linearly many repetitions, \WeakPOne \ holds. 

\begin{example}
Let \(\vectSet{R}(e_i)=\{(x,y) \mid y \leq x^2\} \HybridizationRelation \N^2=\N(\{(1,0), (0,1)\})\) as in the right of Figure \ref{FigureIntuitionSemilinearityAlgorithm}. Assume that \((1,0)\) is bounded, with its only value being \(1\). Then we replace \(\N^2\) by \(\vectSet{L}'(e_i)=(1,0)+\N(\{(0,1)\})\). However, clearly \(\vectSet{R}(e_i) \cap \vectSet{L}'(e_i)\) is finite, in particular we do not have \((\vectSet{R}(e_i) \cap \vectSet{L}'(e_i)) \HybridizationRelation \vectSet{L}'(e_i)\). We have to re-approximate \(\vectSet{R}(e_i) \cap \vectSet{L}'(e_i)\), replacing it by the finitely many points. In the worst case, these points might again be bounded in the global ILP, and we might have to repeat the process. However, the dimension decreases, and the recursion therefore terminates.
\end{example}
\vspace{-2mm}

\para{\WeakPThree:} Assume that some variable \(y_{i,j}\) (\(x_{i,j}\) is symmetric) is bounded and has the set of possible values \(F\). For every \(f\in F\), we create from $\VAS$ a new m-eVASS \(\VAS_f\) 
by adding just after $\VAS_i$ the following sequence of four bridge edges: 
\[\text{Sub}_{j,f}  \ ; \ \text{Del}_j \ ; \  \text{RevDel}_j \ ; \  \text{Add}_{j,f},\] 
where
$\text{Sub}_{j,f}$
subtracts \(f\) from counter \(j\), $\text{Del}_j$ deletes counter \(j\) (which now correctly has value \(0\)), 
$\text{RevDel}_j$ readds counter \(j\), and $\text{Add}_{j,f}$ sets counter \(j\) from its current value \(0\) to \(f\).
This essentially tests  $y_{i,j}=f$.
%Here, \(Sub_f\) subtracts \(f\) from counter \(j\), \(Del_j\) deletes counter \(j\) which is now \(0\), \(RevDel_j\) readds counter \(j\) and \(Add_f\) again sets counter \(j\) to \(f\) by addition.
As a result $\CharSys$ changes, and we may need to repeat the decomposition for \WeakPOne.

\para{\WeakPTwo:} To check whether there is a cycle with non-zero effect on counter \(j\), first we check whether some period of an edge \(e\) changes \(j\). If yes, then clearly such a cycle exists. Otherwise, for the exponentially many possible supports \(S \subseteq E_i\) a cycle could have, and \(\sim \in \{>,<\}\),  we write down an ILP as follows: 
\[\Big(\bigwedge_{e \in S} \vect{w}(e)\geq 1\Big) \wedge \text{HEK}_i(\vect{w})\wedge \Delta(\vect{w})(j) \sim 0.\]
If one of these ILPs has a solution, then a required cycle exists.
Otherwise, the value of counter $j$ can be uniquely determined from the state of the automaton. 
Suppose \WeakPThree \ holds.
In particular there is a unique exit value \(out\) at \(\qfini\), since the entry value \(in\) at \(\qini\) is bounded by assumption,
and hence unique by \WeakPThree. 
%\sla{\WeakPThree \ is used/assumed here}
We replace \(\VAS_i\) by \(\VAS_i'\),
preceded and succeded by the following bridge edges
\[\text{Sub}_{j,in} \ ; \  \text{Del}_j \ ; \  \VAS_i' \ ; \ \text{RevDel}_j \ ; \ \text{Add}_{j,out}\] 
that essentially test  $x_{i,j}=in$, $y_{i,j}=out$.
%where 
%$\text{Sub}_{j,in}$
%subtracts \(in\) from counter \(j\), $\text{Del}_j$ deletes counter \(j\) (which now correctly has value \(0\)), 
%$\text{RevDel}_j$ readds counter \(j\), and $\text{Add}_{j,out}$ sets counter \(j\) from its current value \(0\) to \(out\).
Here, \(\VAS_i'\) is obtained from \(\VAS_i\) by deleting counter \(j\) and \emph{reapproximating} edges using \CTwo, to accommodate the fact that fixing the value of \(j\) might have modified relations $\vectSet R(e)$.
Again, $\CharSys$ changes and we may need to repeat the decomposition for \WeakPOne, \WeakPThree.

%\vspace{-1mm}
\para{Termination:}
The claimed 
complexity bound \(\mathfrak{F}_{\alpha+2}\) is shown as follows.
We invoke the decompositions iteratively in the above specified order, 
namely invoke \WeakPThree \ only if \WeakPOne \ already holds,
and \WeakPTwo \ only if both \WeakPOne \ and \WeakPThree \ hold.
\WeakPOne \ terminates in \(\mathfrak{F}_{\alpha+1}\), as argued above.
Invocations of \WeakPThree \ decomposition eventually terminate, as it reduces the number of source/target variables of non-trivial SCCs which do not have a unique solution in $\CharSys$, while
all the decompositions only introduce new trivial SCCs, or split an existing SCC into smaller ones
(in \COne-\CTwo, which can clearly only happen linearly many times), or introduce bridge edges that satisfy \WeakPOne.
Likewise, invocations of \WeakPTwo \ decomposition also terminate, as it deletes one counter from a non-trivial SCC.
The number of iterations is thus elementary, and 
since \COne-\CTwo \ are in \(\mathfrak{F}_{\alpha+1}\), the whole
\WeakPAll \ decomposition terminates in \(\mathfrak{F}_{\alpha+2}\). 
\end{proof}

%\textbf{Recap}: We explained the structure of the workset algorithm, defined perfectness, the characteristic system and the rank and have finished the first half of the decomposition steps by proving Lemma \ref{LemmaCleaning}. It remains to prove Lemma \ref{LemmaActualDecomposition} and explain correctness and complexity.

\begin{proof}[\textbf{Proof of Lemma \ref{LemmaActualDecomposition} (sketch)}] We have to explain how to decide properties \POne \ and \PTwo \ of Definition \ref{DefinitionPerfectness}, and how to decompose if they do not hold. We assume \COne-\CTwo, \WeakPAll.

\para{\POne:}
In exponential time we can check boundedness of the variables in \(\CharSys\).
It remains to decompose. 

%Sketch of Construction: 
Let \((\VAS_i, \qini, \qfini)\) be the SCC where some edge variable or period variable is bounded. If \(\#(e)\) is bounded by some \(n\), 
then we store in the control state how often $e$ was used already.
This amounts to creating \(n+1\) copies of the SCC without $e$, and redirecting \(e\) to point to the next copy.
% (in particular it is no longer in an SCC). %, vanishing from the rank). 
If a period \(\vect{p}\) of some linear relation \(\vectSet{L}(e)\) is bounded, then we store in the control state how often we used this period already, up to the bound, and remove $\vect p$ from periods of $e$. For all \(i \leq j \leq B\) where \(B\) is the bound, we have an edge from the \(i\)-th copy to the \(j\)-th which is $e$ restricted 
to use the period exactly \(j-i\) many times. What makes this decomposition step tricky is that \emph{every} bounded variable of \(\VAS_i\) has to be removed \emph{at once}, otherwise the rank does not decrease.

\para{\PTwo:} %Let us rephrase the two parts of property \PTwo: 
Since the second part of \PTwo \ corresponds to the first one if we turn around all edges, we w.l.o.g. only consider the first part. It can be reformulated as follows:

\begin{restatable}{lemma}{LemmaRestatableEquivalentPumpingCondition} \label{LemmaEquivalentPropertySix}
Let \(I \subseteq \{1,\dots, n_i\}\) be the set of counters whose source variable \(x_{i,j}\) is bounded. Let \(\pi \colon \N^{n_i} \to \N^{I}\) be the projection removing counters outside \(I\). Consider the m-eVASS \(\pi(\VAS_i)\) obtained from \(\VAS_i\) by deleting counters outside \(I\). Let \(\vect{b}'\) be the unique (since \WeakPThree \ holds) vector of values which solutions of \(\CharSys\) assign to these bounded \(x_{i,j}\).

The first part of \PTwo \ holds \(\iff\) \((\vect{b}', \vect{b}'+\mathbf{up}) \in \Rel(\pi(\VAS_i), \qini, \qini)\) for some vector \(\mathbf{up} \in \N_{\geq 1}^{I}\).
\end{restatable}

\begin{proof}
By standard techniques for VASS.
\end{proof}

To decide the property in Lemma \ref{LemmaEquivalentPropertySix} we use the well-known \emph{backwards coverability algorithm}. It works as follows: We maintain finite sets \(\vectSet{B}^q \subseteq \N^I\), one per
state $q\in Q_i$, s.t. we know for all configurations \(\vect{x} \in \vectSet{B}+\N^I\), there exists \(\vect{y} \geq \vect{b}+1^I\) with \((\vect{x}, \vect{y}) \in \Rel(\pi(\VAS_i), q, \qini)\). Initially we can choose \(\vectSet{B}^{\qini}:=\{\vect{b}+1^I\}\), since we can choose \(\vect{x}=\vect{y} \geq \vect{b}+1^I,\) and trivially \((\vect{x}, \vect{x}) \in \Rel(\pi(\VAS_i), \qini, \qini)\). 
For $q\neq \qini$ we put initially \(\vectSet{B}^{q}:=\emptyset\).

Then in a loop we enlarge the sets \(\vectSet{B}^q\) by ``applying transitions backwards'': If we find \(\vect{x} \not \in \vectSet{B}^p+\N^I\) s.t. \(\vect{x} \to_e \vectSet{B}^q+\N^{I}\), for $e=(p,q)\in E_i$, then we can set \(\vectSet{B}^p := \vectSet{B}^p \cup \{\vect{x}\}\). In order to apply a transition \(e\) backwards, as in \CTwo \ compute \emph{basic} asymptotic overapproximations \(\vectSet{L}_1, \dots, \vectSet{L}_k\) of \((\N^{I} \times (\vectSet{B}_i+\N^{I})) \cap \vectSet{R}(e)\) and use their bases.

The backwards coverability algorithm has an easy termination guarantee by Proposition \ref{PropositionFastGrowingComplexity}: Since we only add vectors which are not larger than a vector we already had, we produce a sequence without an increasing pair. Applying a transition backwards takes \(\mathfrak{F}_{\alpha+1}\) time similar to \CTwo. Strictly speaking the vectors are in \(\N^I\), i.e.\ have potentially up to \(n_i\) entries, leading to \(\mathfrak{F}_{\alpha+n_i+1}\) by Proposition \ref{PropositionFastGrowingComplexity}, but using a slightly cleverer analysis we can bound the complexity by \(\mathfrak{F}_{\alpha+d+1}\), i.e.\ we can show that the dimension of \(\VAS_i\) suffices.

This procedure has the additional advantage that if \PTwo \ does not hold, then we obtain a bound \(B \in \N\) on some counter \(j \in I\) in \(\pi(\VAS_i)\). The bound \(B\) continues to hold if we readd the projected counters, i.e.\ also in \(\VAS_i\) counter \(j\) is bounded by \(B\), and we decompose by deleting counter \(j\) and instead storing it in the control state.
\end{proof}

\para{Recap:} We finished the proof of Lemma \ref{LemmaActualDecomposition}, which is the last decomposition step. Hence perfectness is now decidable, and Algorithm \ref{AlgorithmMainStructure} can be implemented. It remains to give a bound on the running time and explain correctness.

\para{Complexity:} By Lemma \ref{LemmaCleaning} and Lemma \ref{LemmaActualDecomposition}, a single loop iteration is doable in \(\mathfrak{F}_{\alpha+d+1}\). Furthermore, the rank \(\in \N^{d+1}\) decreases lexicographically. Hence we are in the setting of Proposition \ref{PropositionFastGrowingComplexity}: The sequence of ranks is a \((f,m)\)-controlled sequence  in \(\N^{d+1}\) without an increasing pair for some \(f \in \mathfrak{F}_{\alpha+d+1}\), hence it has length at most \(\mathfrak{F}_{\alpha+2d+2}\).

\para{Correctness:} We have to show that \(\Rel(\VAS, \qin, \qfin)\) has \(\pi_{\vect{x}_1, \vect{y}_r}(\sol(\CharSys))\) as an  overapproximation, namely
%
%\begin{lemma}%[Correctness]
$\Rel(\VAS, \qin, \qfin) \HybridizationRelation_g \pi_{\vect{x}_1, \vect{y}_r}(\sol(\CharSys))$
for some \(g \in \mathfrak{F}_{\alpha+2d+2}\),
which follows by Theorem \ref{TheoremCorrectness} stated below, combined with 
Lemma \ref{LemmaNiceOverapproximationProjection}.
%\end{lemma}

% !TEX root = Main.tex

Let \(\rho\) be a run of \(\VAS\) from \(\qin\) to \(\qfin\). It induces a solution \(\sol(\rho)\) of \(\CharSys\) by letting \(\sol(\rho)[\vect{x}_i]\) equal the configuration at which \(\rho\) enters the SCC \(\VAS_i\), accordingly for \(\sol(\rho)[\vect{y}_i]\). For every edge counter variable \(\#(e)\) we set \(\sol(\rho)[\#(e)]=\) number of times \(e\) is used, accordingly \(\sol(\rho)[\#(\vect{p})]\) is the \emph{total} number of times period \(\vect{p}\) is used, across all uses of the edge \(e\) which \(\vect{p}\) belongs to.

We call a solution \(\vect{s}\in \sol(\CharSys)\) \emph{concretizable} if there is a run \(\rho\) of \(\VAS\) from \(\qin\) to \(\qfin\) whose induced solution \(\sol(\rho)\) of \(\CharSys\) equals \(\vect{s}\). We let \(\Concretizable(\VAS, \qin, \qfin)\) be the set of concretizable solutions. Correctness states the following.

\begin{theorem} 
If \((\VAS, \qin, \qfin)\) is perfect, then \(\Concretizable(\VAS, \qin, \qfin) \HybridizationRelation_g \sol(\CharSys)\) for some \(g \in \mathfrak{F}_{\alpha+2d+2}\). \label{TheoremCorrectness}
\end{theorem}

We reprove a property of VASS: An \(n\)-fold repetition of a cycle is enabled iff the first and last repetition of the cycle are.
%Let us first explain why Lemma \ref{LemmaLocalCorrectness} is enough. Algorithm \ref{AlgorithmMainStructure} outputs the projection of \(\CharSys\) to \(\vect{x}_1\) and \(\vect{y}_r\) by definition, which is equal to \(\vectSet{L}_1' \circ \dots \circ \vectSet{L}_r'\). Furthermore, we have \(\vectSet{X}=\vectSet{X}_1 \circ \dots \circ \vectSet{X}_r\). We hence want to apply Lemma \ref{LemmaShiftGoodOverapproximation}, and consider the intersection \(\piout(\vectSet{L}_{i-1}') \cap \piin(\vectSet{L}_{i}')\). By P3, these sets are equal to the projection of the solutions of \(\CharSys\) to \(\vect{y}_{i-1}\) and \(\vect{x}_i\) respectively. Since \(\CharSys\) contains the equations \(\vect{y}_{i-1}=\vect{x}_i\), these sets are hence equal, in particular their intersection is non-degenerate. Hence applying Lemma \ref{LemmaShiftGoodOverapproximation} \(r-1\) times, we get that \(\vectSet{X}=\vectSet{X}_1 \circ \dots \circ \vectSet{X}_r\) has the nice overapproximation \(\vectSet{L}_1' \circ \dots \circ \vectSet{L}_r'\) as claimed.

%Hence it suffices to prove Lemma \ref{LemmaLocalCorrectness}. 

\begin{lemma} \label{LemmaRepeatedCycleExecutable}
Let \((\vect{x}, \vect{y}) \in \Rel(\VAS_i, p,p)\) be some cycle, and \(n \in \N\). Let \((\vect{x}', \vect{y}')\) be s.t. \(\vect{y}'-\vect{x}'=n(\vect{y}-\vect{x})\) and \(\vect{x}', \vect{y}' \geq \vect{x}\). Then \((\vect{x}', \vect{y}') \in \Rel(\VAS_i, p,p)\).
\end{lemma}

\begin{proof}
Since \(\vect{x}' \geq \vect{x}\), by monotonicity we can use the cycle and hence reach \(\vect{x}'+(\vect{y}-\vect{x})=\vect{x}'+\frac{1}{n}(\vect{y}'-\vect{x}')=\frac{n-1}{n}\vect{x}'+\frac{1}{n}\vect{y}'\). This is a convex combination of vectors \(\geq \vect{x}\), and hence still \(\geq \vect{x}\). We can therefore repeat the cycle reaching \(\frac{n-2}{n}\vect{x}'+\frac{2}{n}\vect{y}'\), another such convex combination. We repeat this \(n\) times.
\end{proof}

\begin{proof}[Proof of Theorem \ref{TheoremCorrectness}]
Clearly \(\sol(\CharSys)\) is hybridlinear since it is the solution set of an ILP. Specifically, we have \(\sol(\CharSys)=\vectSet{B}+\N(\vectSet{F})\), where \(\vectSet{B}\) is the set of minimal solutions and \(\vectSet{F}\) is the set of minimal non-zero solutions of \(\HomCharSys\), the homogeneous version of $\CharSys$~\eqref{eq:CharSys}. Hence the main part is to show that for any \(\vect{s} \in \vectSet{B}+\N(\vectSet{F})\) and every \(\vect{h} \in \N_{\geq 1}(\vectSet{F})\) there exists a threshold \(N \in \N\) bounded by \(\mathfrak{F}_{\alpha+2d+2}\) such that \(\vect{s}+\N_{\geq N} \vect{h} \subseteq \Concretizable(\VAS, \qin, \qfin)\). Here \(\vect{h} \in \N_{\geq 1}(\vectSet{F})\) is important: This means that \(\vect{h}\) assigns a non-zero value to every variable for which there is a non-zero solution of \(\HomCharSys\), i.e.\ \(\vect{h}\) is a so-called \emph{full support} homogeneous solution. By adding \(\vect{h}\) often enough, we have to create a run for all large enough \(m > N\). 

We first observe that if a solution \(\vect{s}\) is concretizable when restricted to every single SCC \(\VAS_i\), 
and is also concretizable when restricted to every single bridge,
then it is concretizable. 
Hence in the sequel we only consider w.l.o.g. a single SCC \((\VAS_i, \qini, \qfini)\),
obtaining for each of them separately a threshold $N_i$,
and later consider a single bridge edge $e$, obtaining for each of them a threshold $N(e)$.
As the global threshold $N$ we take then the maximum of all of thresholds
$N_i$, $N(e)$.

\para{I. SCCs:}
We proceed in three steps: First we show that (any) solution gives rise to a so-called \(\Z\)-run \(\rho_{\Z}\). 
As step 2 we use the pumping sequences of \PTwo \ to transform the \(\Z\)-run into an \emph{almost-run} \(\rho_{\N}^{(m)}\) with \(\sol(\rho_{\N}^{(m)})=\vect{s}+m \vect{h}\), for sufficiently large $m$,
where an almost-run \(\rho_{\N}\) is a sequence of transitions which would be a run if \(\vectSet{R}(e)=\vectSet{L}(e) = \vectSet b(e){\uparrow}\) were to hold for every edge \(e\). 
This step will essentially reuse the techniques used for VASS.
At last, we refine step 2 to transform the \(\Z\)-run into an actual run. 
%In particular, step 2 will repeat the necessary techniques for VASS, with step 3 refining them with another technique.

In the following we have to access the value a solution assigns to specific variables of \(\CharSys\). We write \(\vect{s}[v]\) for the value solution \(\vect{s}\) assigns to variable \(v\), often using vectors instead of single variables. In particular, we write \(\vect{s}[\vect{x}_i]\) and \(\vect{s}[\vect{y}_i]\) for the vectors the solution \(\vect{s}\) assigns to the source/target variables \(\vect{x}_i, \vect{y}_i\). 

\para{Step 1 (\(\Z\)-\textbf{runs}):} 
%This also repeats the definition of the characteristic system for the reader. 
By adding \(\vect{h}\) at least once we can assume that \(\vect{s}\) is a full support solution. 
%Consider some SCC \((\VAS_i, \qini, \qfini)\), in particular remember
%\[\LocalCharSys_i:=\text{EK}_i(\vect{w}_i) \wedge \vect{y}_i-\vect{x}_i \in \Effect(\vect{w}_i)+\Effect(\N(\vectSet{F}(\VAS_i))),\]
%where \(\vectSet{F}(\VAS_i)=\bigcup_{e \in E_i} \vectSet{F}(e)\). In a similar fashion, 
We define \(\Z\)-semantics of edges \(e \in E_i\):
\[\Z\text{-sem}(\vectSet{L}(e)):=\{(\vect{x}, \vect{y}) \in \Z^{d_i} \times \Z^{d_i} \mid \vect{y}-\vect{x}\in \Delta(\vectSet{L}(e))\},\]
and use it to define a \(\Z\)-run as a sequence of pairs \((q_j(\vect{z}_j))_{i=1\ldots r}\) with \(q_j \in Q_i, \vect{z}_j \in \Z^d\) s.t. 
$q_1 = \qini$, $q_r = \qfini$, \((q_j, q_{j+1}) = e_j \in E_i\), and \((\vect{z}_j, \vect{z}_{j+1}) \in \Z\text{-sem}(\vectSet{L}(e_j))\). By Euler's Lemma \ref{LemmaBasicEulerKirchhoff}, full support solutions of \(\LocalCharSys_i\) are essentially sequences of steps in the \(\Z\)-semantics, where we forgot about the step at which we used the periods, we only remember the number of periods. By simply using all periods $\vect p(e)$ on the first instance of an edge \(e\), any solution of \(\LocalCharSys_i\) gives rise to a \(\Z\)-run 
$\rho_\Z$ from $\vect z_1 = s[\vect x_i]$ to $\vect z_r = s[\vect y_i]$. 

\para{Step 2 (\(\Z\)-\textbf{runs} \(\to\) \textbf{almost-runs}):} 
By an \emph{almost-run} we mean a $\Z$-run whose every step satisfies
\((\vect{z}_j, \vect{z}_{j+1}) \in \N\text{-sem}(\vectSet{L}(e_j))\), form the more restrictive
\emph{$\N$-semantics} that requires $\vect z_j, \vect z_{j+1}\in\N^{d_i}$ but not necessarily 
$(\vect z_j, \vect z_{j+1}) \in \vectSet L(e_j)$:
\[\N\text{-sem}(\vectSet{L}(e)):=\{(\vect{x}, \vect{y}) \in \N^{d_i} \times \N^{d_i} \mid \vect{y}-\vect{x}\in \Delta(\vectSet{L}(e))\}.\]

We write \(\vect{p}(e)\) for the vector of all period variables of an edge \(e\). If an edge \(e\) has periods \(\vect{p}_1, \dots, \vect{p}_\ell\), then we will succinctly write/abuse notation \(\sum \vect{s}[\vect{p}(e)]:=\sum_{j=1}^\ell \vect{s}[\#(\vect{p}_j)] \cdot \vect{p}_j\), i.e.\ we sum the periods of \(\vectSet{L}(e)\) as often as prescribed by the vector \(\vect{p}(e)\).

By Lemma \ref{LemmaEquivalentPropertySix}, since \PTwo \ holds, we get runs \(\mathbf{up}\) and \(\mathbf{dwn}\) which are enabled at \(\vect{s}[\vect{x}_i]+m \vect{h}[\vect{x}_i]\) and reverse enabled at \(\vect{s}[\vect{y}_i]+m \vect{h}[\vect{y}_i]\) for all large enough \(m\), respectively, and increase the counters from \(\Iin\) and \(\Iout\). Here \(\Iin\) is the set of counters which are 0 in \(\vect{h}[\vect{x}_i]\), and likewise for \(\Iout\) and \(\vect{h}[\vect{y}_i]\). We remark that \(\mathbf{up}\) and \(\mathbf{dwn}\) might however have negative effects on \(\{1,\dots, n_i\} \setminus \Iin\) and \(\{1,\dots,n_i\} \setminus \Iout\) respectively. 
%There is one main difficulty with using \(\mathbf{up}\) and \(\mathbf{dwn}\): They are not homogeneous solutions, but we are only allowed to add multiples of \(\vect{h}\) into our run. Hence w
We will augment them with a third \(\Z\)-run \(\mathbf{diff}\) whose only requirement is that we want to have 
\[sol(\mathbf{up})+\sol(\mathbf{dwn})+\sol(\mathbf{diff})=c \vect{h}\] for some constant \(c\in \N\). We proceed to construct \(\mathbf{diff}\).

By \POne, \(\vect{h}\) assigns a positive value to all edge and period variables \(\#(e)\) and \(\#(\vect{p})\). Hence there exists a large enough constant \(c \in \N\) s.t. 
the run $\mathbf{up}$ is enabled in $\vect s[\vect x_i] +c \vect h[\vect x_i]$ and ends in a vector positive on
all counters outside $\Iin$, 
the run $\mathbf{dwn}$ starts in
a vector positive on all counters outside $\Iout$ and is reverse enabled in $\vect s[\vect y_i] +c \vect h[\vect y_i]$, 
and for every $e \in E_i$ we have
\begin{align} \label{eq:c}
\begin{aligned}
c \cdot \vect{h}[\# e]&-\sol(\mathbf{up})[\# e]-\sol(\mathbf{dwn})[\# e] &\geq 1, \\
c \cdot \vect{h}[\vect{p}(e)]&-\sol(\mathbf{up})[\vect{p}(e)]-\sol(\mathbf{dwn})[\vect{p}(e)] &\geq 0. 
\end{aligned}
\end{align}

By Euler's Lemma \ref{LemmaBasicEulerKirchhoff} there exists a cycle
\(\Z\)-run \(\mathbf{diff}\) with \(\sol(\mathbf{diff})[\#  e]=c \cdot \vect{h}[\# e]-\sol(\mathbf{up})[\# e]-\sol(\mathbf{dwn})[\# e] \geq 1\) for every $e\in E_i$. Since the cycle \(\mathbf{diff}\) uses every edge \(e\), for every edge \(e\) we add \(c \cdot \vect{h}[\vect{p}(e)]-\sol(\mathbf{up})[\vect{p}(e)]-\sol(\mathbf{dwn})[\vect{p}(e)]\) many uses of every period of \(\vectSet{L}(e)\) into the first occurrence of \(e\). This way we achieved \(\sol(\mathbf{up})+\sol(\mathbf{dwn})+\sol(\mathbf{diff})=c \vect{h}\) as required.
Note that transferring periods between different occurrences of the same edge $e$ preserves 
the effect of $\mathbf{diff}$, but may change the configuration $\mathbf{diff}$ is enabled at.

Consider, for every large enough \(m=\ell\cdot c \geq N_i\) multiple of \(c\), the \(\Z\)-run \(\mathbf{up}^{\ell} \mathbf{diff}^{\ell} \rho_{\Z}\, \mathbf{dwn}^{\ell}\). It is easy to check that this is an almost-run: Since \(\mathbf{up}\) and respectively \(\mathbf{down}\) have only positive (resp. only negative) effects on counters not increased already by \(\vect{h}\), we can by monotonicity automatically apply them an arbitrary number \(\ell\in \N\) of times. So \(\mathbf{up}^{\ell}\) is certainly enabled at $\vect s[\vect x_i] + m \vect h[\vect x_i]$, and likewise 
\(\mathbf{dwn}^{\ell}\) is reverse enabled at $\vect s[\vect y_i] + m \vect h[\vect y_i]$. Since both source and target of \(\mathbf{diff}^m\) become arbitrarily large, by Lemma \ref{LemmaRepeatedCycleExecutable} \(\mathbf{diff}^m\) is enabled for all large enough \(m\). Since source/target are now large enough, also \(\rho_{\Z}\) is an almost-run.

We have constructed a run for every large enough $m\geq N_1$ multilple of $c$, while
Theorem \ref{TheoremCorrectness} claims such a run for every large enough \(m \in \N\), not only multiples of \(c\). We achieve this by a simple trick: We repeat the above construction of \(\mathbf{diff}\) to obtain \(\mathbf{diff}_1\) with 
\[\sol(\mathbf{up})+\sol(\mathbf{dwn})+\sol(\mathbf{diff}_1)=(c+1) \vect{h}.\]
We restrict to \(N_i \geq c^2\), and observe that any \(m \geq c^2\) can be written as \(m=j_1 c + j_2 (c+1)\) with \(j_1, j_2 \in \N\). This allows us to define for every \(m \geq N_i\) the following almost-run:
 \[\rho_{\N}^{(m)}:=\mathbf{up}^{j_1+j_2} \mathbf{diff}^{j_1} \mathbf{diff}_1^{j_2} \rho_\Z\,  \mathbf{dwn}^{j_1+j_2}. \]
Enabledness still follows as before, finishing step 2.

\para{Step 3 (\(\Z\)-\textbf{runs} \(\to\) \textbf{runs}):} 
We refine the previous step so that the above-defined almost-run  
becomes an actual run,
by taking sufficiently \emph{larger} values of $c$ and $m$.

According to our construction,  
all periods from $\vect p(e)$ appearing in $\mathbf{diff}$, $\mathbf{diff}_1$ or $\rho_\Z$
are used at first occurrence of $e$,
for every edge $e\in E_i$.
Non-first occurrences $(\vect z_j, \vect z_{j+1})  \in \N\text{-sem}(\vectSet L(e_j))$
are thus easy:
$\vect z_{j+1} - \vect z_j = \Delta(\vect b(e_j))$, and
by choosing $m$ sufficiently large, we lift  $(\vect z_j, \vect z_{j+1})$
to guarantee that $(\vect z_j, \vect z_{j+1})  = \vect b(e_j) + (\vect a, \vect a)$ 
for some $\vect a \in \N^{d_i}$.
By the so far unused assumption all $\vectSet L(e_j)$ are basic, $\vect b(e_j) \in \vectSet R(e_j)$,
which together with monotonicity
gives us $(\vect z_j, \vect z_{j+1})  \in \vectSet R(e_j)$ as required.
Note the above argument is not sufficient for
the first occurrences of edges, as it only gives us
$(\vect z_j, \vect z_{j+1}) \in\vectSet L(e_j)$,
which does not have to imply $(\vect z_j, \vect z_{j+1})  \in \vectSet R(e_j)$.

In order to make the first occurrences of edges in $\mathbf{diff}$ (and $\mathbf{diff}_1$) into actual steps, we sufficiently increase the constant $c$.
Since \(\vect{h}[\vect{p}(e)]\) uses every period of \(\vectSet{L}(e)\), we have \(\sum \vect{h}[\vect{p}(e)] \in \N_{\geq 1}({\vectSet{F}(e)})\). For every \(e_j\) which is the first occurrence of an edge \(e\) we can therefore invoke \(\vectSet{R}(e_j) \HybridizationRelation_g \vectSet{L}(e_j)\) with 
\(\vect{x}(e_j):= (\vect{z}_j, \vect{z}_{j+1})\)
and \(\vect{w}(e_j):=\sum \vect{h}[\vect{p}(e_j)]\) and obtain a number \(N(e_j)\)
 s.t. inserting \(m \vect{h}[\vect{p}(e_j)]\) periods for any number \(m \geq N(e_j)\) yields
 $(\vect{z}_j, \vect{z}_{j+1}) + m \vect w(e_j) \in \vectSet R(e_j)$.
Taking $c' := c + \max_{e_j} N(e_j)$ we get a  $\textbf{diff}$
satisfying
\((\vect{z}_j, \vect{z}_{j+1}) \in \vectSet R(e_j)\) also for the first occurrence of \(e\).
Hence for sufficiently large $m$, all repetitions of  $\mathbf{diff}$ (and $\mathbf{diff}_1$) in $\rho_{\N}^{(m)}$
are actual runs.
%We increase $c$: \(\knew:=c+\max_{e \in E_i} N(e)\).

%We can now build the actual \(\mathbf{diff}\) by adapting \(\mathbf{cyc}\): We add \(\knew \vect{h}[\vect{p}(E_i)]-\sol(\mathbf{up})[\vect{p}(E_i)]-\sol(\mathbf{dwn})[\vect{p}(E_i)]\) times the periods \(\vect{p}(e)\) into the selected edge \(e_{j(e)}'\) of \(\mathbf{cyc}\) to obtain \(\mathbf{cyc}'\). Every step is now feasible in \(\vectSet{R}(e)\), but observe that changing \(c\) to \(\knew\) means we are now missing \((\knew-c)\vect{h}[\#(e)]\) uses of every edge \(e\). We use the last unused property of perfectness: Since the overapproximations \(\vectSet{L}(e)\) are basic, we have \(\vect{b}(e)\in \vectSet{R}(e)\) for every \(e \in E_i\). Hence we just pick any cycle \(\mathbf{cyc}''\) with \(\sol(\mathbf{cyc}'')[\#E_i]=(\knew-c)\vect{h}[E_i]\), and do not use any periods in \(\mathbf{cyc}''\). We define \(\mathbf{diff}=\mathbf{cyc}' \mathbf{cyc}''\), which by construction fulfills \(\sol(\mathbf{up})+\sol(\mathbf{diff})+\sol(\mathbf{dwn})=\knew \vect{h}\) and is enabled at every large enough configuration.

Finally, by increasing $c$ further we can get some additional periods in $\mathbf{diff}$ or $\mathbf{diff}_1$,
%By adding a large enough constant \(c_1 \in \N\) copies of \(\vect{h}\), 
which can be distributed 
among steps of \(\rho_{\Z}\) to ensure that \(\rho_{\Z}\) is an actual run as well.
% (though now between ridiculously large configurations, not its original source and target). 
%For the uses of edges \(e\) required by \(\vect{h}\) observe that all \(\vectSet{L}(e)\) are \emph{basic}: I.e.\ simply add any cycle using every edge \(e\) a total of \(c_1 \cdot \vect{h}[\#(e)]\) many times, while not using any periods, to \(\rho_{\Z, \N}\). The new \(\rho_{\Z, \N}'\) remains a run (between large configurations).
%
%
%\textbf{Finishing the Proof}: At this point \(\mathbf{up}^{m} \mathbf{diff}^m \rho_{\Z, \N}' \mathbf{dwn}^m\) is a run for every large enough \(m\), but we have the same minor problem as in step 2: We can so far only pump a number of homogeneous solutions of the form \(c_1+m \cdot \knew\), since we have to add full copies of \(\mathbf{up}, \mathbf{diff}\) and \(\mathbf{dwn}\), however Theorem \ref{TheoremCorrectness} claims a run for \emph{every} \(m \geq N\). 
%
%The repair is to simply use the \(j_1, j_2\) trick again.

\para{II. Bridge edges:}
Consider a bridge edge $e=(\qfini,\qiniplusone)$,  $\vect z := \vect s[\vect y_i]$ and
$\vect z' := \vect s[\vect x_{i+1}]$.
We can not use monotonicity, but
$\CharSys$ \eqref{eq:CharSys} guarantees that $(\vect z, \vect z')\in \vectSet L(e)$. 
We may thus immediately invoke
\(\vectSet{R}(e) \HybridizationRelation_g \vectSet{L}(e)\), similarly as in Step 3, 
and obtain a number \(N(e)\)
 s.t. inserting \(m \vect{h}[\vect{p}(e)]\) periods for any number \(m \geq N(e)\) yields
 $(\vect{z}, \vect{z}') + m \vect w \in \vectSet R(e)$, where $\vect w = \sum \vect{h}[\vect{p}(e)]$.

\para{Bounding N:} The main requirement on the minimal number \(N\) of homogeneous solutions to add is
the inequalities \eqref{eq:c},
% \(c \cdot \vect{h}[\#E_i]-\sol(\mathbf{up})[\# E_i]-\sol(\mathbf{dwn})[\# E_i] \geq 1^{E_i}\) and \(c \cdot \vect{h}[\vect{p}(E_i)]-\sol(\mathbf{up})[\vect{p}(E_i)]-\sol(\mathbf{dwn})[\vect{p}(E_i)] \geq 0^{\vect{p}(E_i)}\) respectively, 
as well as \(c' \geq N(e_j)\) for first occurrences of \(e \in E_i\). That \(N(e_j)\) has size at most \(\mathfrak{F}_{\alpha}\) follows since class \(\RelationClass\) is approximable in \(\mathfrak{F}_{\alpha}\), for the other condition observe that the size of the pumping sequences corresponds roughly to the complexity of deciding coverability, i.e.\ is at most \(\mathfrak{F}_{\alpha+d+1}\). Hence we can find a function \(g \in \mathfrak{F}_{\alpha+d+1} \subseteq \mathfrak{F}_{\alpha+2d+2}\) as required.
\end{proof}

%\section{Well-Quasi-Orders} \label{SectionWellQuasiOrders}
%
%\input{Well_Quasi_Orders.tex}

%\section{Directed Hybridlinear Sets}\label{SectionNewLinearSets}
%
%\input{Directed_Hybridlinear_Sets.tex}

\section{Asymptotic Overapproximation} \label{SectionHybridization}

% !TEX root = Main.tex

In this section we prove that for classes \(\RelationClass\), which are approximable in \(\mathfrak{F}_{\alpha}\) and effectively closed under intersection with semilinear sets in \(\mathfrak{F}_{\alpha}\), in the following called \(\mathfrak{F}_{\alpha}\)-\emph{effective classes}, like sections of \ConsideredModel, 
a wide range of problems can be solved. We start with the following initial list.

\begin{restatable}{theorem}{TheoremUseHybridizationToDecideEverything} \label{TheoremUseHybridizationToDecideEverything}
Let \(\RelationClass\) be \(\mathfrak{F}_{\alpha}\)-effective.
Then the following problems are decidable in the complexity stated:
\begin{enumerate}
\item[(1)] Reachability (is \(\vectSet{X}\) non-empty?) in \(\mathfrak{F}_{\alpha}\).
\item[(2)] Boundedness (is \(\vectSet{X}\) finite?) in \(\mathfrak{F}_{\alpha}\).
%\item[(3)] Semilinearity (is \(\vectSet{X}\) semilinear, and 
%if yes, output a semilinear representation) in \(\mathfrak{F}_{\alpha+1}\).
%\item[(4)] Given \(\vectSet{X}\) and semilinear \(\vectSet{S}\), is \(\vectSet{S} \subseteq \vectSet{X}\)? In 
%\(\mathfrak{F}_{\alpha+1}\).
\item[(3)] Given \(\vectSet{X}\) and semilinear \(\vectSet{S}\), is \(\vectSet{X} \subseteq \vectSet{S}\)? 
In \(\mathfrak{F}_{\alpha}\).
\item[(4)] \(\mathcal{F}\)-separability for \(\mathcal{F}=\) Semil, Mod, Unary. In \(\mathfrak{F}_{\alpha+1}\).
\end{enumerate}
\end{restatable}

\begin{proof}[Proof of Theorem \ref{TheoremUseHybridizationToDecideEverything}(1)-(3)]
(1) See Corollary \ref{CorollaryVASSnzIdealDecomposition}.

(2): If \(\vectSet{X} \HybridizationRelation \vectSet{L}\), and \(\vectSet{L}\) has a period, then \(\vectSet{X}\) is infinite: Simply pump the period. Hence for (2) we apply the approximation algorithm and check whether some \(\vectSet{L}_j\) has a period.

(3): Check that \(\vectSet{X}_{new}:=\vectSet{X} \cap \vectSet{S}^{C}=\emptyset\) via (1), where \(\vectSet{S}^C\) is the complement of \(\vectSet{S}\).
\end{proof}

\para{Separability:}
Next let us define separability problems. Fix a class \(\mathcal{F}\) of relations, e.g. \(\mathcal{F}=\) Semilinear, \(\mathcal{F}=\) Unary (recognizable sets, i.e.~sets definable via monadic Presburger predicates), or
\(\mathcal{F}=\) Modulo (sets definable by modulo constraints on coordinates).
%
%\begin{definition}
Two sets \(\vectSet{X}\) and \(\vectSet{Y}\) are \(\mathcal{F}\)-\emph{separable} if there exists \(\vectSet{S} \in \mathcal{F}\) s.t. \(\vectSet{X} \subseteq \vectSet{S}\) and \(\vectSet{Y} \cap \vectSet{S}=\emptyset\). The \(\mathcal{F}\)-\emph{separability problem} asks given \(\vectSet{X}\) and \(\vectSet{Y}\), are they \(\mathcal{F}\)-separable?
%\end{definition} 
%
To solve the separability problems, we use the following definition:

\begin{definition}
Let \(\vectSet{X}\) be any set. Let \(\vectSet{S} \in \mathcal{F}\) with \(\vectSet{X} \subseteq \vectSet{S}\). Then \(\vectSet{S}\) is called \emph{up-to-boundary-optimal} (utbo) \(\mathcal{F}\) \emph{overapproximation of} \(\vectSet{X}\) if any other overapproximation \(\vectSet{S}' \in \mathcal{F}\) fulfills \(\vectSet{S}+\vect{p} \subseteq \vectSet{S}'\) for some vector \(\vect{p}\in \N^d\).
\end{definition}

Intuitively, considering a set $\vectSet S \in \mathcal{F}$ including $\vectSet X$ to be an \(\mathcal{F}\) overapproximation  of $\vectSet X$,
an utbo  \(\mathcal{F}\) overapproximation $\vectSet X$ is, up to a shift by some vector $\vect p$, 
the optimal (the least) among all  \(\mathcal{F}\) overapproximations of $\vectSet X$.

In fact, an asymptotic overapproximation of $\vectSet X$ is an utbo semilinear overapproximation of
$\vectSet X$:
% why we defined \(\HybridizationRelation\) the way we did.

\begin{lemma} \label{LemmaUTBOApproximation}
Let \(\vectSet{X} \HybridizationRelation \vectSet{L}=\vectSet{B}+\N(\vectSet{F})\) and \(\vectSet{S}\) be semilinear s.t. \(\vectSet{X} \subseteq \vectSet{S}\). Then there exists \(\vect{p}\in \N(\vectSet{F})\) s.t. \(\vect{p}+\vectSet{L} \subseteq \vectSet{S}\).
\end{lemma} 

\begin{proof}
Define \(\vectSet{S}':=\vectSet{L} \setminus \vectSet{S}\). Write \(\vectSet{S}'=\bigcup_{j=1}^k \vectSet{L}_j\) as a union of linear sets \(\vectSet{L}_j=\vect{b}_j+\N(\vectSet{F}_j)\). 
We claim that all the $\vectSet L_j$ are "parallel to the boundary" of \(\vectSet{L}\), namely
\[%I:=\{j \in \{1,\dots, k\} \mid 
\N(\vectSet{F}_j) \cap \N_{\geq 1}(\vectSet{F}) = \emptyset,
\] 
which would essentially finish the proof. 
%be the set of indices s.t.~\(\vectSet{L}_j\) is "not parallel to the boundary" of \(\vectSet{L}\). 
%We claim that \(I=\emptyset\), i.e.~all of \(\vectSet{L}_j\) are "parallel to the boundary" of \(\vectSet{L}\), which would finish the proof. 
%First we define ``parallel to the boundary''.
%
%Write \(\vectSet{L}_j=\vect{b}_j+\N(\vectSet{F}_j)\) for all \(j\). 

%
Suppose, towards contradiction, that \(\vect{w}_j \in \N(\vectSet{F}_j) \cap \N_{\geq 1}(\vectSet{F})\) for some  \(j \in I\). Since \(\vectSet{X} \HybridizationRelation \vectSet{L}\), there exists \(N \in \N\) s.t. \(\vect{b}_j+N \vect{w}_j \in \vectSet{X}\), in particular \(\vectSet{X} \cap \vectSet{L}_j \neq \emptyset\). 
However, \(\vectSet{X} \cap \vectSet{L}_j \subseteq \vectSet{S} \cap (\vectSet{L} \setminus \vectSet{S})=\emptyset\), contradiction.

Therefore all \(\vectSet{L}_j\) for \(j=1, \dots, k\) are parallel to the boundary of \(\vectSet{L}\) as claimed, and we now simply choose \(\vect{p} \in \N(\vectSet{F})\) large enough that \((\vect{p}+\vectSet{L}) \cap \vectSet{S}'=\emptyset\). 
\end{proof}

Now we can obtain Theorem \ref{TheoremUseHybridizationToDecideEverything}(4) as a consequence of the following:

\begin{theorem} \label{TheoremDecidingSeparability}
Let \(\RelationClass\) be \(\mathfrak{F}_{\alpha}\)-effective.
%approximable in \(\mathfrak{F}_{\alpha}\), and effectively closed under intersection with semilinear sets  in \(\mathfrak{F}_{\alpha}\).
Let \(\mathcal{F}\) be a class of relations s.t. the following hold:
\begin{itemize}
\item \(\mathcal{F}\) is effectively closed under Boolean operations in \(\mathfrak{F}_{\alpha}\).
\item Given \(\vectSet{T} \in \mathcal{F}\), \(\dim(\vectSet{T})\) is computable in \(\mathfrak{F}_{\alpha}\).
\item There is an algorithm in \(\mathfrak{F}_{\alpha}\) that
given an utbo semilinear overapproximation \(\vectSet{S}\) of \(\vectSet{X}\), outputs
an utbo \(\mathcal{F}\) overapproximation \(\vectSet{S}'\) of the same \(\vectSet{X}\).
\end{itemize}
Then \(\mathcal{F}\)-separability is decidable in \(\mathfrak{F}_{\alpha+1}\).
\end{theorem}

\begin{proof}
We check \(\mathcal{F}\)-separability for \(\vectSet{X}\) and \(\vectSet{Y}\) as follows:

Step 1:  Apply the overapproximation algorithm to write \(\vectSet{X}=\vectSet{X}_1 \cup \dots \cup \vectSet{X}_k\) and \(\vectSet{Y}=\vectSet{Y}_1 \cup \dots \cup \vectSet{Y}_s\) with \(\vectSet{X}_j \HybridizationRelation \vectSet{L}_j\) and \(\vectSet{Y}_m \HybridizationRelation \vectSet{S}_m\) for all \(j,m\). Since \(\mathcal{F}\) is closed under all Boolean operations, it is sufficient and necessary to \(\mathcal{F}\)-separate every \(\vectSet{X}_j\) from every \(\vectSet{Y}_m\). 

Step 2: For all indices \(j,m\) do: By Lemma \ref{LemmaUTBOApproximation}, \(\vectSet{L}_j\) is an utbo semilinear overapproximation of \(\vectSet{X}_j\) and similarly \(\vectSet{S}_m\) is an utbo semilinear overapproximation of \(\vectSet{Y}_m\). From these compute utbo \(\mathcal{F}\) overapproximations \(\vectSet{L}_j'\) and \(\vectSet{S}_m'\) of \(\vectSet{X}_j\) and \(\vectSet{Y}_m\) respectively using bullet point 3.

Step 3: Finally we check whether \(\vectSet{L}_j' \cap \vectSet{S}_m'\) is non-degenerate for \emph{some} \(j\) and \(m\). This can be done by bullet points 1 and 2. 

Step 4: If such \(j,m\) exist, then the algorithm rejects. Otherwise the algorithm continues recursively with \(\vectSet{X}_{\text{new}}:=\vectSet{X} \cap \bigcup_{j,m} (\vectSet{L}_j' \cap \vectSet{S}_m')\) and \(\vectSet{Y}_{\text{new}}:=\vectSet{Y} \cap \bigcup_{j,m} (\vectSet{L}_j' \cap \vectSet{S}_m')\), which are of lower dimension.

Correctness: We prove that the algorithm can safely reject if any \(\vectSet{L}_j' \cap \vectSet{S}_m'\) is non-degenerate. We remove the indices for readability, refering to \(\vectSet{X}_j, \vectSet{L}_j'\) as \(\vectSet{X}, \vectSet{L}\) and \(\vectSet{Y}_m, \vectSet{S}_m'\) as \(\vectSet{Y}, \vectSet{S}\).

Since \(\mathcal{F}\) is closed under complement, \(\vectSet{X}, \vectSet{Y}\) are \(\mathcal{F}\)-separable if and only if there exist \(\vectSet{T}_{\vectSet{X}}, \vectSet{T}_{\vectSet{Y}} \in \mathcal{F}\) s.t. \(\vectSet{X} \subseteq \vectSet{T}_{\vectSet{X}}, \vectSet{Y} \subseteq \vectSet{T}_{\vectSet{Y}}\) and \(\vectSet{T}_{\vectSet{X}} \cap \vectSet{T}_{\vectSet{Y}}=\emptyset\). We claim that such \(\vectSet{T}_{\vectSet{X}}, \vectSet{T}_{\vectSet{Y}}\) do not exist.

Indeed, take
arbitrary \(\vectSet{T}_{\vectSet{X}},\vectSet{T}_{\vectSet{Y}} \in \mathcal{F}\) with \(\vectSet{X} \subseteq \vectSet{T}_{\vectSet{X}}\) and \(\vectSet{Y} \subseteq \vectSet{T}_{\vectSet{Y}}\), as candidates for separating. Since \(\vectSet{L}\) and \(\vectSet{S}\) are utbo \(\mathcal{F}\) overapproximations of \(\vectSet{X}, \vectSet{Y}\) respectively, there exist shifts \(\vect{x}, \vect{y}\) such that \(\vect{x}+\vectSet{L} \subseteq \vectSet{T}_{\vectSet{X}}\) and \(\vect{y}+\vectSet{S} \subseteq \vectSet{T}_{\vectSet{Y}}\). Since \(\vectSet{L} \cap \vectSet{S}\) is non-degenerate, also \((\vect{x}+\vectSet{L}) \cap (\vect{y}+\vectSet{S})\) is non-degenerate, and in particular \(\vectSet{T}_{\vectSet{X}} \cap \vectSet{T}_{\vectSet{Y}} \supseteq (\vect{x}+\vectSet{L}) \cap (\vect{y}+\vectSet{S}) \neq \emptyset\) is non-empty, which invalidates the candidates 
\(\vectSet{T}_{\vectSet{X}},\vectSet{T}_{\vectSet{Y}}\).

Complexity Bound: The recursion depth of Step 4 is at most \(n\), i.e.\ linear, 
and one call takes \(\mathfrak{F}_{\alpha}\), hence we obtain \(\mathfrak{F}_{\alpha+1}\).
\end{proof}

We remark that while Theorem \ref{TheoremDecidingSeparability} may be applied for \(\mathcal{F} \not \subseteq\) Semilinear, already for \(\mathcal{F}=\) Semilinear separability of \(\vectSet{X}\) and \(\vectSet{Y}\) is equivalent to disjointness of \(\vectSet{X}\) and \(\vectSet{Y}\). This fact is similar to \cite[Cor. 6.5]{GuttenbergRE23} and is a consequence of the algorithm of Theorem \ref{TheoremDecidingSeparability}  for \(\mathcal{F}=\) Semilinear that only rejects (in step 4) if \(\vectSet{L}_i \cap \vectSet{S}_j\) non-degenerate, i.e.~by Lemma \ref{LemmaShiftGoodOverapproximation}, only when \(\vectSet{X}_i \cap \vectSet{Y}_j \HybridizationRelation \vectSet{L}_i \cap \vectSet{S}_j\), which implies \(\vectSet{X}_i \cap \vectSet{Y}_j \neq \emptyset\). Otherwise the algorithm outputs a semilinear separator.

We are ready to finish the proof of Theorem \ref{TheoremUseHybridizationToDecideEverything}.

\begin{proof}[Proof of Theorem \ref{TheoremUseHybridizationToDecideEverything}(4)]
It is easy to check that the classes mentioned fulfill bullet points 1-3 of Theorem \ref{TheoremDecidingSeparability}. Hence Theorem \ref{TheoremUseHybridizationToDecideEverything}(6) follows from Theorem \ref{TheoremDecidingSeparability}.
\end{proof}

\para{Semilinearity:}
Finally, we called Theorem \ref{TheoremUseHybridizationToDecideEverything} the \emph{initial} list, because another host of problems can be shown to be decidable.
To this aim we use, in place of \(\HybridizationRelation\), a stronger relatrion \(\HybridizationRelation_s\)
defined in \cite[Def.~4.5]{GuttenbergRE23}.
 %defines a stronger version of \(\HybridizationRelation\) which we write \(\HybridizationRelation_s\). 
While the definition of \(\HybridizationRelation_s\) is too long  to be repeated here, let us explain the main difference.

%\begin{definition}
%Let \(\vectSet{X} \subseteq \N^d\), and \(\vect{v} \in \N^d\). Then \(\vect{v}\) is a \emph{direction} of \(\vectSet{X}\) if \(\exists n \in \N, \vect{x} \in \vectSet{X}\) s.t. \(\vect{x}+n \vect{v} \subseteq \vectSet{X}\).
%
%The set of all directions is denoted \(\dir(\vectSet{X})\).
%\end{definition}
%
%I.e.\ a vector \(\vect{v}\) is a direction if a (possibly rescaled) infinite line in this direction is contained in \(\vectSet{X}\). The main part of strong-approximability is that \(\dir(\vectSet{X})\) has to be definable in \(\FO(\Q, +, \leq)\), i.e.\ in linear arithmetic. This formula is intuitively a more precise approximation of the asymptotics.
%
%The second main problem of \(\HybridizationRelation\) which is fixed in \(\HybridizationRelation_s\) is the following.

\begin{example}
Recall that \(\vectSet{X}:=\{(x,y) \mid y \leq x^2\}\) fulfills \(\vectSet{X} \HybridizationRelation \N^2\), in fact \(\vectSet{X} \HybridizationRelation_s \N^2\). Now let \(\vectSet{Y} \subseteq \N\) be some arbitrary infinite set, best an undecidable one, and define \(\vectSet{X}':=\vectSet{X} \cup \{0\} \times \vectSet{Y}\). Then we automatically still have \(\vectSet{X}' \HybridizationRelation \N^2\) since \(\HybridizationRelation\) only becomes easier by adding points to $\vectSet X$. However, we do \emph{not} have \(\vectSet{X}' \HybridizationRelation_s \N^2\) as strong overapproximation is designed to ``detect complex unrelated parts'' in the set.
\end{example}

Next we say that a class \(\RelationClass\) is \(\mathfrak{F}_{\alpha}\)-\emph{strongly effective} if it is \(\mathfrak{F}_{\alpha}\)-effective and the approximation algorithm can be assumed to output \(\vectSet{X}_i \HybridizationRelation_s \vectSet{L}_i\) instead of just \(\vectSet{X}_i \HybridizationRelation \vectSet{L}_i\). In the full version we show the following analogue of Theorem \ref{TheoremIdealDecompositionEVASS}:

\begin{theorem}
Let \(\alpha \geq 2\) and let \(\RelationClass\) be an \(\mathfrak{F}_{\alpha}\)-strongly effective class of relations containing \(Add\).
Then sections of monotone \(\RelationClass\)-eVASS are \(\mathfrak{F}_{\alpha+2d+2}\)-strongly effective. Moreover, the algorithm is still the one given in Section \ref{SectionProofTheoremEVASS}.\label{TheoremStronglyApproximableEVASS}
\end{theorem}

We immediately obtain the analogue of Theorem \ref{TheoremVASSnzIdealDecomposition}:

\begin{corollary}
The class of sections of VASSnz of dimension \(d\) and with \(k\) priorities is \(\mathfrak{F}_{2kd+2k+2d+4}\)-strongly effective.
\end{corollary}

To explain how to apply strong approximability to deciding semilinearity, we need two preliminaries from \cite{GuttenbergRE23}.

\begin{definition}
Let \(\vectSet{L}=\vectSet{B}+\N(\vectSet{F})\) be hybridlinear and \(\vectSet{X} \subseteq \vectSet{L}\). If there is \(\vect{p} \in \N(\vectSet{F})\) s.t. \(\vect{p}+\vectSet{L} \subseteq \vectSet{X}\), then \(\vectSet{X}\) is called \emph{reducible} w.r.t.~$\vectSet L$ (written $\vectSet X \reducible \vectSet L$), otherwise \emph{irreducible}. 

Such a vector \(\vect{p}\) is called a \emph{reduction point}.
\end{definition}

Intuitively, if \(\vectSet{X}\) is reducible with reduction point \(\vect{p}\), then $\vectSet{X}$ and $\vectSet{L}$ coincide for points larger than $\vect{p}$.
By definition, $\vectSet X \reducible \vectSet L$ implies $\vectSet X \HybridizationRelation \vectSet L$, though it does not imply $\vectSet X \HybridizationRelation_s \vectSet L$.
\begin{theorem}
\cite[Theorem~5.7]{GuttenbergRE23} Let class \(\RelationClass\) be \(\mathfrak{F}_{\alpha}\)-strongly effective. Then there is an algorithm in \(\mathfrak{F}_{\alpha+1}\) which given \(\vectSet{X}\in \RelationClass\) and a hybridlinear set \(\vectSet{L}\), decides whether \(\vectSet{X} \reducible \vectSet{L}\), and if yes computes a reduction point \(\vect{p}\). \label{TheoremDecidingReducibility}
\end{theorem}

The second preliminary is the following.

\begin{proposition} \cite[Prop.~5.5]{GuttenbergRE23}
Let \(\RelationClass\) be \(\mathfrak{F}_{\alpha}\)-strongly effective. Then there is an algorithm in \(\mathfrak{F}_{\alpha+1}\) which given \(\vectSet{X} \in \RelationClass\) computes a partition \(\N^n=\bigcup_{j=1}^k \vectSet{L}_j\) s.t. for every \(j\) either \((\vectSet{X} \cap \vectSet{L}_j) \HybridizationRelation \vectSet{L}_j\)
or \(\vectSet{X} \cap \vectSet{L}_j=\emptyset\). \label{PropositionPartition}
\end{proposition}

I.e.\ the approximation algorithm can be assumed to decompose instead of \(\vectSet{X}\) the surrounding space \(\N^n\) into a partition \(\vectSet{L}_j\) s.t. in every part \(\vectSet{X}\) is either empty or has good asymptotics.

%I.e.\ the strong approximation algorithm can in addition to decomposing \(\vectSet{X}\) split the surrounding space \(\vectSet{L}\) into linear sets, such that in each linear set \(\vectSet{L}_j\) all actually occurring \(\vectSet{X}_i\) parts have \(\vectSet{L}_j\) as strong approximation. This purely relies on  (the equivalent of) the closure property Lemma \ref{LemmaShiftGoodOverapproximation}. This way, instead of deciding reducibility w.r.t. the original \(\vectSet{L}\), it is enough to check whether \(\vectSet{X} \cap \vectSet{L}_j\) is reducible w.r.t. every \(\vectSet{L}_j\) which points into the interior of \(\vectSet{L}\). This way we generated the assumption \(\vectSet{X}=\bigcup_{i \in I} \vectSet{X}_i\) with \(\vectSet{X}_i \HybridizationRelation_s \vectSet{L}_j\) for all \(i\).
%
%Next, remember that \(\vectSet{X}_i \HybridizationRelation_s \vectSet{L}\) in particular requires \(\dir(\vectSet{X}_i)\) to be definable by a formula \(\varphi_i \in \FO(\Q, +, \leq)\). Finally, it shown that even reducibility can be expressed in linear arithmetic as a property called ``existence of a complete extraction''.
%
%\begin{lemma}\cite[Theorem~D.4]{GuttenbergRE23}
%Let \(\vectSet{X}=\bigcup_{i=1}^s \vectSet{X}_i\) fulfill \(\vectSet{X}_i \HybridizationRelation_s \vectSet{L}\) for all \(i\). Let \(\varphi_i:=\dir(\vectSet{X}_i)\). Then \(\vectSet{X} \reducible \vectSet{L}\) if and only if \(\{\varphi_1, \dots, \varphi_s\}\) has a complete extraction, and a reduction point can be computed from the complete extraction.
%\end{lemma}

We can now finish with the remaining part of the list of decidable problems:

\begin{restatable}{theorem}{TheoremUseHybridizationToDecideEverythingPartTwo} \label{TheoremUseHybridizationToDecideEverythingPartTwo}
Let \(\RelationClass\) be \(\mathfrak{F}_{\alpha}\)-strongly effective. Then the following problems are decidable in the complexity stated:
\begin{enumerate}
\item[(1)] Semilinearity (is \(\vectSet{X}\) semilinear, and 
if yes, output a semilinear representation) in \(\mathfrak{F}_{\alpha+2}\).
\item[(2)] Given \(\vectSet{X}\) and semilinear \(\vectSet{S}\), is \(\vectSet{S} \subseteq \vectSet{X}\)? In 
\(\mathfrak{F}_{\alpha+2}\).
\end{enumerate}
\end{restatable}

\begin{proof}[Proof of Theorem \ref{TheoremUseHybridizationToDecideEverything}]
(1): Lemma \ref{LemmaUTBOApproximation} shows in particular that if \(\vectSet{X} \HybridizationRelation \vectSet{L}\), and \(\vectSet{X}\) is semilinear, then \(\vectSet{X} \reducible \vectSet L\). Namely \(\vectSet{X}\) itself is a semilinear overapproximation. In contraposition: If \(\vectSet{X} \HybridizationRelation \vectSet{L}\) and \(\vectSet{X} \not\reducible \vectset L\), 
% is irreducible, 
then \(\vectSet{X}\) is non-semilinear. 

We obtain an obvious algorithm for deciding semilinearity of a set \(\vectSet{X}\): Use Proposition \ref{PropositionPartition} to decompose \(\N^n=\bigcup_{j=1}^k \vectSet{L}_j\). Let \(J\) be the set of indices with \(\vectSet{X} \cap \vectSet{L}_j \HybridizationRelation \vectSet{L}_j\). We ignore all \(j \not \in J\), and for all \(j\in J\) we 
%have \(\vectSet{X} \cap \vectSet{L}_j \HybridizationRelation \vectSet{L}_j\). 
check whether \(\vectSet{X} \cap \vectSet{L}_j\) is reducible w.r.t. \(\vectSet{L}_j\). If no, then \(\vectSet{X} \cap \vectSet{L}_j\) and therefore \(\vectSet{X}\) is non-semilinear by the above contraposition statement. Otherwise we obtain a reduction point \(\vect{p}_j\) with \(\vect{p}_j+\vectSet{L}_j \subseteq \vectSet{X}\) and continue recursively with \(\vectSet{X}_{new}:=\vectSet{X} \setminus \bigcup_{j \in J} (\vect{p}_j+\vectSet{L}_j)\).

%We explained the algorithm in the paragraphs above/refer to \cite{GuttenbergRE23}, Section 5. For the time bound 
Complexity bound: Proposition \ref{PropositionPartition} needs \(\mathfrak{F}_{\alpha+1}\) time, and we have recursion depth at most \(n\), hence the time bound is an \(n\)-fold application of a function in \(\mathfrak{F}_{\alpha+1}\), leading to \(\mathfrak{F}_{\alpha+2}\).

(2): Apply the semilinearity algorithm on \(\vectSet{X}_{new}:=\vectSet{X} \cap \vectSet{S}\). If you obtain a semilinear representation equivalent to \(\vectSet{S}\), return true, otherwise return false. 
The complexity is again \(\mathfrak{F}_{\alpha+2}\).
\end{proof}

Theorems %\ref{TheoremVASSnzIdealDecomposition}, 
\ref{TheoremUseHybridizationToDecideEverything}, \ref{TheoremStronglyApproximableEVASS} and \ref{TheoremUseHybridizationToDecideEverythingPartTwo} in particular imply:

\begin{corollary}
The reachability, boundedness, semilinearity and \(\mathcal{F}\)-separability (for \(\mathcal{F}=\) Semilinear, Modulo, Unary) problems are decidable in \(\mathfrak{F}_{\omega}\) for VASS and \ConsideredModel.
\end{corollary}

%\section{Ideal Decomposition Implies Axioms} \label{SectionIdealDecompositionToAxioms}
%
%\input{IdealDecompositionToGeometry.tex}

\section{Conclusion} \label{SectionConclusion}

% !TEX root = Main.tex

To achieve our results, culminating in 
Theorems \ref{TheoremUseHybridizationToDecideEverything} and
\ref{TheoremUseHybridizationToDecideEverythingPartTwo}, 
we have introduced monotone \(\RelationClass\)-eVASS, indirectly also answering the question of ``What types of transitions can we allow a VASS to have while remaining decidable?'' Namely, by Theorem \ref{TheoremIdealDecompositionEVASS}, under minor assumptions on \(\RelationClass\), monotone \(\RelationClass\)-eVASS sections are approximable if and only if \(\RelationClass\) is approximable.

This leaves the following questions: 1) Can the parameterized complexity be improved? On the level of \(\RelationClass\)-eVASS, can we avoid \(+d\) for coverability? Does the index have to depend on \(k\)? Do VASSnz have lower complexity than nested \(\RelationClass\)-eVASS? In particular, one can ask about the parameterized complexity for both VASSnz and \(\RelationClass\)-eVASS.

2) Are there other applications of \(\RelationClass\)-eVASS, or other classes \(\RelationClass\) for which they are interesting to consider?

3) Are there other classes of systems which are approximable and reducible, i.e.\ other systems where Theorem \ref{TheoremUseHybridizationToDecideEverything} applies? For example Pushdown VASS?

\newpage

\bibliography{Main.bib}
\appendix

% !TEX root = Main.tex

\subsection{Appendix of Section \ref{AlgorithmToolbox}}

In this section we prove the closure property Lemma \ref{LemmaShiftGoodOverapproximation} from the main text.

First we observe that in Lemma \ref{LemmaShiftGoodOverapproximation} we have not stated which hybridlinear representation will be used for \(\vectSet{L} \cap \vectSet{L}'\), hence we have to prove that \(\vectSet{X} \HybridizationRelation \vectSet{L}\) is independent of this choice. This requires a few lemmas.

\begin{lemma}
Let \(\vectSet{X} \subseteq \vectSet{B}+\N(\vectSet{F}) \subseteq \N^n\). Then \(\vectSet{X} \HybridizationRelation \vectSet{B}+\N(\vectSet{F})\) if and only if \((\vectSet{X} \cap (\vect{b}+\N(\vectSet{F}))) \HybridizationRelation \vect{b}+\N(\vectSet{F})\) for all \(\vect{b} \in \vectSet{B}\). \label{LemmaReduceToLinear}
\end{lemma}

\begin{proof}
``\(\Rightarrow\)'': Let \(\vect{x} \in (\vectSet{X} \cap (\vect{b}+\N(\vectSet{F})))\), and \(\vect{w} \in \N_{\geq 1}(\vectSet{F})\). We have to show that there exists \(N \in \N\) s.t. \(\vect{x}+\N_{\geq N} \vect{w} \subseteq \vectSet{X}\) and \(\vect{x}+\N_{\geq N} \vect{w} \subseteq \vect{b}+\N(\vectSet{F})\). 

For \(\vectSet{X}\) this follows from \(\vectSet{X} \HybridizationRelation \vectSet{B}+\N(\vectSet{F})\), for \(\vect{b}+\N(\vectSet{F})\) this follows since \(\vect{w}\) is a period of the linear set \(\vect{b}+\N(\vectSet{F})\).

``\(\Leftarrow\)'': Let \(\vect{x} \in \vectSet{X}\) and \(\vect{w} \in \N_{\geq 1}(\vectSet{F})\). We have to again construct \(N\). Let \(\vect{b} \in \vectSet{B}\) s.t. \(\vect{x} \in \vect{b}+\N(\vectSet{F})\), which exists since \(\vectSet{X} \subseteq \vectSet{B}+\N(\vectSet{F})\). Use the assumption for \(\vect{b}\).
\end{proof}

By Lemma \ref{LemmaReduceToLinear} we can w.l.o.g. simplify representations to linear ones. 

The simplest way to now finish proving independence of the representation is to consider the following definition only depending on \(\vectSet{L}\), not some representation of it.

\begin{definition} \label{DefinitionInterior}
Let \(\vectSet{P}\) be \(\N\)-g., and \(\vect{v} \in \vectSet{P}\). 

The vector \(\vect{v}\) is in the \emph{interior} of \(\vectSet{P}\) if for every \(\vect{x} \in \vectSet{P}\), there exists \(m \in \N\) s.t. \(m \vect{v}-\vect{x} \in \vectSet{P}\).

The set of all interior vectors is denoted \(\interior(\vectSet{P})\).
\end{definition}

Accordingly, we define the interior of a linear set as \(\interior(\vectSet{L})=\vect{b}+\interior(\vectSet{L}-\vect{b})\), where \(\vect{b}\) is the base point of \(\vectSet{L}\). Observe that since linear sets are \(\subseteq \N^n\), the base point is the unique minimal point in \(\vectSet{L}\), hence this definition does not depend on the representation. 

Now we can prove that asymptotic overapproximations do not depend on the representation, by proving that \(\vect{w} \in \N_{\geq 1}(\vectSet{F})\) can equivalently be replaced by \(\vect{w} \in \interior(\vectSet{L})-\vect{b}\):

\begin{lemma} \label{LemmaIndependentOfRepresentationOfL}
Let \(\vectSet{L}=\vect{b}+\N(\vectSet{F})\) be a linear set, and \(\vectSet{X} \subseteq \vectSet{L}\). Then in Definition \ref{DefinitionGoodOverapproximation} ``\(\forall \vect{w} \in \N_{\geq 1}(\vectSet{F})\)'' can equivalently be replaced by ``\(\forall \vect{w} \in \interior(\vectSet{L})-\vect{b}\)''. 

In particular, the definition of asymptotic overapproximation is independent of the representation of \(\vectSet{L}\).
\end{lemma}

\begin{proof}
First observe that by definition of \(\interior(\vectSet{L})\), we have \(\interior(\vectSet{L})-\vect{b}=(\vect{b}+\interior(\vectSet{L}-\vect{b}))-\vect{b}=\interior(\vectSet{L}-\vect{b})=\interior(\N(\vectSet{F}))\).

``\(\Leftarrow\)'': We assume the line containment property holds for all \(\vect{w} \in \interior(\vectSet{L})- \vect{b}=\interior(\N(\vectSet{F}))\) and prove it for all \(\vect{w} \in \N_{\geq 1}(\vectSet{F})\). It suffices to prove \(\N_{\geq 1}(\vectSet{F}) \subseteq \interior(\N(\vectSet{F}))\). Hence let \(\vect{w} \in \N_{\geq 1}(\vectSet{F})\). To prove \(\vect{w} \in \interior(\N(\vectSet{F}))\), let \(\vect{x} \in \N(\vectSet{F})\) arbitrary. We have to show that there exists \(m \in \N\) s.t. \(m \vect{w} - \vect{x} \in \N(\vectSet{F})\). Write \(\vect{x}=\sum_{\vect{f} \in \vectSet{F}} \lambda_{\vect{f}} \vect{f}\), and simply define \(n:=\max_{\vect{f} \in \vectSet{F}} \lambda_{\vect{f}}\). Then \(m \vect{w} - \vect{x}\) uses every \(\vect{f} \in \vectSet{F}\) at least \(m-\lambda_{\vect{f}} \geq 0\) times, i.e.\ is in \(\N(\vectSet{F})\) as claimed.

``\(\Rightarrow\)'': We assume the line containment property holds for all \(\vect{w} \in \N_{\geq 1}(\vectSet{F})\) and have to prove it for all \(\vect{w} \in \interior(\N(\vectSet{F}))\). Hence let \(\vect{w} \in \interior(\N(\vectSet{F}))\) and \(\vect{x} \in \vectSet{L}\). We claim the following: There exists \(m \in \N\) s.t. \(m \vect{w} \in \N_{\geq 1}(\vectSet{F})\). 

Proof of claim: For all \(\vect{f} \in \vectSet{F}\), we do the following: We use that \(\vect{w}\) is in the interior to obtain \(m_{\vect{f}} \in \N\) s.t. \(m_{\vect{f}} \vect{w}- \vect{f} \in \N(\vectSet{F})\). Then \(m:=\sum_{\vect{f} \in \vectSet{F}} m_{\vect{f}}\) proves the claim.

Therefore let \(m \in \N\) be s.t. \(m \vect{w} \in \N_{\geq 1}(\vectSet{F})\). For all \(0 \leq j \leq m-1\) we define \(\vect{x}_j:=\vect{x}+j \vect{w}\). Observe that \[\bigcup_{j=0}^{m-1} \vect{x}_j+m \N \vect{w}=\vect{x}+\N \vect{w},\] since \(\N=\{0, \dots, m-1\}+m\N\). For all \(0 \leq j \leq m-1\) we use our assumption with \(\vect{x}_j \in \vectSet{L}\) and \(m\vect{w} \in \N_{\geq 1}(\vectSet{F})\) to obtain \(N_j \in \N\) s.t. \(\vect{x}_j+\N_{\geq N_j} m \vect{w} \subseteq \vectSet{X}\). Defining \(N:=m+\max_{0 \leq j \leq m-1} N_j\), we obtain \(\vect{x}+\N_{\geq N} \vect{w} \subseteq \vectSet{X}\).
\end{proof}

We can now proceed towards the closure property Lemma \ref{LemmaShiftGoodOverapproximation}. The main argument will be that because \(\vectSet{L} \cap \vectSet{L}'\) is non-degenerate, the interior \(\interior(\vectSet{L} \cap \vectSet{L}')\) will be ``almost all'' of \(\vectSet{L}'\). Towards this end, we first need to introduce lemmas towards dimension.

The first half of properties of dimension is basic.

\begin{restatable}{lemma}{BasicDimensionProperties}
Let \(\vectSet{X}, \vectSet{X}' \subseteq \mathbb{Q}^n, \vect{b}\in \mathbb{Q}^n\). Then \(\dim(\vectSet{X})=\dim(\vect{b}+\vectSet{X})\) and 
\(\dim(\vectSet{X} \cup \vectSet{X}')=\max \{\dim(\vectSet{X}), \dim(\vectSet{X}')\}\). 

Further, if \(\vectSet{X} \subseteq \vectSet{X}'\), then \(\dim(\vectSet{X}) \leq \dim(\vectSet{X}')\). \label{BasicDimensionProperties}
\end{restatable}

The second half deals with linear sets. For \(\Q\)-generated sets \(\vectSet{V}\) and \(\vectSet{V}_i\) it is known that \(\vectSet{V} \subseteq \bigcup_{i=1}^r \vect{b}_i+\vectSet{V}_i\) implies \(\vectSet{V} \subseteq \vectSet{V}_i\) for some \(i\), similar results hold for any \(\N\)-g. set \(\vectSet{P}\subseteq \Z^n\) and lead to the following lemma.

\begin{lemma}{\cite[Lemma 5.3]{Leroux11}} \label{LemmaFromJerome}
Let \(\vectSet{P} \subseteq \Z^n\) be \(\N\)-generated. Then \(\dim(\vectSet{P})=\dim(\Q(\vectSet{P}))\).
\end{lemma}

Now we can start with ``the interior takes up almost all of \(\vectSet{L}'\)'', through two lemmas.

\begin{lemma}\label{LemmaBoundaryLowerDimension}
Let \(\vectSet{P}=\N(\vectSet{F})\) be \(\N\)-finitely generated.

Then \(\dim(\vectSet{P} \setminus \interior(\vectSet{P})) < \dim(\vectSet{P})\).
\end{lemma}

\begin{proof}
First observe that by Lemma \ref{LemmaFromJerome} we have \(\dim(\vectSet{P})=\dim(\Q(\vectSet{P}))\), hence it suffices to prove \(\dim(\vectSet{P} \setminus \interior(\vectSet{P})) < \dim(\Q(\vectSet{P}))\). I.e.\ we have to show that \(\vectSet{P} \setminus \interior(\vectSet{P})\) is contained in a union of lower dimensional vector spaces/\(\Q\)-generated sets. To this end, consider any vector \(\vect{v} \in \vectSet{P} \setminus \interior(\vectSet{P})\). 

Write \(\vect{v}=\sum_{\vect{f} \in \vectSet{F}} \lambda_{\vect{f}} \vect{f}\). Since \(\vect{v}\) is not interior, there exists a \(\vect{x} \in \vectSet{P}\) (in fact \(\vect{x} \in \vectSet{F})\) s.t. there is no scalar \(n \in \N\) with \(n \vect{v} \geq \vect{x}\). Let \(\vectSet{F}':=\{\vect{f} \in \vectSet{F} \mid \lambda_{\vect{f}} \geq 1\}\) be the set of coefficients used for \(\vect{v}\). Then \(\vect{x} \not \in \Q(\vectSet{F}')\) by the above: Otherwise such a scalar \(n \in \N\) would exist. 

Hence \(\dim(\Q(\vectSet{F}'))<\dim(\Q(\vectSet{F}')+\Q(\vect{x}))\leq \dim(\Q(\vectSet{P}))\). 

Since we chose \(\vect{v} \in \vectSet{P} \setminus \interior(\vectSet{P})\) arbitrary, we obtain 
\[\vectSet{P} \setminus \interior(\vectSet{P}) \subseteq \bigcup_{\vectSet{F}' \subseteq \vectSet{F}, \dim(\Q(\vectSet{F}'))<\dim(\vectSet{P})} \Q(\vectSet{F}'),\]
which has lower dimension than \(\vectSet{P}\).
\end{proof}

The second lemma states a dichotomy: If \(\vectSet{P}' \subseteq \vectSet{P}\), then either \(\interior(\vectSet{P}') \subseteq \interior(\vectSet{P})\), or their intersection is empty. The lemma even states \(\vectSet{P}' \cap \interior(\vectSet{P})\) being empty.

\begin{lemma}\label{LemmaInteriorDichotomy}
Let \(\vectSet{P}' \subseteq \vectSet{P}\) be \(\N\)-generated sets. Then either \(\interior(\vectSet{P}') \subseteq \interior(\vectSet{P})\) or \(\vectSet{P}' \cap \interior(\vectSet{P})=\emptyset\).
\end{lemma}

\begin{proof}
Assume that \(\vectSet{P}' \cap \interior(\vectSet{P})\neq \emptyset\). We have to prove that \(\interior(\vectSet{P}') \subseteq \interior(\vectSet{P})\). Let \(\vect{w} \in \interior(\vectSet{P}')\). It suffices to prove \(\vect{w} \in \interior(\vectSet{P})\). To prove this, we have to take an arbitrary \(\vect{x} \in \vectSet{P}\) and prove that there exists \(m \in \N\) s.t. \(m \vect{w}-\vect{x} \in \vectSet{P}\). 

Since \(\vectSet{P}' \cap \interior(\vectSet{P}) \neq \emptyset\), there exists \(\vect{v} \in \vectSet{P}' \cap \interior(\vectSet{P})\). By definition of \(\interior(\vectSet{P})\) there exists \(m_1 \in \N\) s.t. \(m_1 \vect{v}-\vect{x} \in \vectSet{P}\). Since \(\vect{w} \in \interior(\vectSet{P}')\) and \(\vect{v} \in \vectSet{P}'\), there exists \(m_2 \in \N\) s.t. \(m_2 \vect{w} - \vect{v} \in \vectSet{P}'\). It follows that 
\[m_1m_2 \vect{w} - \vect{x}=m_1(m_2\vect{w}-\vect{v})+(m_1 \vect{v} - \vect{x})\in \vectSet{P}' + \vectSet{P} \subseteq \vectSet{P},\]

since \(\vectSet{P}' \subseteq \vectSet{P}\) are \(\N\)-generated.
\end{proof}

Now we have everything ready to prove the closure property from the main text.
\LemmaShiftGoodOverapproximation*

\begin{proof}
Part 1 (Intersection): We stated the important observation in the main text, and said it remains to prove what we can now write as \(\interior(\vectSet{L} \cap \vectSet{L}') \subseteq \interior(\vectSet{L})\).

Let \(\vectSet{X} \HybridizationRelation \vectSet{L}\) and \(\vectSet{L} \cap \vectSet{L}'\) be non-degenerate. We claim that \(\interior(\vectSet{L}) \cap \vectSet{L}'\neq \emptyset\): Assume otherwise. Then \(\vectSet{L} \cap \vectSet{L}' \subseteq \vectSet{L} \setminus \interior(\vectSet{L})\). By Lemma \ref{BasicDimensionProperties} and Lemma \ref{LemmaBoundaryLowerDimension}, we obtain 
\[\dim(\vectSet{L} \cap \vectSet{L}') \leq \dim(\vectSet{L} \setminus \interior(\vectSet{L}))<\dim(\vectSet{L})=\dim(\vectSet{L} \cap \vectSet{L}'),\]
where the last equality is due to the intersection being non-degenerate. This is a contradiction.

Hence \(\interior(\vectSet{L}) \cap \vectSet{L}' \neq \emptyset\). By the dichotomy in Lemma \ref{LemmaInteriorDichotomy}, we obtain \(\interior(\vectSet{L} \cap \vectSet{L}') \subseteq \interior(\vectSet{L})\) as required.
\end{proof}

\subsection{Appendix of Section \ref{SectionMainAlgorithm}}

\LemmaRestatableEquivalentlyZeroTest*

\begin{proof}
We remark that the lemma explicitly allows us to zerotest different coordinates in the input and output, in fact in our construction the initial configuration will be \(\vect{0}\).

Let \(m\) be the number of counters of \(\VAS\). We will create \(\VAS'\) with \(n=4m\) counters. The idea is simple: We will start at the configuration \(0^{4m}\), and write some element of \(\vectSet{R}\) onto the first \(2m\) counters. Afterwards, we will move an element of \(\vectSet{L}\) from the first \(2m\) to the last \(2m\) counters, and finish with testing the first \(2m\) counters for \(0\). I.e.\ in total we use the zero test \(WZT(\{1,\dots, 4m\}, \{1,\dots, 2m\}, 4m)\). Since we moved an element of \(\vectSet{L}\), we have indirectly checked membership in \(\vectSet{L}\), and since we wrote an element of \(\vectSet{R}\), we checked membership in \(\vectSet{R}\).

Formally: Write \(\VAS=(Q,E,\qin, \qfin)\) and \(\vectSet{L}=\vect{b}+\N(\{\vect{f}_1, \dots, \vect{f}_k\})\). We define \(\VAS'=(Q',E',q_1,q_{\vect{f}_k})\) as follows:
\begin{align*}
&Q':=\{q_1, \dots, q_m\} \cup Q \cup \{q_{\vect{f}_j} \mid 1 \leq j \leq k\} \\
&E':=\{(q_j, q_{j+1}) \mid 1 \leq j \leq m-1\} \cup E \\
&\cup \{(q_m, \qin), (\qfin, q_{\vect{f}_1})\} \cup \{(q_{\vect{f}_j}, q_{\vect{f}_{j+1}}) \mid 1 \leq j \leq k-1\} \\
&\cup \{(q_j, q_j) \mid 1 \leq j \leq m\} \cup \{(q_{\vect{f}_j}, q_{\vect{f}_j}) \mid 1 \leq j \leq k\}.
\end{align*}
Edges \(e \in E\) keep their label \(\vectSet{R}(e)\), except they now operate on the second \(m\) counters, because this is closer to the intuition above. All the other edges only perform addition. Therefore with \(\vect{a}\in \Z^{4m}\) we write \(\vectSet{R}(q,q')=\vect{a}\) to express that \(\vectSet{R}(q,q')=\{(\vect{x}, \vect{y})\in \N^{4m} \times \N^{4m} \mid \vect{y}=\vect{x}+\vect{a})\}\).

We define \(\vectSet{R}(q_j, q_{j+1})=\vectSet{R}(q_m, \qin)=\vectSet{R}(q_{\vect{f}_j}, q_{\vect{f}_{j+1}})=0^{4m}\) as well as \(\vectSet{R}(q_j,q_j)=(\vect{e}_j, \vect{e}_j, 0^m, 0^m)\), where \(\vect{e}_j \in \N^m\) is the \(j\)-th unit vector. Finally, \(\vectSet{R}(\qfin, q_{\vect{f}_1})=(-\vect{b},+\vect{b})\) and \(\vectSet{R}(q_{\vect{f}_j}, q_{\vect{f}_j})=(-\vect{f}_j,\vect{f}_j)\), remember that both \(\vect{f}_j\) and \(\vect{b}\) are vectors in \(\N^{2m}\).

The states \(q_1, q_m\) ensure that once we enter \(\qin\), we have a configuration \((\vect{x}, \vect{x}, 0^m, 0^m)\) for some \(\vect{x} \in \N^m\). To ensure this, \(q_j\) adds an arbitrary amount to counters \(j, j+m\). Then we run \(\VAS\) without change, but on the second \(m\) counters, to end up with \((\vect{x}, \vect{y}, 0^m, 0^m)\) s.t. \((\vect{x}, \vect{y}) \in \Rel(\VAS,\qin,\qfin)\).

Next we subtract an arbitrary element of \(\vectSet{L}\) from the first \(2m\) counters and readd it on the last \(2m\) counters. To do this, the states \(q_{\vect{f}_j}\) are responsible for moving period \(\vect{f}_j\) respectively, and the edge \((\qfin, q_{\vect{f}_1})\) moves \(\vect{b}\).

Regarding dimension, we observe that every state \(q \not \in Q\) has exactly one cycle, i.e.\ the corresponding cycle space dimension is \(1\). Therefore the maximal cycle space is attained in \(Q\), which is an SCC we did not touch, i.e.\ we still have \(\dim(\VAS')=\dim(\VAS)\). Moreover, clearly the aforementioned description of \(Q'\) and \(E'\) can be written in polynomial time as claimed.
\end{proof}

\subsection{Appendix of Section \ref{SectionProofTheoremEVASS}}

There are some minor checks we did not perform in the main text, in particular the following:

\begin{enumerate}
\item Why are edges inside an SCC always monotone, i.e.\ why are the new \(\VAS_j\) still \emph{monotone} \(\RelationClass\)-eVASS?
\item Why do cleaning properties not disturb the respective properties restored before?
\item Why does the rank decrease/not increase?
\end{enumerate}

Hence we have to revisit every one of the decomposition steps, and ensure these three properties.

\textbf{\COne}: Monotonicity inside SCCs is preserved as we do not change edge labels. By removing full SCCs we can at most cause a decrease in the rank.

\textbf{\CTwo}: Monotonicity inside SCCs is preserved as we explicitly required it in Definition \ref{DefinitionApproximable}, which our \(\ApproximationAlgorithm\) adheres to. Since we only change edges and not states, the rank remains the same by definition, or may decrease if approximating an edge happened to reduce cycle space dimension.

\textbf{\WeakPOne}: We only change the label of \(\vectSet{L}(e_i)\), which is not inside an SCC. Both monotonicity and rank only depend on SCCs.

\textbf{\WeakPThree}: Monotonicity is preserved since we did not change anything inside an SCC. Clearly we remain a line graph, and no edge labels have changed, preserving \CTwo. Regarding \WeakPOne \ and \WeakPTwo, write the set of solutions of \(\CharSys\) as \(\vectSet{B}+\N(\vectSet{F})\). After the replacement, we have the set of solutions \(\vectSet{B}'+\N(\vectSet{F})\), where \(\vectSet{B}' \subseteq \vectSet{B}\) contains the solutions assigning value \(f\) to \(y_{i,j}\). Since boundedness of variables only depends on \(\N(\vectSet{F})\), bounded variables remain bounded and unbounded variables remain unbounded.

Regarding the rank, observe that the rank is defined only considering non-trivial strongly-connected components, and the only SCCs we add are trivial.

\textbf{\WeakPTwo}: Deleting a counter which is already stored in the state only performs changes on edge labels. Projecting preserves monotonicity. Furthermore, we restricted the behaviour of edges, therefore the cycle space did not increase. Since we moreover left the number of states the same, the rank is preserved.

\textbf{\POne}: Let \((\VAS_i=(Q_i, E_i), \qini, \qfini)\) be an SCC in an m-eVASS s.t. some edge or period is bounded. For every edge \(e\) we write \(\vectSet{L}(e)=\vect{b}(e)+\N(\vectSet{F}(e))\). Let \(\vectSet{F}'(e) \subseteq \vectSet{F}(e)\) be the set of non-monotonicity(!) periods which are bounded, where a monotonicity period is a period \(\vect{p}\) which fulfills \(\Effect(\vect{p})=\vect{0}\). For all \(\vect{p} \in \vectSet{F}'(e)\), let \(k(\vect{p})\) be the maximal number of times this period can be used. Let \(E_i' \subseteq E_i\) be the set of edges which are bounded. For all \(e \in E_i'\), let \(k(e)\) be the maximal number of times edge \(e\) can be taken. Let \([k]:=\{0, \dots, k\}\) be the index set of \(k\), beware we start at \(0\) and end at \(k\), so we include \(k+1\) values. We use the new set of states \(Q_i':=Q_i \times W\) where
\[W:= \prod_{e \in E_i'} [k(e)] \times \prod_{e \in E_i, \vect{p} \in \vectSet{F}'(e)} [k(\vect{p})],\]
i.e.\ we track \emph{every} bounded value simultaneously in the state. The idea for the set of edges \(E_i'\) is as follows: If we take a bounded edge \(e\), we increment the corresponding counter in the state by \(1\), blocking if we are at the maximum already. Similarly, if we use a bounded period \(\vect{p}\) some \(j \leq k(\vect{p})\) number of times, then we increment the corresponding counter by \(j\), blocking if we would exceed the maximum. Formally, for every \(e=(p,q) \in E_i'\), every \((p, \vect{w}_1, \vect{w}_2) \in Q_i'\) and every vector \(\vect{w} \in \N^{\vectSet{F}'(e)}\) (counting how often we use the periods) s.t. \(\vect{w}_1(e)\leq k(e)-1\) and \(\vect{w}_2(\vect{p})+\vect{w}(\vect{p})\leq k(\vect{p})\) for all \(\vect{p} \in \vectSet{F}'(e)\), we add an edge 
\[((p, \vect{w}_1, \vect{w}_2), (q, \vect{w}_1+1_e, \vect{w}_2+\vect{w})) \in E_i',\]
where \(\vect{w}_1+1_e\) is the same as \(\vect{w}_1\), except \((\vect{w}_1+1_e)(e)=\vect{w}_1(e)+1\). We add similar edges for \(e=(p,q) \in E_i \setminus E_i'\), the difference is we do not add any \(1_e\) to \(\vect{w}_1\).

For edge labels, there is an important mistake which has to be avoided: After a degenerate intersection, we might lose the fact that \(\vectSet{R}(e) \HybridizationRelation \vectSet{L}(e)\). Since bounding a period in particular is a degenerate intersection, we will hence most certainly lose property \CTwo. Instead let \(e'=((p, \vect{w}_1, \vect{w}_2), (q, \vect{v}_1, \vect{v}_2)) \in E_i',\) and let \(e=(p,q)\in E_i\) be the edge \(e'\) originated from. We define \(\vect{b}'(e):=\vect{b}(e)+(\vect{v}_2-\vect{w}_2) \cdot \vectSet{F}'(e)\) and use the label 
\[\vectSet{R}(e'):=\left(\vect{b}'(e)+\N(\vectSet{F}(e) \setminus \vectSet{F}'(e))\right) \cap \vectSet{R}(e).\]

Finally, we set \(\qini':=(\qini, \vect{0}, \vect{0})\), and define the decomposed m-EVASS. We create for every \((\vect{w}_1, \vect{w}_2) \in W\) an m-eVASS 
\[\VAS_{i, \vect{w}_1, \vect{w}_2}:=(Q_i', E_i', \qini', (\qfini, \vect{w}_1, \vect{w}_2)).\] 
We define \(\VAS_{\vect{w}_1, \vect{w}_2}\) as an adaptation of the full m-eVASS \(\VAS\): Namely we replace \(\VAS_i\) by \(\VAS_{i, \vect{w}_1, \vect{w}_2}\). Finally we return \(\{\VAS_{\vect{w}_1, \vect{w}_2} \mid \vect{w}_1, \vect{w}_2 \text{ as above}\}\).

Edge labels \(\vectSet{R}(e')\) are still monotone since we have explicitly not removed the monotonicity periods and, as intersection of monotone relations, \(\vectSet{R}(e')\) is monotone again.

To see that the rank decreases, we refer to \cite{LerouxS19}: Leroux observed that when removing all bounded objects at once, the dimension of the vector space of cycle effects decreases. Observe that not removing the monotonicity periods does not influence the vector space of cycle effects. Since the dimension decreases, we are allowed to add as many SCCs/states of the lower dimension as we like while still decreasing the rank.

\textbf{\PTwo}: In the main text we left a lemma without proof, which we will prove here, and moreover we have to prove that coverability scales with dimension, not the number of counters.

\LemmaRestatableEquivalentPumpingCondition*

\begin{proof}
For both implications write \(\pi_{\src_i}(\sol(\CharSys(\VAS)))=\vectSet{B}+\N(\vectSet{F})\), and observe \(\pi(\vectSet{B})=\{\vect{b}'\}\).

``\(\Rightarrow\)'': Choose configuration \(\qin(\vect{b}+1^{n_i})\), and apply property \PTwo: We obtain vectors \(\vect{x}, \vect{y}\) s.t. \((\vect{x}, \vect{y}) \in \Rel(\VAS_i, \qini, \qini)\),  \(\vect{x} \in \vectSet{B}+\N(\vectSet{F})\) and \(\vect{y} \geq \vect{b}+1^{n_i}\). Defining \(\mathbf{up}:=\vect{y}-\vect{b}\) and observing that deleting counters makes existence of runs easier, we obtain \((\vect{b}', \vect{b}'+\pi(\mathbf{up}))=(\pi(\vect{x}), \pi(\vect{y})) \in \vectSet{R}(\pi(\VAS_i))\) as required.

``\(\Leftarrow\)'': By definition of monotone \(\RelationClass\)-eVASS, all edges inside an SCC have monotone semantics \(\vectSet{R}(e)\). Since \(\VAS_i\) and hence \(\pi(\VAS_i)\) is strongly-connected, \emph{every} edge has monotone semantics. This implies that we also have \((\pi(\vect{b})+m\mathbf{up}, \pi(\vect{b})+(m+1)\mathbf{up}) \in \Rel(\pi(\VAS_i), \qini, \qini)\) for all \(m \in \N\). By transitivity of \(\Rel(\pi(\VAS_i), \qini, \qini)\) we therefore obtain \((\pi(\vect{b}), \pi(\vect{b})+m\mathbf{up}) \in \Rel(\pi(\VAS_i), \qini, \qini)\) for all \(m \in \N\), reaching arbitrarily large configurations.

To see that this transfers back to \(\VAS_i\), again use monotonicity: Let \(p(\vect{x}_t)\) be any configuration of \(\VAS_i\). We have to show that we can cover it. Let \(m:=||\vect{x}_t||_{\infty}\), and let \(\rho\) be a run in \(\pi(\VAS_i)\) reaching \(\geq m\) on every counter. Write \(\rho=p_0(\vect{x}_0'), \dots, p_k(\vect{x}_k')\) with \(\vect{x}_0'=\pi(\vect{b}), p_0=\qini\) and \(p_k=p\). Since the steps \((\vect{x}_l', \vect{x}_{l+1}')\in \pi(\vectSet{R}(e_l))\) exist in the projection, there exist vectors \(\vect{x}_l, \vect{y}_l \in \N^{n_i}\) s.t. \(\pi(\vect{x}_l)=\vect{x}_l', \pi(\vect{y}_l)=\vect{x}_{l+1}'\) and \((\vect{x}_l, \vect{y}_l) \in \vectSet{R}(e_l)\). Define \(\vect{x} \in \N^{n_i}\) by \(\vect{x}[j]=\vect{x}_0'[j]\) if \(j \in I\), and \(\vect{x}[j]=\sum_{l=0}^{k-1} \vect{x}_l[j]\) otherwise. Similarly define \(\vect{y}\) by \(\vect{y}[j]=\vect{x}_k'[j]\) if \(j \in I\) and \(\vect{y}[j]=\sum_{l=1}^{k} \vect{y}_l[j]\) otherwise. In particular \(\pi(\vect{x})=\pi(\vect{b})\) and \(\pi(\vect{y})[j] \geq m\) for every \(j \in I\). By monotonicity and transitivity we have \((\vect{x}, \vect{y}) \in \Rel(\VAS_i, \qini, p)\). Since \(\pi\) only projects away coordinates which can be increased by \(\N(\vectSet{F})\), by taking every period of \(\vectSet{F}\) a total of \(\InfinityNorm{\vect{x}}\) many times, we have \(\vect{x}_{\text{new}}:=\vect{b}+\InfinityNorm{\vect{x}} \sum_{\vect{f} \in \vectSet{F}} \vect{f}\geq \vect{x}.\) By monotonicity we obtain \((\vect{x}_{\text{new}}, \vect{y}+\vect{x}_{\text{new}}-\vect{x}) \in \Rel(\VAS_i, \qini, p)\), obtaining a run in \(\VAS_i\) starting in \(\pi_{\src_i}(\sol(\CharSys(\VAS)))\) and reaching large values as claimed.
\end{proof}

\textbf{Coverability depending on d}: It remains to prove that as claimed in the main text, the complexity of the backwards coverability algorithm scales with the dimension \(d\), not the number of counters \(n\). Let \(\vectSet{BReach}\) be the set of backwards reachable vectors. 

Our technique is as follows: In order to bound the complexity of the backwards coverability algorithm, we will exhibit a semilinear set \(\vectSet{S}\) s.t. \(\vectSet{BReach} \subseteq \vectSet{S}\), and an injective function \(\rank' \colon \vectSet{S} \to F \times \N^d\) for some finite set \(F\), which is bounded by an elementary function. We order \(F \times \N^d\) pointwise on \(\N^d\) and on \(F\) we use equality. We require \(\rank'\) to have the following property, which one might call a form of order-preserving:

\begin{enumerate}
\item[(OP)] Let \(\vect{x}, \vect{y} \in \vectSet{S}\) with \(\rank'(\vect{x})\leq \rank'(\vect{y})\). Then \(\vect{x} \leq \vect{y}\).
\end{enumerate}

Observe that this would suffice to prove the claimed complexity bound of the backwards coverability algorithm: The sequence of vectors \(\vect{x}_1, \vect{x}_2,\dots\) added by the backwards coverability algorithm has no increasing pair. Since \(\vectSet{BReach} \subseteq \vectSet{S}\), we can consider the sequence \(\rank'(\vect{x}_1), \rank'(\vect{x}_2),\dots\), which by (OP) again contains no increasing pair. Hence in order to bound the length of the sequence \(\vect{x}_1, \vect{x}_2, \dots\), we bound the length of \(\rank'(\vect{x}_1), \rank'(\vect{x}_2), \dots\) instead, for which we can now apply Proposition \ref{PropositionFastGrowingComplexity} with \(k=d\) to obtain \(\mathfrak{F}_{\alpha+d+1}\).

Therefore it suffices to define \(\vectSet{S}\) and \(\rank'\). Let \(\vectSet{V}\) be the vector space of cycle effects. We have \(\dim(\vectSet{V}) \leq d\), and for all states \(q \in Q_i\), there is a vector \(\vect{v}_q \in \Q^n\) s.t. whenever we enter state \(q\), our counter valuation is in \(\vect{v}_q+\vectSet{V}\). Hence we have \(\vectSet{BReach} \subseteq (\{\vect{v}_q \mid q\in Q_i\} + \vectSet{V}) \cap \N^n=:\vectSet{S}\). Observe that \(\vectSet{S}\) is semilinear. Furthermore, we have \(\dim(\vectSet{S}) \leq \dim(\vectSet{V}) \leq d\) by Lemma \ref{BasicDimensionProperties}. 

It remains to define the function \(\rank'\), and prove (OP). We use a theorem by Ginsburg and Spanier \cite{ginsburg1963bounded}, or more precisely we use \cite[Thm.~12]{ChistikovH16}, where the exponential bound for the involved objects was proven:

\begin{theorem}
\cite[Thm.~12]{ChistikovH16} Let \(\vectSet{S}\) be any semilinear set. 

One can in exponential time compute an \emph{unambiguous decomposition} of \(\vectSet{S}\), i.e.\ a semilinear description \(\vectSet{S}=\bigcup_{j=1}^k \vect{b}_j + \N(\vectSet{F}_j)\) s.t. the union is disjoint, and \(|\vectSet{F}_j|=\dim(\Q(\vectSet{F}_j))\), i.e.\ the sets \(\vectSet{F}_j\) of periods are linearly independent.
\end{theorem}

The reason this decomposition is called unambiguous is that every point \(\vect{x} \in \vectSet{S}\) can then be uniquely represented by the index \(j\) s.t. \(\vect{x} \in \vect{b}_j+\N(\vectSet{F}_j)\), together with the coefficients \(\lambda_1, \dots, \lambda_{|\vectSet{F}_j|} \in \N\) of the unique linear combination. 

In our case, we compute an unambiguous decomposition of \(\vectSet{S}\), the overapproximation of \(\vectSet{BReach}\). Since \(\dim(\vectSet{S}) \leq d\), we obtain \(\vectSet{S}=\bigcup_{j=1}^k \vect{b}_j + \N(\vectSet{F}_j)\) where \(|\vectSet{F}_j| \leq d\ \forall j\). This allows us to represent any \(\vect{x} \in \vectSet{S}\) uniquely by \(\rank'(\vectSet{S})=(j, \lambda_1, \dots, \lambda_{|\vectSet{F}_j|}, 0, \dots, 0) \in \{1,\dots, k\} \times \N^d\), i.e.\ the representation as above potentially padded with \(0\)'s if \(|\vectSet{F}_j| < d\).

\textbf{Proof of (OP)}: The function \(\rank'\) is injective since the decomposition of \(\vectSet{S}\) is unambiguous, and if \(\rank'(\vect{x}) \leq \rank'(\vect{y})\), then \(\vect{y} \in \vect{x}+\N(\vectSet{F}_j)\), where \(j\) is the first entry of \(\rank'(\vect{x})\). In particular, \(\vect{x} \leq \vect{y}\), proving (OP).

\textbf{\PTwo \ Rank Decrease}: To see that the rank decreases, observe first of all that since \WeakPTwo \ holds, there is a cycle with non-zero effect on the deleted counter. We can proceed exactly as in \cite{LerouxS19}: Clearly cycles in the new m-eVASS are cycles in the old m-eVASS with effect \(0\) on the deleted counter. Hence the new vector space of cycle effects is contained in the old vector space. Since the cycle with non-zero effect on the deleted counter which existed before does not exist anymore, the vector space and therefore dimension strictly decrease. Therefore it does not matter how many SCC/states we create, we decrease in the lexicographic ordering.

Projecting counters clearly preserves monotonicity. 

%\textbf{Preserving Monotonicity in Output}: Remember Definition \ref{DefinitionApproximable}, part 2): We have to ensure that if we try to approximate a monotone relation \(\vectSet{X}\), the the resulting \(\vectSet{X}_j\) are again monotone. First of all, we want to assure the reader of the following: Even if this did not hold, it would be easy to readd monotonicity in a last postprocessing step: Simply double the dimension and add a new first relation \(\vectSet{X}_{in}\) and a new last relation \(\vectSet{X}_{out}\) to any \(\RelationClass\)-KLM sequence in the output. In \(\vectSet{X}_{in}\) move an arbitrary amount of value from the first \(d\) to the last \(d\) counters, and in the last component \(\vectSet{X}_{out}\) move the values back, checking that the auxiliary counters are \(0\). Clearly this would still be perfect.
%
%However, analyzing the steps of the algorithm, none of them jeopardizes that the relation is monotone: The characteristic system only consider \emph{effects}, and we remove certain \emph{effects}. But effects are independent of monotonicity. On the other hand, deleting a counter if it is bounded again does not jeopardize this property, and similarly for replacing edge labels in P1.

\subsection{Example that asymptotic overapproximations have to be basic}
In \CTwo, we required that the asymptotic overapproximations of edges \(e\) have to be basic. This is necessary in order for the algorithm to be correct. If some edge \(e\) had an asymptotic overapproximation which is not basic, then \(\pi_{\src, \tgt}(\sol(\CharSys))\) might not be an asymptotic overapproximation even if the m-eVASS is perfect. As a counter example we use the \(2\)-dimensional monotone \(Semil\)-eVASS in Figure \ref{FigureExamplePropertyOne}.

\begin{figure}[h!]
\begin{centering}
\begin{tikzpicture}
		%\tikzset{every edge/.append style={font=\large}}
		
		\newcommand*{\distancesubx}{2.5cm}
	
		% States
		\node[place, double] (A) at (1,0) {q};
		\node[place, double] (B) at (8, 0) {p};
		\node[white!100] (C) at (0,0) {};
		
		% Edges
		\path[->, thick, out=30, in=150, looseness=0.6] (A) edge[] node[above] {\(\vectSet{R}(e_1)=\text{dec}(x)+\N(\{\text{inc}(x,y),\text{dec}(x,y),\text{dec}(x)\})+\text{Id}\)} (B);
		\path[->, thick, out=210, in=-30, looseness=0.6] (B) edge[] node[below] {\(\vectSet{R}(e_2)=\vectSet{L}(e_2)=\text{dec}(y)+\text{Id}\)} (A);
		\path[->, thick] (C) edge[] (A);
\end{tikzpicture}
\end{centering}

\caption{Example that asymptotic overapproximations have to be basic.}\label{FigureExamplePropertyOne}
\end{figure}

The labels of the edges are respectively increments and decrements of variables only. In particular \(\text{inc}(x,y):=((0,0),(1,1))\), \(\text{dec}(x,y):=((1,1), (0,0))\), \(\text{dec}(x):=((1,0),(0,0))\) and \(\text{dec}(y):=((0,1), (0,0))\) increment and decrement the respective variables. Finally, adding \(\text{Id}:=\N(\{((1,0),(1,0)), ((0,1),(0,1))\})\) guarantees monotonicity. Edge \(e_2\) is basic, while for edge \(e_1\) we have \(\vectSet{L}(e_1)=\N(\{\text{inc}(x,y),\text{dec}(x,y),\text{dec}(x)\})+\text{Id}\), i.e.\ the the decrement on \(x\) does not have to be performed. The source configuration is \(q(1,1)\) and the target is \(q(2,1)\). 

Intuition: The only reason for requiring a \(2\)-dimensional example and adding the ``\(\text{inc}(x,y), \text{dec}(x,y)\)'' parts to \(e_1\) is to ensure that \PTwo \ holds, because we can increment and decrement both variables in tandem. The actual basis of the example however is \(1\)-dimensional with the value \(y-x\): Since \(\vectSet{R}(e_1) \HybridizationRelation \vectSet{L}(e_1)\) is not basic, \(\CharSys\) believes that using \(e_1 e_2\) without any periods decrements \(y\) without decrementing \(x\), therefore decreasing \(y-x\). However, in actuality \(x\) would also be decremented and \(y-x\) would therefore remain the same. This leads to \(\CharSys\) not detecting that the value \(y-x\) cannot be decreased, claiming reachability.

Formally: Clearly all properties except \CTwo \ hold. However, there is a semilinear inductive invariant \(\vectSet{S}=q:\{(x,y) \mid y-x \geq 0\}, p: \{(x,y) \mid y-x \geq 1\}\) proving non-reachability.

\subsection{Appendix of Section \ref{SectionHybridization}}

This section of the appendix is split into two parts. First we have to introduce some theory regarding \emph{directed hybridlinear} sets, a subclass of hybridlinear sets which are close to linear sets, but admits a closure property which the class of linear sets does not possess (Lemma \ref{LemmaDirectedHybridlinearNondegenerateIntersection}). To this end, we will introduce multiple equivalent definitions of directed hybridlinear. This is done in part 1. 

In part 2, we then introduce heavy theory from the VASS community to prove Theorem \ref{TheoremStronglyApproximableEVASS}. In particular we will define \(\HybridizationRelation_s\).

\subsection{Appendix of Section \ref{SectionHybridization}, Part 1}

In this part we define directed hybridlinear sets and prove basic properties for this class. Towards this end, we recall some lemmas regarding \(\mathbb{S}\)-finitely generated (f.g.) sets for \(\mathbb{S} \in \{\N, \Q_{\geq 0}, \Z, \Q\}\). Recall that \(\mathbb{S}(\vectSet{F}):=\{ \sum_{i=1}^m \lambda_i \vect{f}_i \mid m \in \N, \vect{f}_i \in \vectSet{F}, \lambda_i \in \mathbb{S}\}\) and for (finite) $\vectSet{F}$ we say that the set \(\mathbb{S}(\vectSet{F})\) is \(\mathbb{S}\)-(finitely) generated. Also recall that a partial order \((\vectSet{X}, \leq_{\vectSet{X}})\) is a \emph{well-quasi-order} (wqo) if \(\leq\) is well-founded and every subset \(\vectSet{U} \subseteq \vectSet{X}\) has finitely many minimal elements. The most famous example of a wqo is \((\N^n, \leq_{\N^n})\), where \(\leq_{\N^n}\) is the component-wise \(\leq\) ordering.

For \(\mathbb{S} \in \{\Z, \Q\}\) every \(\mathbb{S}\)-g. set is \(\mathbb{S}\)-f.g., because the whole space \(\Q^n\) or respectively \(\Z^n\) is finitely generated, and moving to a substructure does not increase the number of necessary generators. For \(\Q\)-generated sets, this is well-known linear algebra, for \(\Z\)-generated sets this follows from the existence of the hermite normal form for integer matrices (see \cite{LinearProgramming}, Chapter 4). For \(\mathbb{S} \in \{\N, \Q_{\geq 0}\}\) the picture is different. 

For \(\mathbb{S} = \Q_{\geq 0}\) we have:

\begin{lemma}\cite[Cor. 7.1a]{LinearProgramming}
Let \(\vectSet{C} \subseteq \Q^n\) be a \(\Q_{\geq 0}\)-g. set.

Then \(\vectSet{C}=\{\vect{x} \in \Q^n \mid A \vect{x} \geq \vect{0}\}\) is the preimage of \(\Q_{\geq 0}^{n'}\) for some matrix \(A \in \Z^{n' \times n}\) iff \(\vectSet{C}\) is \(\Q_{\geq 0}\)-finitely generated.\label{LemmaFinitelyGeneratedCone}%
\end{lemma}

The essential idea behind this lemma is the following: A \(\Q_{\geq 0}\)-f.g. set (or equivalently the set of points on some ray from the top of a pyramid to the base) can be classified depending on what the base of the pyramid looks like. For an ice cream cone this is a circle, for the pyramids in Egypt a square. The lemma classifies pyramids where the base is a polygon: The base of the pyramid has finitely many vertices (generators of the cone) if and only if the base has finitely many faces/side edges. Every side edge gives rise to an inequality \(\vect{a} \vect{x} \geq 0\), which we can rescale to obtain \(\vect{a} \in \Z^n\).

For \(\mathbb{S} = \N\)  in \cite{Leroux13} Leroux proved a similar characterization.

\begin{definition}
Let \(\vectSet{P}\) be \(\N\)-g.. The canonical partial order on \(\vectSet{P}\) is \(\leq_{\vectSet{P}}\) defined via \(\vect{x} \leq_{\vectSet{P}} \vect{y}\) if \(\vect{y}=\vect{x}+\vect{p}\) for some \(\vect{p}\in \vectSet{P}\).
\end{definition}

\begin{lemma}\cite[Lemma V.5]{Leroux13}
Let \(\vectSet{P} \subseteq \N^n\) be \(\N\)-g.. T.F.A.E.:

\begin{enumerate}
\item[(1)] \(\vectSet{P}\) is \(\N\)-finitely generated.
\item[(2)] \((\vectSet{P}, \leq_{\vectSet{P}})\) is a wqo. 
\item[(3)] \(\Q_{\geq 0} \cdot \vectSet{P}\) is \(\Q_{\geq 0}\)-finitely generated.
\end{enumerate} \label{LemmaFinitelyGeneratedPeriodicSets}
\end{lemma}

\begin{proof}[Proof sketch]
(1) \(\Rightarrow\) (3): Obvious (use the same generators).

 (3) \(\Rightarrow\) (2): (Only intuition): We show this by providing an order isomorphism \((\vectSet{P}, \leq_{\vectSet{P}}) \simeq (\N^{n'}, \leq)\) to the wqo on \(\N^{n'}\).
 
The ordering \(\leq_{\vectSet{P}}\) (up to minor details) fulfills \(\vect{x} \leq_{\vectSet{P}} \vect{y}\) iff \(A \vect{x} \leq A \vect{y}\), where \(A\) is the matrix of Lemma \ref{LemmaFinitelyGeneratedCone} for \(\Q_{\geq 0} \cdot \vectSet{P}\).
We have \(A\vect{x}\) and \(A\vect{y}\) in \(\N^{n'} = \Z^{n'} \cap \Q_{\geq 0}^{n'}\). Indeed, \(A\vect{x}\) is an integer because we performed multiplication and addition of integers. Additionally $\vect{x} \in \vectSet{P} \subseteq \Q_{\geq 0} \cdot \vectSet{P}\) is in the preimage of \(\Q_{\geq 0}^{n'}\) by choice of \(A\). Hence multiplication by \(A\) is the required order isomorphism. Intuitively, this isomorphism shows that $\leq_{\vectSet{P}}$ orders the points with respect to the distance to the borders of the \(\Q_{\geq 0}\)-f.g. set \(\Q_{\geq 0} \cdot \vectSet{P}\): the larger the distances to the borders the larger the point w.r.t. $\leq_{\vectSet{P}}$.
 
 (2) \(\Rightarrow\) (1): As \((\vectSet{P}, \leq_{\vectSet{P}})\) is a wqo, \(\vectSet{P} \setminus \{\vect{0}\}\) has finitely
 many minimal elements w.r.t. \(\leq_{\vectSet{P}}\). These generate $\vectSet{P}$.
\end{proof}

So far these results are known, we now utilize the ordering \(\leq_{\vectSet{P}}\) of Lemma \ref{LemmaFinitelyGeneratedPeriodicSets}(2) to classify semilinear sets.

\begin{restatable}{definition}{DefinitionRestatablePreservants} \label{DefinitionPreservants}
Let \(\vectSet{X} \subseteq \N^n\). A vector \(\vect{p}\) is a \emph{preservant} of \(\vectSet{X}\) if \(\vectSet{X}+\vect{p} \subseteq \vectSet{X}\). The set of all preservants is denoted \(\PX\).
\end{restatable}

Clearly \(\PX\) is \(\N\)-generated, i.e.\ closed under addition, for every set \(\vectSet{X}\). In many cases however, \(\PX\) is very small compared to \(\vectSet{X}\), see the left of Figure \ref{FigureIntuitionCones}. \(\PX\) being large and finitely generated in fact characterizes linear sets and more generally hybridlinear sets.

\begin{figure}[h!]
\begin{minipage}{4.5cm}
\begin{tikzpicture}
\begin{axis}[
    axis lines = left,
    xlabel = { },
    ylabel = { },
    xmin=0, xmax=8,
    ymin=0, ymax=8,
    xtick={0,2,4,6,8},
    ytick={0,2,4,6,8},
    ymajorgrids=true,
    xmajorgrids=true,
    thick,
    smooth,
    no markers,
]

\addplot+[
    name path=A,
    color=blue,
]
coordinates {(0,0) (8,8)};

\addplot+[
    name path=B,
    domain=0:8,
    color=blue,
]
{log2(x+1)};

\addplot[blue!40] fill between[of=A and B];

\addplot[
fill=red,
fill opacity=0.7,
only marks,
]
coordinates {
(0,0)(1,0)(2,0)(3,0)(4,0)(5,0)(6,0)(7,0)(8,0)
};

%\addplot[
%    fill=green,
%    fill opacity=0.7,
%    only marks,
%    ]
%    coordinates {
%    (3,4)(3,5)(3,6)(3,7)(3,8)(4,4)(4,5)(4,6)(4,7)(4,8)(5,4)(5,5)(5,6)(5,7)(5,8)(6,4)(6,5)(6,6)(6,7)(6,8)(7,4)(7,5)(7,6)(7,7)(7,8)(8,4)(8,5)(8,6)(8,7)(8,8)
%    };

\end{axis}
\end{tikzpicture}
\end{minipage}%
\begin{minipage}{4.5cm}
\begin{tikzpicture}
\begin{axis}[
    axis lines = left,
    xlabel = { },
    ylabel = { },
    xmin=0, xmax=4,
    ymin=0, ymax=8,
    xtick={0,1,2,3,4},
    ytick={0,2,4,6,8},
    ymajorgrids=true,
    xmajorgrids=true,
    thick,
    smooth,
    no markers,
]
    
\addplot+[
    name path=A,
    color=blue,
    thick,
    ]
    coordinates {
    (0,0)(4,8)
    };
    
\addplot+[
    name path=B,
    color=blue,
    thick,
]
coordinates {
    (0,0)(4,0)
    };
    
\addplot[blue!40] fill between[of=A and B];

\addplot+[
    name path=C,
    color=red,
    thick,
    ]
    coordinates {
    (0,1)(4,5)
    };
    
\addplot+[
    name path=D,
    color=red,
    thick,
]
coordinates {
    (0,1)(4,1)
    };
    
\addplot[red!40] fill between[of=C and D];

\addplot[
    fill=black,
    fill opacity=0.7,
    only marks,
    ]
    coordinates {
    (1,1)(1,2)
    };

\end{axis}
\end{tikzpicture}
\end{minipage}%

\caption{\textit{Left}: The set \(\vectSet{X}:=\{(x,y) \in \N^2 \mid x \geq y \geq \log_2(x+1)\}\) (blue region) \(\cup \)\{x-axis\} (red) fulfills \(\PX=\{\vect{0}\}\). Namely vectors \((x,y)\) with \(y>0\) are \(\not \in \PX\) because of the red points, and for \(x>0, y=0\) the reason is the blue region.  \newline
\textit{Right}: The blue \(\{(x,y) \in \N^2 \mid 0 \leq y \leq 2x\}\) and the red \(1 \leq y \leq 1+x\) regions depict linear sets. \(\vectSet{L}_1 \cap \vectSet{L}_2\) is not linear anymore, as it requires both the black points as base points.}\label{FigureIntuitionCones}
\end{figure}

\begin{restatable}{proposition}{PropositionCharacterizeHybridlinear}
Let \(\vectSet{X}\subseteq \N^n\) be any set. 

Then \(\vectSet{X}\) is hybridlinear if and only if \((\vectSet{X}, \leq_{\PX})\) is a wqo. \label{PropositionCharacterizeHybridlinear}
\end{restatable}

\begin{proof}
If \(\vectSet{X}=\emptyset\) then both statements trivially hold. In the sequel we hence assume \(\vectSet{X} \neq \emptyset\).

``\(\Leftarrow\)'': Let \(\vect{x} \in \vectSet{X}\). Since \((\vectSet{X}, \leq_{\PX})\) is a wqo, also \((\vect{x}+\PX, \leq_{\PX})\) is a wqo. This is isomorphic to \((\PX, \leq_{\PX})\), hence by Lemma \ref{LemmaFinitelyGeneratedPeriodicSets} \(\PX\) is finitely generated. Let \(\vect{b}_1, \dots, \vect{b}_r\) be the minimal elements of \(\vectSet{X}\) w.r.t. \(\leq_{\PX}\), observe that \(\vectSet{X}=\{\vect{b}_1, \dots, \vect{b}_r\}+\PX\) and we are done.

``\(\Rightarrow\)'': Write \(\vectSet{X}=\{\vect{b}_1, \dots, \vect{b}_r\}+\N(\vectSet{F})\). Write \(\vectSet{P}:=\N(\vectSet{F})\). We will first show that \((\PX, \leq_{\PX})\) is a wqo by proving the claim \(\Q_{\geq 0} \cdot \vectSet{P}=\Q_{\geq 0} \cdot \PX\) and then using Lemma \ref{LemmaFinitelyGeneratedPeriodicSets}.

By definition of \(\PX\) we have \(\vectSet{P} \subseteq \PX\). Observe that the other inclusion is not always true, as \( \vectSet{X}=\{0,1\}+2\N\) fulfills \(\PX=\N\). On the other hand we claim that \(r! \cdot \PX \subseteq \vectSet{P}\), where \(r!\) denotes the factorial of \(r\), which would finish the proof of \(\Q_{\geq 0} \vectSet{P}=\Q_{\geq 0} \PX\). 

Let \(\vect{p} \in \PX\) be any vector. Consider the map \(f_{\vect{p}} \colon \vectSet{X} \to \vectSet{X}, \vect{x} \mapsto \vect{x}+\vect{p}\). We will prove that this map induces a map \(\tau \colon \{1,\dots, r\} \to \{1,\dots, r\}\). The important observation is that if we have \(f_{\vect{p}}(\vect{b}_i) \in \vect{b}_j +\vectSet{P}\), then in fact also \(f_{\vect{p}}(\vect{b}_i+\vectSet{P}) \subseteq \vect{b}_j + \vectSet{P}\), since 

\(f_{\vect{p}}(\vect{b}_i+\vectSet{P})=f_{\vect{p}}(\vect{b}_i)+\vectSet{P} \subseteq (\vect{b}_j + \vectSet{P})+\vectSet{P} \subseteq \vect{b}_j+\vectSet{P}\).

For every \(i\), since \(f_{\vect{p}}(\vect{b}_i) \in \vectSet{X}\), there is \(\tau(i)=j\) with \(f_{\vect{p}}(\vect{b}_i) \in \vect{b}_j + \vectSet{P}\). Make any such choice of \(\tau\), thereby defining a map \(\tau \colon \{1,\dots, r\} \to \{1,\dots, r\}\). Since \(\tau\) is a map between finite sets, \(\tau\) has some cycle. Let \(m\) be the length of this cycle, and \(i\) some element of the cycle. Since \(m\)-fold application of \(f_{\vect{p}}\) causes \(\vect{b}_i\) to move back into \(\vect{b}_i+\vectSet{P}\), we obtain \(\vect{b}_i+m \vect{p} \in \vect{b}_i+\vectSet{P}\), which implies \(m \vect{p} \in \vectSet{P}\) and also \(r! \vect{p} \in \vectSet{P}\). The claim is therefore proven.

By Lemma \ref{LemmaFinitelyGeneratedPeriodicSets} we obtain the following implications: \(\vectSet{P}=\N(\vectSet{F})\) is \(\N\)-f.g. implies \(\Q_{\geq 0}\vectSet{P}=\Q_{\geq 0} \PX\) is \(\Q_{\geq 0}\)-f.g. implies that \((\PX, \leq_{\PX})\) is a wqo. Observe that \((\PX, \leq_{\PX}) \simeq (\vect{b}_i+\PX, \leq_{\PX})\) for all \(i\), because adding the vector \(\vect{b}_i\) on both sides does not influence the definition \(\vect{y}=\vect{x}+\vect{p}\). Since finite unions of wqo's are wqos, we obtain that \((\vectSet{X}, \leq_{\PX})\) is a wqo as claimed.
\end{proof}

\begin{corollary}
A set \(\vectSet{X}\) is linear if and only if \((\vectSet{X}, \leq_{\PX})\) is a wqo with a unique minimal element.
\end{corollary}

However, rather than a unique minimal element, an important condition for a wqo is \emph{directedness}.

\begin{definition}
Let \((\vectSet{X}, \leq_{\vectSet{X}})\) be a wqo. We say that \((\vectSet{X}, \leq_{\vectSet{X}})\) is \emph{directed} if for all \(\vect{x},\vect{y} \in \vectSet{X}\), there exists \(\vect{z} \in \vectSet{X}\) such that \(\vect{x} \leq_{\vectSet{X}} \vect{z}\) and \(\vect{y} \leq_{\vectSet{X}} \vect{z}\). 

\(\vectSet{X}\) is \emph{directed hybridlinear} if \((\vectSet{X}, \leq_{\PX})\) is a directed wqo.
\end{definition}

A good intuition is that a set \(\vectSet{L}\) is directed hybridlinear if it is equal to a linear set minus finitely many points. Though as a formal statement this is only true in dimension \(2\): Starting in dimension \(3\) for example the set \(\{(1,1,0), (1,0,1),(0,1,1)\}+\N^3\) is directed hybridlinear, but missing three lines compared to the linear set \(\N^3\). In dimension \(4\) we can remove planes from \(\N^4\) and so on.

Next we give equivalent definitions of directed hybridlinear, preventing descriptions like \(\N=\{0,1\}+(2\N)\) in the sequel.

\begin{restatable}{lemma}{LemmaDirectedHybridlinearEquivalence}\label{LemmaDirectedHybridlinearEquivalence}
Let \(\vectSet{X}\) be a hybridlinear set. T.F.A.E.:
\begin{enumerate}
\item[(1)] \(\vectSet{X}\) is directed hybridlinear.
\item[(2)] There exists a representation \(\vectSet{X}=\{\vect{b}_1, \dots, \vect{b}_r\}+\N(\vectSet{F})\) s.t. \((\vect{b}_i+\N(\vectSet{F})) \cap (\vect{b}_j+\N(\vectSet{F})) \neq \emptyset\) for all \(1 \leq i,j \leq r\).
\item[(3)] There exists a representation \(\vectSet{X}=\{\vect{b}_1, \dots, \vect{b}_r\}+\N(\vectSet{F})\) s.t. \(\vect{b}_i-\vect{b}_j \in \Z(\vectSet{F})\) for all \(1 \leq i,j \leq r\).
\end{enumerate}
\end{restatable}

\begin{proof}
For \(\vectSet{X}=\emptyset\) all statements hold. Hence assume \(\vectSet{X} \neq \emptyset\).

(1) \(\Rightarrow\) (2): Since \((\vectSet{X}, \leq_{\PX})\) is a wqo, \(\vectSet{X}\) has finitely many minimal elements \(\vect{b}_1, \dots, \vect{b}_r\). We have \(\vectSet{X}=\{\vect{b}_1, \dots, \vect{b}_r\}+\PX\). We claim that this representation fulfills 2).

Proof of claim: Observe first that since \((\vectSet{X}, \leq_{\PX})\) is a wqo, also \((\PX, \leq_{\PX})\) is a wqo, and hence by Lemma \ref{LemmaFinitelyGeneratedPeriodicSets} \(\PX\) is \(\N\)-f.g.. Using that \((\vectSet{X}, \leq_{\PX})\) is directed for each pair \((\vect{b}_i, \vect{b}_j)\) with \(1 \leq i < j \leq r\), we obtain elements \(\vect{z}_{i,j}\) such that \(\vect{b}_i \leq_{\PX} \vect{z}_{i,j}\) and \(\vect{b}_j \leq_{\PX} \vect{z}_{i,j}\), i.e.\ we have \(\vect{b}_i+\vect{p}=\vect{z}_{i,j}=\vect{b}_j+\vect{p}'\) for some \(\vect{p}, \vect{p}' \in \PX\). Hence \(\vect{z}_{i,j} \in (\vect{b}_i + \PX) \cap (\vect{b}_j + \PX)\), proving non-emptiness.

(2) \(\Rightarrow\) (3): Let \(\{\vect{b}_1, \dots, \vect{b}_r\}+\N(\vectSet{F})\) be such a representation. We claim that the same representation works for 3). Let \(1 \leq i,j \leq r\). By (2) there exists \(\vect{z}_{i,j} \in (\vect{b}_i + \N(\vectSet{F})) \cap (\vect{b}_j+\N(\vectSet{F}))\). I.e. \(\vect{z}_{i,j}-\vect{b}_i \in \N(\vectSet{F})\) and \(\vect{b}_j - \vect{z}_{i,j} \in - \N(\vectSet{F})\). Hence \(\vect{b}_j - \vect{b}_i=(\vect{b}_j - \vect{z}_{i,j})+(\vect{z}_{i,j}-\vect{b}_i) \in \Z(\vectSet{F})\).

(3) \(\Rightarrow\) (1): Let \(\vectSet{X}=\{\vect{b}_1, \dots, \vect{b}_r\}+\N(\vectSet{F})\) be such a representation. Let \(\vect{x}=\vect{b}_i+\vect{p}_{\vect{x}}\) and \(\vect{y}=\vect{b}_j+\vect{p}_{\vect{y}}\) be arbitrary points in \(\vectSet{X}\). First observe that \(\N(\vectSet{F}) \subseteq \PX\), hence it is enough to show that there exists a point \(\vect{z} \in \vectSet{X}\) with \(\vect{x} \leq_{\N(\vectSet{F})} \vect{z}\) and \(\vect{y} \leq_{\N(\vectSet{F})} \vect{z}\). 

Write \(\vect{b}_j - \vect{b}_i=\vect{p}_+ - \vect{p}_-\) with \(\vect{p}_+, \vect{p}_- \in \N(\vectSet{F})\).
Therefore \(\vect{b}_i+\vect{p}_+ = \vect{b}_j+\vect{p}_-\).
Define \(\vect{z}:=\vect{b}_i+\vect{p}_{\vect{x}}+\vect{p}_{\vect{y}} +\vect{p}_+ \in \vectSet{X}\). We have \(\vect{z}=\vect{x}+(\vect{p}_{\vect{y}} + \vect{p}_+) \geq_{\N(\vectSet{F})} \vect{x}\). We also have \(\vect{z}=\vect{y}+(\vect{p}_{\vect{x}}+\vect{p}_-) \geq_{\N(\vectSet{F})} \vect{y}\).
\end{proof}

For an example of a directed hybridlinear set which is not linear, see \(\vectSet{L}_1 \cap \vectSet{L}_2\) on the right of Figure \ref{FigureIntuitionCones}. In fact the figure even shows that a non-degenerate intersection [remember this means that \(\dim(\vectSet{L}_1 \cap \vectSet{L}_2)=\dim(\vectSet{L}_1)=\dim(\vectSet{L}_2)\)] of linear sets is not necessarily linear anymore. On the other hand, an example of a hybridlinear set which is not directed is the union of two parallel lines, for example \(\{(0,0), (0,1)\}+ \N (1,0) \subseteq \N^2\). In Primitive 2, such hybridlinear sets would create problems, since the disjoint components cannot interact, and should be treated differently.

Remember the main goal of this part was to obtain the closure property under non-degenerate intersection. We only need one last lemma from \cite{GuttenbergRE23}.

\begin{lemma}[\cite{GuttenbergRE23}, Prop. 3.9 (3.), special case of finitely generated]
Let \(\vectSet{P}:=\N(\vectSet{F})\) and \(\vectSet{P}':=\N(\vectSet{F}')\) be \(\N\)-f.g. sets with a non-degenerate intersection. Then \(\Z(\vectSet{F}) \cap \Z(\vectSet{F}')=\Z(\vectSet{P} \cap \vectSet{P}')\).\label{LemmaFinitelyGeneratedNondegenerateIntersection}%
\end{lemma}

Essentially this says that if the sets \(\vectSet{P}\) and \(\vectSet{P}'\) are ``similar enough'', in the sense of having a non-degenerate intersection, then the  \(\Z\)-generated set of the intersection is the intersection of the original \(\Z\)-generated sets. Clearly this will not be correct otherwise, simply intersect sets \(\vectSet{P}:=\N, \vectSet{P}':=-\N\) ``in opposite directions'', then \(\vectSet{P} \cap \vectSet{P}'=\{\vect{0}\}\) and hence \(\Z(\vectSet{P} \cap \vectSet{P}')=\{0\}\), but on the other hand we have \(\Z(\vectSet{P})=\Z=\Z(\vectSet{P}')\), and therefore \(\Z(\vectSet{P})\cap \Z(\vectSet{P}')=\Z\) is larger. Next we show the closure property.

\begin{restatable}{lemma}{LemmaDirectedHybridlinearNondegenerateIntersection}
Let \(\vectSet{L}\) and \(\vectSet{L}'\) be directed hybridlinear with a non-degenerate intersection. Then \(\vectSet{L} \cap \vectSet{L}'\) is dir. hybridlinear. \label{LemmaDirectedHybridlinearNondegenerateIntersection}
\end{restatable}

\begin{proof}
Write \(\vectSet{L}=\{\vect{b}_1, \dots, \vect{b}_r\}+\N(\vectSet{F})\) and \(\vectSet{L}'=\{\vect{c}_1, \dots, \vect{c}_s\}+\N(\vectSet{F}')\), where these representations fulfill (3) of Lemma \ref{LemmaDirectedHybridlinearEquivalence}.
By Proposition~\ref{PropositionCharacterizeHybridlinear} orders $(\vectSet{L}, \leq_{\N(\vectSet{F})})$ and $(\vectSet{L}', \leq_{\N(\vectSet{F'})})$ are wqos.
By Dickson's lemma the intersection \((\vectSet{L} \cap \vectSet{L}', \leq_{\N(\vectSet{F})} \cap \leq_{\N(\vectSet{F}')})\) is again a wqo. Hence there exist finitely many minimal elements \(\{\vect{d}_1, \dots, \vect{d}_k\}\) of \(\vectSet{L} \cap \vectSet{L}'\). By definition of \(\leq_{\N(\vectSet{F})} \cap \leq_{\N(\vectSet{F}')}=\leq_{\N(\vectSet{F}) \cap \N(\vectSet{F})'}\) we obtain \(\vectSet{L} \cap \vectSet{L}'=\{\vect{d}_1, \dots, \vect{d}_k\}+(\N(\vectSet{F}) \cap \N(\vectSet{F}'))\) is hybridlinear, directedness remains. 

We prove that this is a representation fulfilling (3) of Lemma \ref{LemmaDirectedHybridlinearEquivalence}. Hence let \(1 \leq i,j \leq k\). By Lemma \ref{LemmaFinitelyGeneratedNondegenerateIntersection} it is enough to prove that \(\vect{d}_j - \vect{d}_i \in \Z(\vectSet{F})\) and \(\vect{d}_j - \vect{d}_i \in \Z(\vectSet{F}')\). We show the first, the second follows by symmetry. We have \(\vect{d}_i, \vect{d}_j \in \vectSet{L} \cap \vectSet{L}' \subseteq \vectSet{L}\), hence we can write \(\vect{d}_i=\vect{b}_m+\vect{p}\) and \(\vect{d}_j=\vect{b}_l + \vect{p}'\) for some \(1 \leq m,l \leq r\) and \(\vect{p}, \vect{p}' \in \N(\vectSet{F})\). Since we chose $\vect{b}_i$ satisfying point (3) of Lemma \ref{LemmaDirectedHybridlinearEquivalence}, we have $\vect{b}_l - \vect{b}_m \in \Z(\vectSet{F})$ and therefore:

\(\vect{d}_j-\vect{d}_i=(\vect{b}_l - \vect{b}_m)+(\vect{p}' - \vect{p}) \in \Z(\vectSet{F})+\Z(\vectSet{F})=\Z(\vectSet{F})\).
\end{proof}

The other property we wanted to show is in Lemma \ref{LemmaReduceDimension}. Let us start with a useful notation. Let \(\vectSet{L}=\vectSet{B}+\N(\vectSet{F})\) be a directed hybridlinear set, and \(\vect{x} \in \vectSet{L}\) a point. We write \(\UpwardClosureL{\vect{x}}:=\vect{x}+\N(\vectSet{F})\) for the ``upward-closure'' of the point \(\vect{x}\).

\begin{lemma}
Let \(\vectSet{L}\) be directed hybridlinear, and \(\vect{x} \in \vectSet{L}\). 

Then \(\dim(\vectSet{L} \setminus \UpwardClosureL{\vect{x}})<\dim(\vectSet{L})\). \label{LemmaReduceDimension}
\end{lemma}

\begin{proof}
For the case of \(\vectSet{L}\) linear see Corollary D.3 of \cite{Leroux13}/ remember Lemma \ref{LemmaDimensionDecreaseShiftedL}. We will now prove the lemma by reducing to the linear case.

Write \(\vectSet{L}=\{\vect{b}_1, \dots, \vect{b}_r\}+\N(\vectSet{F})\). Since \(\vectSet{L}\) is directed, there exists a point \(\vect{y} \in \UpwardClosureL{\vect{x}} \cap \bigcap_{i=1}^r \UpwardClosureL{\vect{b}_i}\). In particular \(\UpwardClosureL{\vect{y}} \subseteq \UpwardClosureL{\vect{x}}\), hence we have 
\[\vectSet{L} \setminus (\UpwardClosureL{\vect{x}}) \subseteq \vectSet{L} \setminus (\UpwardClosureL{\vect{y}}) \subseteq \bigcup_{i=1}^r [(\UpwardClosureL{\vect{b}_i}) \setminus (\UpwardClosureL{\vect{y}})]\]
We obtain \(\dim(\UpwardClosureL{\vect{b}_i} \setminus \UpwardClosureL{\vect{y}}) < \dim(\UpwardClosureL{\vect{b}_i})\leq \dim(\vectSet{L})\) by the linear case.
\end{proof}

\subsection{Appendix of Section \ref{SectionHybridization}, Part 2}

%TODO: Do we want to add this? Given relations \(\vectSet{R}_1 \subseteq \N^{n_1} \times \N^{n_2}\) and \(\vectSet{R}_2 \subseteq \N^{n_2} \times \N^{n_3}\), we write \(\vectSet{R}_1 \circ \vectSet{R}_2 :=\{(\vect{v}, \vect{w}) \in \N^{n_1} \times \N^{n_3} \mid \exists \vect{x} \in \N^{n_2}: (\vect{v}, \vect{x}) \in \vectSet{R}_1, (\vect{x}, \vect{w}) \in \vectSet{R}_2\}\) for composition. Given \(\vectSet{R} \subseteq \N^{n} \times \N^{n}\), we write \(\vectSet{R}^{\ast}\) for the reflexive and transitive closure (w.r.t. \(\circ\)).

In this subsection we introduce some heavy VASS theory in order to prove Theorem \ref{TheoremStronglyApproximableEVASS}. The essence of the proof is the following definition.

\begin{definition}
Let \(\vectSet{X} \subseteq \N^n\) be any set. A vector \(\vect{v} \in \N^n\) is a \emph{pump} of \(\vectSet{X}\) if there exists a point \(\vect{x}\) s.t. \(\vect{x}+\N \vect{v} \subseteq \vectSet{X}\).

The set of all pumps of \(\vectSet{X}\) is denoted \(\Pumps(\vectSet{X})\).
%A vector \(\vect{v} \in \N^d\) is a \emph{direction} of \(\vectSet{X}\) if there exists \(n \in \N\) and a point \(\vect{x}\) s.t. \(\vect{x}+\N n \vect{v} \subseteq \vectSet{X}\).
%
%The set of directions of \(\vectSet{X}\) is denoted \(\dir(\vectSet{X})\). 
\end{definition}

I.e.\ a vector \(\vect{v}\) is a pump of \(\vectSet{X}\) if some infinite line with step \(\vect{v}\) is contained in \(\vectSet{X}\). For example the pumps of a linear set \(\N(\vectSet{F})\) are given as follows:

\begin{lemma}\cite[Lemma F.1, special case]{Leroux13}
Let \(\N(\vectSet{F})\) be \(\N\)-f.g.. Then \(\Pumps(\N(\vectSet{F}))=\Q_{\geq 0}(\vectSet{F}) \cap \Z(\vectSet{F})\). \label{LemmaPumpsLinearSet}
\end{lemma}

Observe in particular that the \(\Pumps\) are not only \(\N(\vectSet{F})\), but slightly more: Consider the left of Figure \ref{FigureIntuitionDoublePumps} for an example for the difference. Intuitively, we fill out the ``cone''.

\begin{figure}[h!]
\begin{minipage}{4.5cm}
\begin{tikzpicture}
\begin{axis}[
    axis lines = left,
    xlabel = { },
    ylabel = { },
    xmin=0, xmax=5,
    ymin=0, ymax=15,
    xtick={0,1,2,3,4,5},
    ytick={0,3,6,9,12,15},
    ymajorgrids=true,
    xmajorgrids=true,
    thick,
    smooth,
    no markers,
]

\addplot[
    color=blue,
]
{3 * x};

%\addplot[
%    color=black,
%    very thick,
%]
%coordinates {
%(0.33, 1)(5,1)
%};

\addplot+[
    color=red,
    name path=A,
    domain=2:5, 
]
{3};

\addplot+[
    color=red,
    name path=B,
    domain=2:5, 
]
{3*x-3};

\addplot[
    color=red!40,
]
fill between[of=A and B];

\addplot[
    fill=blue,
    fill opacity=0.5,
    only marks,
    ]
    coordinates {
    (0,0)(1,0)(1,2)(1,3)(2,0)(2,2)(2,3)(2,4)(2,5)(2,6)(3,0)(3,2)(3,3)(3,4)(3,5)(3,6)(3,7)(3,8)(3,9)(4,0)(4,2)(4,3)(4,4)(4,5)(4,6)(4,7)(4,8)(4,9)(4,10)(4,11)(4,12)(5,0)(5,2)(5,3)(5,4)(5,5)(5,6)(5,7)(5,8)(5,9)(5,10)(5,11)(5,12)(5,13)(5,14)(5,15)
    };

\end{axis}
\end{tikzpicture}
\end{minipage}%
\begin{minipage}{4.5cm}
\begin{tikzpicture}
\begin{axis}[
    axis lines = left,
    xlabel = { },
    ylabel = { },
    xmin=0, xmax=4,
    ymin=0, ymax=16,
    xtick={0,1,2,3,4},
    ytick={0,4,8,12,16},
    ymajorgrids=true,
    xmajorgrids=true,
    thick,
    smooth,
    no markers,
]
    
\addplot+[
    name path=E,
    domain=0:4,
    color=blue,
]
{x^2};
    
\addplot+[
    name path=F,
    domain=0:4,
    color=blue,
]
coordinates {(0,0) (4,0)};
    
\addplot[blue!40] fill between[of=E and F];

\end{axis}
\end{tikzpicture}
\end{minipage}%

\caption{\textit{Left}: Consider the following example of an \(\N\)-g. set from \cite{GuttenbergRE23}: \(\N(\vectSet{F})\) for \(\vectSet{F}:=\{(1,0), (1,2), (1,3)\}\). This set fulfills \(\Q_{\geq 0}(\vectSet{F}) \cap \Z(\vectSet{F})=\{(x,y) \mid 0 \leq y \leq 3x\}\), and is almost equal to this, except for the points with \(y=1\): These are \emph{holes}, as they were called in \cite{GuttenbergRE23}. \newline
\textit{Right}: The blue \(\{(x,y) \in \N^2 \mid 0 \leq y \leq x^2\}\) has the set of Pumps \(\{(0,0)\} \cup \{(x,y) \in \N^2 \mid x>0\}\). After the second application we stabilize to \(\N^2\). This happens with any set which has a partner \(\vectSet{L}\) for \(\vectSet{X} \HybridizationRelation \vectSet{L}\): The second application of \(\Pumps\) adds the boundary, stabilizing to \(\vectSet{L}\).}\label{FigureIntuitionDoublePumps}
\end{figure}

\textbf{Essence of }\(\HybridizationRelation\)/Proof of Theorem \ref{TheoremStronglyApproximableEVASS}: The essence is the following: Overapproximate the given VASSnz section \(\vectSet{X}\), and for every obtained m-eVASS \((\VAS, \qin, \qfin)\) compute a semilinear representation of \(\Pumps(\Rel(\VAS,\qin,\qfin))\). The difficult part of this is visible at first glance: Given an arbitrary VASSnz section \(\vectSet{X}\), which can be highly non-semilinear, why would \(\Pumps(\vectSet{X})\) even be semilinear in the first place? The proof is complicated, and was in case of VASS given by Leroux in \cite{Leroux13} and then extended to VASSnz by Guttenberg in \cite{Guttenberg24}. We have to trace their ideas here, which requires basic understanding about the \(\Pumps\) operator.

\textbf{Essentials about the Pumps operator}: We start with the observation that \(\Pumps\) is a monotone operator, and we know a large class of fixed points of \(\Pumps\): Both euclidean closed \(\Q_{\geq 0}\)-g. sets (by Lemma \ref{LemmaFinitelyGeneratedCone} in particular all \(\Q_{\geq 0}\)-\emph{f}.g. sets fulfill this) as well as \(\Z\)-g. sets are fixed points of \(\Pumps\).

\begin{lemma} \label{LemmaObviousPropertiesOfPumps}
Let \(\vectSet{C}\) be \(\Q_{\geq 0}\)-g. and euclidean closed. 

Then \(\Pumps(\vectSet{C})=\vectSet{C}\).

Let \(\Z(\vectSet{F})\) be \(\Z\)-g.. Then \(\Pumps(\Z(\vectSet{F}))=\Z(\vectSet{F})\).

Let \(\vectSet{X} \subseteq \vectSet{X}'\). Then \(\Pumps(\vectSet{X}) \subseteq \Pumps(\vectSet{X}')\). 
\end{lemma}

\begin{proof}
``\(\subseteq\)'': Let \(\vect{v} \in \vectSet{C}\). Then \(\N \vect{v} \subseteq \Q_{\geq 0}(\vect{v}) \subseteq \vectSet{C}\), and we obtain \(\vect{v} \in \Pumps(\vectSet{C})\).

``\(\supseteq\)'': Let \(\vect{v} \in \Pumps(\vectSet{C})\). We have to show \(\vect{v} \in \vectSet{C}\). Since \(\vect{v} \in \Pumps(\vectSet{C})\), we obtain \(\vect{x}+\N \vect{v} \subseteq \vectSet{C}\) for some \(\vect{x}\). We define \(\vect{x}_m:=\frac{1}{m}(\vect{x}+m \vect{v}) \in \vectSet{C}\). Since \(\vectSet{C}\) is euclidean closed, we obtain that the limit \(\vect{v}=\lim_{m \to \infty} \vect{x}_m \in \vectSet{C}\).

``\(\subseteq\)'': Similar to above we have \(\N \vect{v} \subseteq \Z(\vect{v}) \subseteq \Z(\vectSet{F})\).

``\(\supseteq\)'': Let \(\vect{v} \in \Pumps(\Z(\vectSet{F}))\). Then \(\exists \vect{x}\) s.t. \(\vect{x}+\N \vect{v} \subseteq \Z(\vectSet{F})\). Hence \(\vect{v}=(\vect{x}+\vect{v})-\vect{x} \in \Z(\vectSet{F})-\Z(\vectSet{F})=\Z(\vectSet{F})\).

Monotonicity of Pumps is obvious.
\end{proof}

Now we can extend Lemma \ref{LemmaPumpsLinearSet} to directed hybridlinear sets \(\vectSet{L}\) (while Lemma \ref{LemmaPumpsLinearSet} does not hold for \(\{0,1\}+2\N\)!):

\begin{lemma} \label{LemmaPumpsDirectedHybridlinear}
Let \(\vectSet{L}=\vectSet{B}+\N(\vectSet{F})\) be directed hybridlinear. Then \(\Pumps(\vectSet{L})=\Q_{\geq 0}(\vectSet{F}) \cap \Z(\vectSet{F})\).
\end{lemma}

\begin{proof}
``\(\supseteq\)'' follows from Lemma \ref{LemmaPumpsLinearSet}.

``\(\subseteq\)'': Since \(\Q_{\geq 0}(\vectSet{F})\) is a closed cone, we obtain \(\Pumps(\vectSet{L}) \subseteq \Q_{\geq 0}(\vectSet{F})\) by Lemma \ref{LemmaObviousPropertiesOfPumps}. It remains to prove \(\Pumps(\vectSet{L}) \subseteq \Z(\vectSet{F})\). Let \(\vect{v} \in \Pumps(\vectSet{L})\). Then there exists \(\vect{x}\) s.t. \(\vect{x}+\N \vect{v} \subseteq \vectSet{L}\). Let \(\vect{x}_0:=\vect{x}\) and \(\vect{x}_1:=\vect{x}+\vect{v}\). Write \(\vect{x}_1=\vect{b}_1+\vect{w}_1\) and \(\vect{x}_0=\vect{b}_0+\vect{w}_0\) with \(\vect{w}_0, \vect{w}_1 \in \N(\vectSet{F})\) and \(\vect{b}_0, \vect{b}_1 \in \vectSet{B}\). By Lemma \ref{LemmaDirectedHybridlinearEquivalence}(2) we have \(\vect{b}_1-\vect{b}_0 \in \N(\vectSet{F})-\N(\vectSet{F})=\Z(\vectSet{F})\). This implies 
\[\vect{v}=\vect{x}_1-\vect{x}_0=(\vect{b}_1-\vect{b}_0)+(\vect{w}_1-\vect{w}_0) \in \Z(\vectSet{F})+\Z(\vectSet{F})=\Z(\vectSet{F}),\]
finishing the proof.
\end{proof}

The relation of pumps to our problem is that if \(\vectSet{X} \HybridizationRelation \vectSet{L}\) holds for a directed hybridlinear set \(\vectSet{L}\), then we have information not about \(\Pumps(\vectSet{X})\), but about the \emph{second} application of the \(\Pumps\) operator, see also the right of Figure \ref{FigureIntuitionDoublePumps}. From now on we write \(\HybridizationRelationDirected\) to mean \(\HybridizationRelation\) with the added constraint that \(\vectSet{L}\) is directed hybridlinear.

\begin{lemma} \label{LemmaPumpsOfNiceSet}
Let \(\vectSet{X}\) be \(\N\)-g., and \(\vectSet{X} \HybridizationRelationDirected \vectSet{L}=\vectSet{B}+\N(\vectSet{F})\). Then \(\Pumps(\Pumps(\vectSet{X}))=\Q_{\geq 0}(\vectSet{F}) \cap \Z(\vectSet{F})\).
\end{lemma}

\begin{proof}
Write \(\vectSet{L}=\vectSet{B}+\N(\vectSet{F})\). By definition of \(\HybridizationRelation\), we have \(\N_{\geq 1}(\vectSet{F}) \subseteq \Pumps(\vectSet{X})\). Since \(\Pumps\) is monotone, by Lemma \ref{LemmaPumpsDirectedHybridlinear} we obtain \(\Q_{\geq 0}(\vectSet{F}) \cap \Z(\vectSet{F}) = \Pumps(\vectSet{L})=\Pumps(\N_{\geq 1}(\vectSet{F})) \subseteq \Pumps(\Pumps(\vectSet{X}))\). Furthermore, by Lemma \ref{LemmaObviousPropertiesOfPumps}, no matter how often we apply pumps on \(\vectSet{L}\), we will remain at \(\Pumps(\vectSet{L})=\Q_{\geq 0}(\vectSet{F}) \cap \Z(\vectSet{F})\), since these are a closed \(\Q_{\geq 0}\)-g. set and a \(\Z\)-g. set. Since \(\vectSet{X} \subseteq \vectSet{L}\) and pumps is a monotone operator, we obtain that also arbitrarily many applications of pumps on \(\vectSet{X}\) lead to at most \(\Q_{\geq 0}(\vectSet{F}) \cap \Z(\vectSet{F})\). In total we obtain \(\Pumps(\Pumps(\vectSet{X}))=\Q_{\geq 0}(\vectSet{F}) \cap \Z(\vectSet{F})\).
\end{proof}

Applying Lemma \ref{LemmaPumpsOfNiceSet} for an m-eVASS \((\VAS, \qin, \qfin)\), since we have shown \(\vectSet{X} \HybridizationRelation \pi_{\text{src}, \text{tgt}}(\sol(\CharSys(\VAS)))\), we know that \(\Pumps(\Pumps(\Rel(\VAS, \qin, \qfin)))=\Pumps(\pi_{\text{src}, \text{tgt}}(\sol(\CharSys(\VAS))))\) is semilinear, in fact of the form \(\Q_{\geq 0}(\vectSet{F}) \cap \Z(\vectSet{F})\). But as in the right of Figure \ref{FigureIntuitionDoublePumps} one question remains: Which parts of the boundary are in the actual set \(\Pumps(\Rel(\VAS, \qin, \qfin))\)?

For the rest of the appendix fix an m-eVASS \((\VAS, \qin, \qfin)\), and use the same notation for its SCCs, etc. as in Section \ref{SectionProofTheoremEVASS}.

\textbf{Overview}: First we define a class of sets we call \emph{uniform}, and is used in the definition of \(\HybridizationRelation_s\).

Then computing \(\Pumps(\Rel(\VAS, \qin, \qfin))\) a \(2\)-step process. In Step 1 we prove that \(\vectSet{X}_i:=\Rel(\VAS_i,\qini,\qfini) \cap \pi_{\vect{x}_i, \vect{y}_i}(\sol(\CharSys))\) is uniform, and compute \(\Pumps_i:=\Pumps(\vectSet{X}_i)\). The set \(\vectSet{X}_i\) essentially describes the runs of the \(i\)-th SCC which are ``globally relevant'': The intersection ensures that the runs considered have the correct values at source/target as required by the characteristic system, which guarantees that the run can be part of a global solution/run. 

In step 2 we will then prove that the pumps of a composition of \emph{uniform} relations is the composition of the pumps, and therefore obtain \(\Pumps(\Rel(\VAS, \qin, \qfin))\) as
\[\Pumps_1 \circ \Pumps(\vectSet{R}(e_1)) \circ \dots \circ \Pumps(\vectSet{R}(e_{r-1})) \circ \Pumps_r.\]
I.e.\ similar to how the characteristic system simply composes the overapproximations of edges and SCCs, we will also simply compose the pumps.

We now proceed with defining \(\HybridizationRelation_s\).

\subsection{Defining strong overapproximations}

Finally, we define uniform sets. We have to start with three definitions, the first of which is mainly a reminder.

\DefinitionRestatablePreservants*

\begin{definition}
Let \(\vectSet{X} \subseteq \N^d\), and \(\vect{v} \in \N^d\). Then \(\vect{v}\) is a \emph{direction} of \(\vectSet{X}\) if \(\exists n \in \N\) s.t. \(n \vect{v} \in \Pumps(\vectSet{X})\).

The set of all directions is denoted \(\dir(\vectSet{X})\).
\end{definition}

I.e.\ a direction is simply an arbitrarily scaled pump. This has the advantage that for many sets, \(\dir(\vectSet{X})\) is \(\Q_{\geq 0}\)-generated. Namely compare for example with Lemma \ref{LemmaPumpsLinearSet} and Lemma \ref{LemmaPumpsDirectedHybridlinear}: The Pumps were an intersection of two sets, a \(\Z\)-g. set and a \(\Q_{\geq 0}\)-g. set, the \(\Q_{\geq 0}\)-g. set is exactly the set of directions. A similar property will hold for m-eVASS.

The third necessary definition is the following.

\begin{definition}
Let \(\vectSet{X} \subseteq \N^d\). Then \(\vectSet{X}\) is called \emph{well-directed} if for every infinite sequence \(\vect{x}_0, \vect{x}_1, \dots \in \vectSet{X}\) there exists an infinite set \(N \subseteq \N \) of indices (i.e.\ an infinite subsequence) s.t. \(\vect{x}_j+\N(\vect{x}_k-\vect{x}_j) \subseteq \vectSet{X}\) for all \(k>j\) in \(N\).
\end{definition}

Let us explain these definitions, and why preservants make another appearance here. The importance of preservants is best explained using the left of Figure \ref{FigureIntuitionPreservants}. The blue and red sets depicted there are closed under addition and contain \(\vect{0}\), hence they fulfill \(\vectSet{X}=\PX\). However, the union has very little preservants, i.e.\ vectors which can be pumped everywhere: The blue and red sets are simply too different. The operator \(\Pumps\) however would not see this problem.

Other simple observations to gain understanding about preservants is that \(\PX\) is always closed under addition and contains \(\vect{0}\), i.e.\ is \(\N\)-generated, or that the function graph of the sin/cos functions has \((2\pi,0)\) as preservant, but then automatically also all multiples \((2\pi m, 0)\).

\begin{figure}[h!]
\begin{minipage}{4.5cm}
\begin{tikzpicture}
\begin{axis}[
    axis lines = left,
    xlabel = { },
    ylabel = { },
    xmin=0, xmax=8,
    ymin=0, ymax=8,
    xtick={0,2,4,6,8},
    ytick={0,2,4,6,8},
    ymajorgrids=true,
    xmajorgrids=true,
    thick,
    smooth,
    no markers,
]

\addplot+[
    color=red,
    name path=A,
    domain=0:8, 
]
{log2(x+1)};

\addplot+[
    color=red,
    name path=B,
    domain=0:8, 
]
{x};

\addplot[
    color=blue!40,
]
fill between[of=A and B];

\addplot[
    fill=red,
    fill opacity=0.5,
    only marks,
    ]
    coordinates {
    (0,0)(1,0)(2,0)(3,0)(4,0)(5,0)(6,0)(7,0)(8,0)
    };

\end{axis}
\end{tikzpicture}
\end{minipage}%
\begin{minipage}{4.5cm}
\begin{tikzpicture}
\begin{axis}[
    axis lines = left,
    xlabel = { },
    ylabel = { },
    xmin=0, xmax=8,
    ymin=0, ymax=8,
    xtick={0,2,4,6,8},
    ytick={0,2,4,6,8},
    ymajorgrids=true,
    xmajorgrids=true,
    thick,
    smooth,
    no markers,
]

\addplot[
    color=red,
]
coordinates {(2,3) (8,9)};

\addplot[
    color=red,
]
coordinates {(2,2) (8,8)};
    
\addplot+[
    name path=E,
    domain=2:8,
    color=blue,
]
{8};
    
\addplot+[
    name path=F,
    domain=2:8,
    color=blue,
]
coordinates {(2,4) (8,10)};
    
\addplot[blue!40] fill between[of=E and F];

\addplot+[
    name path=G,
    domain=2:8,
    color=blue,
]
coordinates {(2,1) (9,8)};
    
\addplot+[
    name path=H,
    domain=2:8,
    color=blue,
]
coordinates {(2,1) (8,1)};
    
\addplot[blue!40] fill between[of=G and H];

\end{axis}
\end{tikzpicture}
\end{minipage}%

\caption{\textit{Left}: The set \(\vectSet{X}=\vectSet{X}_1 \cup \vectSet{X}_2\) with \(\vectSet{X}_1:=\{(x,y) \mid \log_2(x+1) \leq y \leq x\}\) and \(\vectSet{X}_2:=\{(x,y) \mid y=0\}\) (the blue and red parts respectively) fulfills \(\PX=\{0\}\). Namely a vector \((x,y)\) with \(y>0\) is not a preservant due to the red points, and if \(y=0\) then the blue points are the problem. However, \(\Pumps(\vectSet{X})=\{(x,y) \mid 0 \leq y \leq x\}\) does not see this problem at all. Intuitively, \(\Pumps(\vectSet{X})\) consider pumps which work \emph{somewhere}, and \(\PX\) contains pumps which \emph{work everywhere}. \newline
\textit{Right}: In case of the depicted blue set, we would find a complete extraction. However, clearly we are not reducible yet: We have to do a recursive call on the red lines.}\label{FigureIntuitionPreservants}
\end{figure}

Well-directed is a completely different property: It says that the set is ``well-structured enough'' s.t. between infinitely many points, some pumping will always be possible. A set which is not well-directed is for example \(\{2^x \mid x\in \N\}\), which contains infinitely many points, but no infinite line. However, even \(\N\)-g. sets may not be well-directed, for example \(\vectSet{P}:=\{(x,m) \mid x \text{ has at most }m\text{ bits set in binary}\}\) is \(\N\)-g., but not well-directed. We will not go into details, but due to a well-quasi-order on runs, it is guaranteed that all the sets we will consider are automatically well-directed.

We can now define uniform sets.

\begin{definition} \label{DefinitionUniform}
Let \(\emptyset \neq \vectSet{X} \subseteq \N^d\). Then \(\vectSet{X}\) is called \emph{uniform} if \(\vectSet{X}\) is well-directed and the following hold:

\begin{enumerate}
\item \(\Pumps(\vectSet{X})=\Pumps(\PX)\),
\item \(\dir(\vectSet{X})\) is definable in \(\FO(\Q, \leq, +)\),
\item There exists a finite set \(\vectSet{B}\) s.t. \(\vectSet{B}+\Pumps(\Pumps(\vectSet{X}))\) is directed and \(\vectSet{X} \subseteq \vectSet{B}+\Pumps(\Pumps(\vectSet{X}))\).
\end{enumerate}

In such a case we write \(\vectSet{X} \HybridizationRelation_s (\vectSet{B}+\Pumps(\Pumps(\vectSet{X})))\).
\end{definition}

There are a lot of parts to this definition, but first understand that the main conditions are 1+2). The others like \(\neq \emptyset\), well-directed, etc. are secondary but necessary to ensure the closure property we want. So what does 1) say? Again, the left of Figure \ref{FigureIntuitionPreservants} is a good example: \(\vectSet{X}=\vectSet{X}_1 \cup \vectSet{X}_2\) is highly non-uniform, in fact \(\Pumps(\vectSet{X})\) has dimension \(2\) while \(\Pumps(\PX)=\PX=\{\vect{0}\}\) is small. However, both \(\vectSet{X}_1\) and \(\vectSet{X}_2\) themselves are closed under addition and hence uniform. Uniform is exactly meant to capture that every point can ``pump approximately the same''. But uniform is just enough of an extension of ``closed under addition'', that perfect m-eVASS fall into it (as we will now show).

On the other hand, condition 2) is used to guarantee that \(\Pumps(\vectSet{X})\) is semilinear.

A comment for people familiar with VASS theory: In prior works most abstract statements started with ``Let \(\vectSet{X}=\vect{b}+\vectSet{P}\) for \(\vectSet{P}\) smooth''. One reason for the definition of uniform is that \emph{essentially all} these statements still work if we replace this with ``Let \(\vectSet{X}\) be uniform''. Hence we are now able to take full advantage of theory by Leroux \cite{Leroux11,Leroux12, Leroux13} and Guttenberg \cite{GuttenbergRE23, Guttenberg24}, but can do so \emph{computationally}.

\subsection{Step 1 Towards Strong Approximability}

The answer for how to compute the pumps of an SCC \(\VAS_i\) is a formula first given by Leroux in \cite{Leroux13} and then extended to VASSnz in \cite{Guttenberg24}. 

The necessary \(\Z\)-g. set we will need is easy to read off from the characteristic system \(\CharSys(\VAS)\): The solutions of its homogeneous variant \(\HomCharSys\) form an \(\N\)-f.g. set \(\sol(\HomCharSys)\). We need its \(\Z\)-g. set \(\Z(\sol(\HomCharSys))\), which, as for every \(\N\)-g. set, can be obtained as \(\sol(\HomCharSys)-\sol(\HomCharSys)\). I.e.\ as the set of all differences between homogeneous solutions. To obtain the \(\Z\)-g. set for our SCC, we simply project to the source/target variables for our SCC, i.e.\ we obtain 
\[\text{Grid}:=\pi_{\vect{x}_i, \vect{y}_i}(\sol(\HomCharSys)-\sol(\HomCharSys)).\]

To obtain the directions, we use the crux of the aforementioned formula by Leroux. Let \(\vectSet{C}_i\) be the set of pairs \((\vect{e}, \vect{f})\in \N^{n_i} \times \N^{n_i}\) s.t. \(\vect{f}-\vect{e}\) is the effect of some cycle in \(\VAS_i\). \emph{Beware}: We do \emph{not} require that the cycle uses every edge. Hence this relation \emph{cannot} be captured using a single integer linear program, it requires one per possible support as in the decomposition for \WeakPTwo. Leroux's formula is
\begin{align*}
\dir(\Rel(\VAS_i, \qini, \qfini))=\dir(\vectSet{C}_i)^{\ast}
\end{align*}
where \(\ast\) is transitive closure w.r.t. composition as always.

This is not hard to explain: This removes disconnected solutions of the characteristic system, but not fully: We are allowed to perform any number \(k\in \N\) of \emph{connected} cycles in the SCC in sequence, even if the \(i\)-th cycle is not connected to the \((i+1)\)-st. The idea is also known under the name \emph{sequentially enabled cycles}: Imagine we have cycles \(\rho_1, \dots, \rho_k\) with \(\dirOfRun(\rho_j)=(\vect{v}_{j-1}, \vect{v}_j)\). Then a run \(\rho\) from \(\qini\) to \(\qfini\) which enables the cycles \(\rho_1, \dots, \rho_k\) in this sequence shows \((\vect{v}_0, \vect{v}_k) \in \Pumps(\Rel(\VAS_i, \qini, \qfini))\), even though the cycles \(\rho_j\) were disconnected. To see this, write \(\rho=\rho_0' \rho_1' \dots \rho_k'\), where \(\rho_{j+1}\) is enabled after performing \(\rho_j'\). Then the runs \(\rho_0' \rho_1^{\ast} \rho_1' \dots \rho_k^{\ast} \rho_k'\) prove \(\dirOfRun(\rho)+\N(\vect{v}_0, \vect{v}_k)\subseteq \Rel(\VAS_i, \qini, \qfini)\).

Now we obtain 
\begin{align}
\Pumps_i=\text{Grid} \cap \dir(\vectSet{C}_i)^{\ast}. \tag{*} \label{EquationPumpsFormula}
\end{align}
The one question mark which remains is: How do we compute the transitive closure? Leroux proved the following, where a \(\Q_{\geq 0}\)-g. relation is called \emph{definable} if it is definable in linear arithmetic, i.e.\ in \(\FO(\Q, \leq, +)\):

\begin{theorem}\cite[Theorem VII.1]{Leroux13}
Let \(\vectSet{C}_1, \dots, \vectSet{C}_k \subseteq \Q_{\geq 0}^{n_i} \times \Q_{\geq 0}^{n_i}\) be relations which are reflexive, definable and \(\Q_{\geq 0}\)-g.. Then \((\bigcup_{j=1}^k \vectSet{C}_j)^{\ast}\) is \(\Q_{\geq 0}\)-g. and \emph{again definable}.

A corresponding formula \(\varphi \in \FO(\Q, \leq, +)\) can be computed in polynomial time from formulas \(\varphi_j\) for the \(\vectSet{C}_j\).
\end{theorem}

I.e.\ for \(\Q_{\geq 0}\)-g. relations instead of semilinear relations the transitive closure is easy to compute. We remark that while the theorem states polynomial time, we have to perform quantifier elimination on the obtained formula afterwards, which causes a possibly elementary blowup.

We have to prove that \(\vectSet{X}_i:=\Rel(\VAS_i,\qini,\qfini) \cap \pi_{\vect{x}_i, \vect{y}_i}(\sol(\CharSys))\) is uniform and that the above formula remains correct also in our setting, since our setting is a slight extension of \cite{Leroux13} and \cite{Guttenberg24}.

\begin{lemma} \label{LemmaSingleComponentUniform}
Let \(\vectSet{X}_i:=\Rel(\VAS_i,\qini,\qfini) \cap \pi_{\vect{x}_i, \vect{y}_i}(\sol(\CharSys))\) be an SCC in a perfect m-eVASS for \(\RelationClass=\)VASSnzSec\((k-1)\), the sections of \(k-1\) priority VASSnz. Then \(\vectSet{X}_i\) is uniform with \(\Pumps(\vectSet{X}_i)\) as in formula \eqref{EquationPumpsFormula}.
\end{lemma}

\begin{proof}
Proof by induction on \(k\).

\(k=0\) (normal VASS) are a subcase of the induction step.

\(k \to k+1\): For every edge \(e\in E_i\), the relation \(\vectSet{R}(e)\) is uniform by induction. Let us consider the formula \eqref{EquationPumpsFormula}, and first prove that the pumps are at most the set given in the formula. To see \(\subseteq \text{Grid}\), remember that \(\Pumps\) is monotone and \(\text{Grid}\) is a fixed point, hence it remains to prove \(\vectSet{X}_i \subseteq \text{Grid}\). Since \(\vectSet{X}_i \subseteq \pi_{\vect{x}_i, \vect{y}_i}(\sol(\CharSys))\) by definition, we have \(\Pumps(\vectSet{X}_i) \subseteq \Pumps(\pi_{\vect{x}_i, \vect{y}_i}(\sol(\CharSys)))=\text{Grid}\).

The easiest way to argue containment in \(\dir(\vectSet{C}_i)^{\ast}\) would be by defining an ordering on runs, but we can cite a result of \cite{Guttenberg24} who performed this step for us: Semilinear subreachability relations of VASSnz (in particular the ones of the form \(\vect{x}+\N \vect{v}\) obtained from pumps) are \emph{flattable}, i.e.\ there exist finitely many runs \(\rho_0', \dots, \rho_r'\) and cycles \(\rho_1, \dots, \rho_r\) s.t. all of \(\vect{x}+\N \vect{v}\) can be reached using runs in \(\rho_0' \rho_1^{\ast} \rho_1' \dots \rho_r^{\ast} \rho_r'\) as required.

So the interesting direction is the other containment. Let \(\vect{v} \in \dir(\vectSet{C}_i)^{\ast} \cap \text{Grid}\). Then there exist vectors \(\vect{v}_0, \vect{v}_1, \dots, \vect{v}_k\) s.t. \((\vect{v}_j, \vect{v}_{j+1}) \in \dir(\vectSet{C}_i)\) for all \(0 \leq j \leq k-1\). Being \(\in \dir(\vectSet{C}_i)\) means there exist \(m_j \in \N\) (we use their common denominator to use a common \(m\)) s.t. \(m(\vect{v}_{j+1}- \vect{v}_j)\) is the effect of some cycle \(\rho_j\). At this point we can explain the main two difficulties of this proof: 1): What does it even mean to prove that \(\vect{v}\) is a pump of \(\vectSet{P}_{\vectSet{X}_i}\), and not of \(\vectSet{X}_i\) itself? 2) How to remove the factor \(m\)? The reader might think of the \(j_1, j_2\) trick from proving correctness, but it does not work here: That trick requires a full support solution. For a counter example, simply consider a 1-VASS with two states \(q_1, q_2\). On state \(q_1\) we have a loop with \(x+=2\), which is the direction we want to pump. On state \(q_2\) there is a loop with \(x-=1\). Clearly, in order to find a cycle with effect \(+=1\) it has to be allowed to move to state \(q_2\), i.e.\ we need a full support solution to pump \(+=1\), otherwise we can indeed only pump \(+=2\).

Let us explain the solutions to these questions. Regarding 1), in order to prove membership in \(\vectSet{P}_{\vectSet{X}_i}\), we have to fix an amount of homogeneous solutions only depending on \(\vect{v}\) s.t. no matter what element of \(\vectSet{X}_i\), i.e.\ what \emph{actual} run \(\rho\), we start from, we can add the fixed amount of homogeneous solutions and ensure that this only slightly enlarged run \(\rho'\) can pump \(\vect{v}\). Regarding 2), there is only one logical solution: We need to write \(\N=\{0, \dots, m-1\}+m\N\) and prove that not just \(\rho'\) itself, but also \(m-1\) neighbouring runs \(\rho_j\) with \(\dirOfRun(\rho_j)=\dirOfRun(\rho')+j \vect{v}\) for \(0 \leq j \leq m-1\) can pump \(m\vect{v}\). Then the whole line is reachable, even though we did not find a corresponding single cycle.

We inevitably have to build the solution to 2) first. Let \(\vect{w}_0, \dots, \vect{w}_{m-1} \in \vectSet{P}_{\vectSet{X}_i}\) with \(\vect{w}_{j+1}=\vect{w}_j+\vect{v}\) for all \(0 \leq j \leq m-1\). These clearly exist since ``almost all interior vectors'' are in \(\vectSet{P}_{\vectSet{X}_i}\) by correctness and the \(j_1, j_2\) trick.

Now we choose the fixed amount of homogeneous solutions to add: Fix an amount \(k\) s.t. \(k\) homogeneous solutions alone are enough to be able to embed enough up- and down-pumping sequences to perform a cycle \(\rho_{\text{enable}}\) on \(\qini\) which sequentially enables not only the \(\rho_j\) but also the loops corresponding to \(\vect{w}_0, \dots, \vect{w}_{m-1}\).

Now we have to verify our choice. Let \(\vect{x} \in \vectSet{X}_i\), i.e.\ let \(\rho\) be a run with \(\dirOfRun(\rho)=\vect{x}\) and which we know is executable. We use the run \(\rho':=\mathbf{up}^k  \rho_{\text{enable}} \mathbf{diff}^k \rho \mathbf{dwn}^k \). We argue that the run is enabled and can pump the required vectors: The run \(\mathbf{up}^k \rho_{\text{enable}} \mathbf{diff}^k\) is enabled by choice of \(k\). Also by choice of \(k\) and \(\rho_{\text{enable}}\), this only consumes the extra values we added using the homogeneous solutions. Hence the run \(\rho\) which was originally enabled is still enabled. Finally it is easy to see that also \(\mathbf{dwn}^k\) is enabled, since it was chosen as a down-pumping sequence.

To see that the new run can enable all the \(\rho_j\) in sequence, remember that this was specifically how \(\rho_{\text{enable}}\) was chosen. To argue the factor \(m\) away, remember that \(\rho_{\text{enable}}\) also enables the \(\vect{w}_j\). Hence we have \((\dirOfRun(\rho')+\vect{w}_0)+\N \vect{v} \subseteq \vectSet{X}_i\). In particular not \(\rho'\) itself can pump \(\N \vect{v}\), but \(\rho'+\vect{w}_0\). Since \(\rho\) was chosen arbitrarily, we not only found a pump in \(\vectSet{X}_i\), in fact we have shown that \(k_{\text{hom}}+\vect{w}_0+\N \vect{v} \subseteq \vectSet{P}_{\vectSet{X}_i}\), where \(k_{\text{hom}}\) stands for the constant number of homogeneous solutions we had to add.
\end{proof}

\subsection{Step 2 Towards Strong Approximability}

In this section we prove that uniform sets are stable under non-degenerate composition. To this end, we first have to repeat the most important property connected to non-degenerate intersection of \(\N\)-g. sets: Fusing lines. We do this over the course of two lemmas.

\begin{lemma}
Let \(\vectSet{P}_1 \subseteq \vectSet{P}_2\) be \(\N\)-g. with \(\dim(\vectSet{P}_1)=\dim(\vectSet{P}_2)\). Then they generate the same vector space \(\vectSet{V}\), and for all \(\vect{p}_2 \in \vectSet{P}_2\) there exists \(\vect{p}_2' \in \vectSet{P}_2\) with \(\vect{p}_2+\vect{p}_2' \in \vectSet{P}_1\). \label{LemmaNondegenerateInclusion}
\end{lemma}

\begin{proof}
Let \(\vectSet{V}_1\) be the vector space generated by \(\vectSet{P}_1\), and \(\vectSet{V}_2\) the vector space generated by \(\vectSet{P}_2\). Since \(\vectSet{P}_1 \subseteq \vectSet{P}_2\) we have \(\vectSet{V}_1 \subseteq \vectSet{V}_2\). On the other hand remember Lemma \ref{LemmaFromJerome}: We obtain \(\dim(\vectSet{V}_1)=\dim(\vectSet{P}_1)= \dim(\vectSet{P}_2)=\dim(\vectSet{V}_2)\). Together we obtain \(\vectSet{V}_1=\vectSet{V}_2\).

Now let \(\vect{p}_2\in \vectSet{P}_2\). Write \(\vect{p}_2 \in \vectSet{V}_2=\vectSet{V}_1\) as linear combination of \(\vectSet{P}_1\). Simply choose \(\vect{p}_2''\) such that all negative components of the linear combination giving \(\vect{p}_2\) become positive. Now there exists \(m \in \N\) s.t. \(m(\vect{p}_2+\vect{p}_2'') \in \vectSet{P}_1\). Choose \(\vect{p}_2':=(m-1) \vect{p}_2 + m \vect{p}_2' \in \N(\vectSet{P}_2)=\vectSet{P}_2\).
\end{proof}

\begin{lemma}
Let \(\vectSet{P}_1, \vectSet{P}_2\) be \(\N\)-g. with a non-degenerate intersection. Then for all \(\vect{p}_1 \in \vectSet{P}_1\) and \(\vect{p}_2 \in \vectSet{P}_2\) there exist \(\vect{p}_1' \in \vectSet{P}_1\) and \(\vect{p}_2' \in \vectSet{P}_2\) such that \(\vect{p}_1+\vect{p}_1'=\vect{p}_2+\vect{p}_2'\). \label{LemmaNonDegenerateCombineVectors}
\end{lemma}

\begin{proof}
Since \(\vectSet{P}_1 \cap \vectSet{P}_2\) is non-degenerate, we can apply Lemma \ref{LemmaNondegenerateInclusion} twice, namely with \(\vectSet{P}_1 \cap \vectSet{P}_2 \subseteq \vectSet{P}_i\) for \(i\in \{1,2\}\). The vector spaces generated by \(\vectSet{P}_1\) and \(\vectSet{P}_2\) are hence equal, and there are \(\vect{p}_1'', \vect{p}_2''\) with \(\vect{p}_1+\vect{p}_1'', \vect{p}_2+\vect{p}_2'' \in \vectSet{P}_1 \cap \vectSet{P}_2\). Choose \(\vect{p}_1':=\vect{p}_1''+(\vect{p}_2+\vect{p}_2'') \in \vectSet{P}_1+\vectSet{P}_1 \subseteq \vectSet{P}_1\), and similarly \(\vect{p}_2':=\vect{p}_2''+(\vect{p}_1+\vect{p}_1'') \in \vectSet{P}_2\). Clearly \(\vect{p}_1+\vect{p}_1'=\vect{p}_2+\vect{p}_2'\).
\end{proof}

Lemma \ref{LemmaNonDegenerateCombineVectors} might seem inconspicuous, but is the basis of computing with pumps in case of non-degenerate intersection/composition. Namely Lemma \ref{LemmaNonDegenerateCombineVectors} implies many results like the following, always using exactly the same technique we call \emph{fusing lines}:

\begin{lemma}
Let \(\vectSet{P}_1, \vectSet{P}_2\) be \(\N\)-g., \(\vectSet{P}_1 \cap \vectSet{P}_2\) non-degenerate. Then \(\Pumps(\vectSet{P}_1)\cap \Pumps(\vectSet{P}_2) = \Pumps(\vectSet{P}_1 \cap \vectSet{P}_2)\).
\end{lemma}

\begin{proof}
``\(\supseteq\)'' is trivial. Hence let \(\vect{v} \in \Pumps(\vectSet{P}_1)\cap \Pumps(\vectSet{P}_2)\). We have to prove that \(\vect{v} \in \Pumps(\vectSet{P}_1 \cap \vectSet{P}_2)\).

By definition of \(\Pumps\), there exist \(\vect{x}_1, \vect{x}_2\) s.t. \(\vect{x}_1+\N \vect{v} \subseteq \vectSet{P}_1\) and \(\vect{x}_2+\N \vect{v} \subseteq \vectSet{P}_2\). By Lemma \ref{LemmaNonDegenerateCombineVectors}, there exist \(\vect{p}_1, \vect{p}_2\) s.t. \(\vect{x}:=\vect{x}_1+\vect{p}_1=\vect{x}_2+\vect{p}_2 \in \vectSet{P}_1 \cap \vectSet{P}_2\). We claim \(\vect{x}+\N \vect{v} \subseteq \vectSet{P}_1 \cap \vectSet{P}_2\). We prove the two containments separately. To see containment in \(\vectSet{P}_1\), for all \(m\in \N\) we have \[\vect{x}+m \vect{v}=\vect{x}_1+\vect{p}_1+m\vect{v}=(\vect{x}_1+m \vect{v})+\vect{p}_1 \in \vectSet{P}_1+\vectSet{P}_1 \subseteq \vectSet{P}_1\]
as claimed.
\end{proof}

What happened is simple: When ``fusing'' \(\vect{x}_1\) and \(\vect{x}_2\) into a common point \(\vect{x}\), also the ``lines above them'' in direction \(\vect{v}\) were fused. 

One might wonder how to prove that arbitrary \(\N\)-g. sets have a non-degenerate intersection. Regarding this problem we have the following:

\begin{proposition}\cite[Prop. 3.9]{GuttenbergRE23} \label{PropositionCheckNonDegenerate}
Let \(\vectSet{P}_1, \vectSet{P}_2\) be \(\N\)-g. s.t. \(\Pumps(\Pumps(\vectSet{P}_i))\) are semilinear for \(i \in \{1,2\}\) and \(\Pumps(\Pumps(\vectSet{P}_1)) \cap \Pumps(\Pumps(\vectSet{P}_2))\) is non-degenerate. Then \(\vectSet{P} \cap \vectSet{P}'\) is non-degenerate.
\end{proposition}

I.e.\ non-degenerate intersection transfers from the overapproximation to the actual sets. The simplest proof of Proposition \ref{PropositionCheckNonDegenerate} using theory from this paper is to use that \(\vectSet{P} \HybridizationRelation \Pumps(\Pumps(\vectSet{P}))\) and use utbo overapproximation, i.e.\ Lemma \ref{LemmaUTBOApproximation}. 

We now start the main proof of this section.

\begin{lemma}
Let \(\vectSet{X}_{12}\HybridizationRelation \vectSet{L}_{12}\) and \(\vectSet{X}_{23}\HybridizationRelation \vectSet{L}_{23}\) be s.t. \(\vectSet{X}_{12}, \vectSet{X}_{23}\) are uniform and \(\piin(\vectSet{L}_{23}) \cap \piout(\vectSet{L}_{12})\) is non-degenerate. Then \(\vectSet{X}_{12} \circ \vectSet{X}_{23} \HybridizationRelationDirected \vectSet{L}_{12} \circ \vectSet{L}_{23}\) is uniform with the set of pumps 

\(\Pumps(\vectSet{X}_{12} \circ \vectSet{X}_{23})=\Pumps(\vectSet{X}_{12}) \circ \Pumps(\vectSet{X}_{23})\).
\end{lemma}

\begin{proof}
Observe first that by Proposition \ref{PropositionCheckNonDegenerate} not just \(\piin(\vectSet{L}_{23})=\piout(\vectSet{L}_{12})\) is non-degenerate, but also \(\piin(\vectSet{P}_{\vectSet{X}_{23}}) \cap \piout(\vectSet{P}_{\vectSet{X}_{12}})\) is non-degenerate, hence we can use line fusion.

First we prove that \(\vectSet{X}_{12} \circ \vectSet{X}_{23} \neq \emptyset\). We know that \(\vectSet{X}_{12}\) and \(\vectSet{X}_{23}\) are non-empty since they are uniform. Hence let \((\vect{x}_1, \vect{x}_2) \in \vectSet{X}_{12}\) and \((\vect{x}_2', \vect{x}_3) \in \vectSet{X}_{23}\). By Lemma \ref{LemmaNonDegenerateCombineVectors}, we can fuse the points \(\vect{x}_2\) and \(\vect{x}_2'\), in the sense that there exists \(\vect{p}_2 \in \piout(\vectSet{P}_{\vectSet{X}_{12}})\) and \(\vect{p}_2' \in \piin(\vectSet{P}_{\vectSet{X}_{23}})\) s.t. \(\vect{x}_2+\vect{p}_2=\vect{x}_2'+\vect{p}_2'\). By definition of projection, there exist \(\vect{p}_1, \vect{p}_3\) s.t. \((\vect{p}_1, \vect{p}_2) \in \vectSet{P}_{\vectSet{X}_{12}}\) and \((\vect{p}_2', \vect{p}_3) \in \vectSet{P}_{\vectSet{X}_{23}}\). By definition of preservants, we have \((\vect{x}_1+\vect{p}_1, \vect{x}_2+\vect{p}_2) \in \vectSet{X}_{12}\) and \((\vect{x}_2'+\vect{p}_2', \vect{x}_3+\vect{p}_3) \in \vectSet{X}_{23}\). Therefore \((\vect{x}_1+\vect{p}_1, \vect{x}_3+\vect{p}_3) \in \vectSet{X}_{12} \circ \vectSet{X}_{23}\), proving non-emptiness. A similar line fusion will happen many more times in this proof.

As a composition of well-directed relations, \(\vectSet{X}_{12} \circ \vectSet{X}_{23}\) is again well-directed. To prove this, simply pick \(N_2 \subseteq N_1 \subseteq \N\), i.e.\ a subsubsequence.

Next we check the properties of Definition \ref{DefinitionUniform}. 

1) We will prove that \(\Pumps(\vectSet{X}_{12} \circ \vectSet{X}_{23}) \subseteq \Pumps(\vectSet{X}_{12}) \circ \Pumps(\vectSet{X}_{23}) = \Pumps(\vectSet{P}_{\vectSet{X}_{12}}) \circ \Pumps(\vectSet{P}_{\vectSet{X}_{23}}) \subseteq \Pumps(\vectSet{P}_{\vectSet{X}_{12} \circ \vectSet{X}_{23}}) \subseteq \Pumps(\vectSet{X}_{12} \circ \vectSet{X}_{23})\). Observe first that the \(=\) in the middle immediately follows since \(\vectSet{X}_{12}, \vectSet{X}_{23}\) are uniform, and the last \(\subseteq\) is trivial since the preservants always have less pumps. It remains to show the other containments.

\(\Pumps(\vectSet{P}_{\vectSet{X}_{12}}) \circ \Pumps(\vectSet{P}_{\vectSet{X}_{23}}) \subseteq \Pumps(\vectSet{P}_{\vectSet{X}_{12} \circ \vectSet{X}_{23}})\): Let \((\vect{v}_1, \vect{v}_2) \in \Pumps(\vectSet{P}_{\vectSet{X}_{12}})\) and \((\vect{v}_2, \vect{v}_3) \in \Pumps(\vectSet{P}_{\vectSet{X}_{23}})\). Then there exists \((\vect{x}_1, \vect{x}_2)\) and \((\vect{x}_2', \vect{x}_3)\) s.t. \((\vect{x}_1, \vect{x}_2)+\N(\vect{v}_1, \vect{v}_2) \subseteq \vectSet{P}_{\vectSet{X}_{12}}\) and \((\vect{x}_2',\vect{x}_3)+\N (\vect{v}_2, \vect{v}_3) \subseteq \vectSet{P}_{\vectSet{X}_{23}}\). Since \(\piout(\vectSet{P}_{\vectSet{X}_{12}}) \cap \piin(\vectSet{P}_{\vectSet{X}_{23}})\) is non-degenerate, we fuse the lines using Lemma \ref{LemmaNonDegenerateCombineVectors} and obtain \((\vect{v}_1, \vect{v}_3) \in \Pumps(\vectSet{P}_{\vectSet{X}_{12} \circ \vectSet{X}_{23}})\) as claimed.

\(\Pumps(\vectSet{X}_{12} \circ \vectSet{X}_{23}) \subseteq \Pumps(\vectSet{X}_{12}) \circ \Pumps(\vectSet{X}_{23})\): Let \(\vect{v}=(\vect{v}_1, \vect{v}_3) \in \Pumps(\vectSet{X}_{12} \circ \vectSet{X}_{23})\). Then there exists \(\vect{x}=(\vect{x}_1, \vect{x}_3)\) s.t. \((\vect{x}_1, \vect{x}_3)+\N (\vect{v}_1, \vect{v}_3) \subseteq \vectSet{X}_{12} \circ \vectSet{X}_{23}\). By definition of \(\circ\), there exist \(\vect{x}_{2,m}\) s.t. \((\vect{x}_1+m \vect{v}_1, \vect{x}_{2,m})\in \vectSet{X}_{12}\) for all \(m\) and \((\vect{x}_{2,m}, \vect{x}_3+m \vect{v}_3) \in \vectSet{X}_{23}\).

Since \(\vectSet{X}_{12}\) is well-directed, there exists a subset \(N_1 \subseteq \N\) of indices s.t. \( ((k-j) \vect{v}_1, \vect{x}_{2,k}-\vect{x}_{2,j}) \in \dir(\vectSet{X}_{12})\) for all \(k>j\) in \(N_1\). Since \(\vectSet{X}_{23}\) is well-directed, there exists a subset \(N_2 \subseteq N_1\) of indices s.t. additionally \((\vect{x}_{2,k}-\vect{x}_{2,j}, (k-j) \vect{v}_3) \in \dir(\vectSet{X}_{23})\) for all \(k>j\) in \(N_2\). Hence \((k-j) \vect{v} \in \dir(\vectSet{X}_{12}) \circ \dir(\vectSet{X}_{23})\). We will not describe the way to instead obtain \(\Pumps(\vectSet{X}_{12}) \circ \Pumps(\vectSet{X}_{23})\) in detail, but the idea is simple: The reason we only obtained directions is due to the same problem we had in the proof of Lemma \ref{LemmaSingleComponentUniform}: Along the boundary, a single run can only pump multiples of the direction. But by cleverly choosing interior directions to add, we generate both in \(\vect{X}_{12}\) and \(\vectSet{X}_{23}\) a line without the factor \((k-j)\). But now the lines are in completely different locations. Hence we have to use the line fusion idea Lemma \ref{LemmaNonDegenerateCombineVectors} to finish the proof.

2) Obvious. A composition of definable cones is definable.

3) Let us first prove that \(\vectSet{L}_{12} \circ \vectSet{L}_{23}\) is directed. Let \((\vect{x}_1, \vect{x}_3), (\vect{y}_1, \vect{y}_3) \in \vectSet{L}_{12} \circ \vectSet{L}_{23}\). Then there exist \(\vect{x}_2, \vect{y}_2\) s.t. \((\vect{x}_1, \vect{x}_2), (\vect{y}_1, \vect{y}_2) \in \vectSet{L}_{12}\) and \((\vect{x}_2, \vect{x}_3), (\vect{y}_2, \vect{y}_3) \in \vectSet{L}_{23}\). Since \(\vectSet{L}_{12}\) is directed, there exists \((\vect{z}_1, \vect{z}_2) \in \vectSet{L}_{12}\) above both \((\vect{x}_1, \vect{x}_2)\) and \((\vect{y}_1, \vect{y}_2)\). Similarly, since \(\vectSet{L}_{23}\) is directed, there exists \((\vect{z}_2', \vect{z}_3) \in \vectSet{L}_{23}\) above both \((\vect{x}_2, \vect{x}_3)\) and \((\vect{y}_2, \vect{y}_3)\). Since \(\piin(\vectSet{L}_{23}) \cap \piout(\vectSet{L}_{12})\) is non-degenerate, we can combine the points \(\vect{z}_2\) and \(\vect{z}_2'\) to a common point \(\vect{z}_{2,\text{new}}\).

We use the base points of \(\vectSet{L}_{12} \circ \vectSet{L}_{23}\) as \(\vectSet{B}\).

To prove \(\vectSet{X} \subseteq \vectSet{B}+\Pumps(\Pumps(\vectSet{X}))\), we claim \(\Pumps(\Pumps(\vectSet{X}_{12} \circ \vectSet{X}_{23}))=\Pumps(\vectSet{L}_{12} \circ \vectSet{L}_{23})\). 3) then follows by choice of \(\vectSet{B}\).

Proof of claim: By monotonicity of \(\Pumps\) we have \(\Pumps(\Pumps(\vectSet{X}_{12} \circ \vectSet{X}_{23}))=\Pumps(\Pumps(\vectSet{X}_{12}) \circ \Pumps(\vectSet{X}_{23}))\) is lower bounded by \(\Pumps(\interior(\vectSet{L}_{12}) \circ \interior(\vectSet{L}_{23}))\) and upper bounded by \(\Pumps(\Pumps(\vectSet{L}_{12}) \circ \Pumps(\vectSet{L}_{23}))\). These are both equal to \(\Pumps(\vectSet{L}_{12} \circ \vectSet{L}_{23})\) by observing that adding or removing along the boundary does not influence the pumps of a directed hybridlinear set.
\end{proof}

\subsection{Step 3 Towards Primitive 2}

At this point we have shown that for every perfect m-eVASS \(\VAS\), one can in elementary time compute a semilinear formula for \(\Pumps(\vectSet{X})\), where \(\vectSet{X}:=\Rel(\VAS, \qin, \qfin)\). Namely, we use \eqref{EquationPumpsFormula} for every strongly-connected component \(\VAS_i\), and perform composition. It remains to check Primitive 2 using this fact. As already mentioned, our main contribution here was the introduction of the definition of \emph{uniform} and thereby making the theory of \cite{Leroux13} and \cite{GuttenbergRE23} computable. At this point we can simply rely on their semilinearity algorithm, while substituting our procedure at certain locations. We only summarize the idea of the algorithm quickly here. 

Since we will now work across perfect m-eVASSs, as opposed to before where \(\VAS_i\) used to refer to an SCC of \(\VAS\), we now use different indices \(\VAS_j\) to refer to different perfect m-eVASS, and write \(\vectSet{Y}_j:=\Rel(\VAS_j, \qin, \qfin)\) (\(\vectSet{Y}\) instead of \(\vectSet{X}\) to again clarify it is not just an SCC, but a different VASS).

The algorithm for Primitive 2 is as follows: Some of the notions we will explain afterwards, but we prefer to give a full description of the algorithm first.

Step 1: Apply the approximation algorithm to obtain a finite set \(\{\vectSet{Y}_1, \dots, \vectSet{Y}_k\}\) of perfect m-eVASS.

Step 2: Compute semilinear sets \(\vectSet{S}_j=\vectSet{C}_j \cap \vectSet{G}_j=\Pumps(\vectSet{Y}_j)\) for all \(j \in \{1,\dots, k\}\), where \(\vectSet{C}_j\) is \(\Q_{\geq 0}\)-g. and \(\vectSet{G}_j\) is \(\Z\)-g. (we use the letter \(\vectSet{G}\) because \(\Z\)-g. sets are also called \emph{grids}). Discard all m-eVASS where \(\dim(\vectSet{S}_j)< \dim(\vectSet{L})\).

Step 3: Let \(\{\vect{b}_1, \dots, \vect{b}_m\}\) be a representative system for the cosets of \(\vectSet{G}:=\bigcap_{j=1}^k \vectSet{G}_j\) inside \(\vectSet{L}-\vectSet{L}\). For every \(\text{iter}=1, \dots, m\) do: Define \(I_{\text{iter}}:=\{j \in \{1,\dots, k\} \mid \vectSet{Y}_j \cap (\vect{b}_j+ \vectSet{G}) \neq \emptyset\}\) as the set of m-eVASS which intersect this coset. Now for the finite set \(\{\vectSet{C}_j \mid j \in I_{\text{iter}}\}\) of \(\Q_{\geq 0}\)-g. relations check for existence of a so-called \emph{complete extraction} (described in a moment). If it does not exist, reject. Otherwise we are in the situation in the right of Figure \ref{FigureIntuitionPreservants}: We compute for the corresponding complete extraction \(\{\vectSet{K}_j \mid j \in I_{\text{iter}}\}\) base points \(\vect{x}_{\text{iter},j}\) such that \(\vect{x}_{\text{iter}, j}+(\vectSet{K}_j \cap \vectSet{G}) \subseteq \vectSet{Y}_j\). We then do a recursive call on \(\vectSet{L} \setminus \bigcup_{\text{iter}, j} \vect{x}_{\text{iter},j}+(\vectSet{K}_j \cap \vectSet{G})\), which corresponds to the red lines in the right of Figure \ref{FigureIntuitionPreservants}.

There is a lot to unpack and understand here. First of all, let us start with possible questions regarding the \(\vectSet{G}_j\) and \(\vectSet{G}\). A \(\Z\)-g. set carries modulo information, i.e.\ is defined by some conjunction of formulas \(\vect{a} \vect{x} \equiv b \mod m\), where \(b,m \in \N\) and \(\vect{a} \in \Z^n\). When we form the intersection of the \(\vectSet{G}_j\), we hence deal with the fact that some set might only be able to reach even numbered points, and another set only ones divisible by \(3\) etc. This also explains the ``representative system for the cosets'': Simply use one representative for every possible remainder \(r\) \(\mod \vectSet{G}\). Then the set \(I_{\text{iter}}\) contains the ``active'' sets in an SCC: If a set is only active on even numbers, then for the coset of odd numbers this set will not exist. 

Finally, we can define complete extraction (see also \cite{Leroux13}, Appendix E): 

\begin{definition}
Let \(\mathcal{K}=\{\vectSet{C}_1, \dots, \vectSet{C}_k\}\) be a finite set of \(\FO(\Q,\leq,+)\) definable \(\Q_{\geq 0}\)-g. sets. A \emph{complete extraction} is a finite set \(\{\vectSet{K}_1, \dots, \vectSet{K}_k\}\) of \(\Q_{\geq 0}\)-\emph{finitely} generated sets s.t. \(\vectSet{K}_j \subseteq \vectSet{C}_j\) for all \(1 \leq j \leq k\) and \(\bigcup_{j=1}^k \vectSet{K}_j =\bigcup_{j=1}^k \vectSet{C}_j\).
\end{definition}

The idea is best explained using the left of Figure \ref{FigureIntuitionPreservants}: It consists of two uniform sets, whose union of pumps is \(\{(x,y) \mid 0 \leq y \leq x\}\). So when we only consider the union, the algorithm would not detect the ``hole'' between the two uniform sets. However, complete extraction catches this: The set of cones \(\vectSet{C}_1:=\{(x,y) \mid y=0\}\) and \(\vectSet{C}_2:=\{(x,y) \mid 0 < y \leq x\}\) does not have a complete extraction, since any finite set of vectors from \(\vectSet{C}_2\) (remember \(\vectSet{K}_2\) has to be finitely generated) would have a minimal steepness \(>0\), and hence some angle would remain uncovered. 

In fact, complete extraction perfectly characterizes whether the union of multiple uniform sets (which agree on the modulus, hence we had to consider the cosets of \(\vectSet{G}\)) cover ``most of the space'': If there is no complete extraction, then a full dimensional hole will remain, otherwise we will be at least in the situation depicted on the right of Figure \ref{FigureIntuitionPreservants}: The complete extractions give rise to linear sets contained in the respective uniform sets. The starting points might be dislodged, such that some lines remain in the middle, but these can now be taken care of by recursion.

Returning to the actual algorithm, what we just explained is the reason for considering every coset of \(\vectSet{G}\) separately, and checking for complete extractions: From it we will be able to determine linear sets contained in the respective m-eVASS which together cover most of the space. Finally, in the algorithm we indirectly state what kind of linear set will be contained in the uniform set when finding a complete extraction: Some base point \(\vect{x}_{\text{iter},j}+(\vectSet{K}_j \cap \vectSet{G})\). And by complete extraction, they will cover most of the coset.

Since \(\vectSet{K}_j\) is a \(\Q_{\geq 0}\)-finitely generated set, \(\vectSet{K}_j \cap \vectSet{G}\) is \(\N\)-finitely generated, and our algorithm computing \(\vect{x}_{\text{iter},j}\) simply applies the argument of Lemma \ref{LemmaSingleComponentUniform} to find a run pumping the finitely many directions \(\vect{v}_1, \dots, \vect{v}_k\) which have to be pumped.

\end{document}